%% file: main.tex
\documentclass[11pt]{article}

\usepackage[english]{babel}
\usepackage{a4wide}
\usepackage[latin1]{inputenc}
\usepackage[T1]{fontenc}
\usepackage{color}
\usepackage{eurosym}
\usepackage[top=4cm,bottom=4cm,left=3cm,right=3cm]{geometry}
\usepackage{fancyhdr}
\usepackage{float}
\usepackage{amssymb,amsmath, amsthm}
\usepackage{url}

\usepackage{booktabs}
\usepackage{stmaryrd} 
\usepackage{mathenv}
\usepackage{dsfont}
\usepackage{amssymb}
\usepackage{latexsym}
\usepackage{amsmath}
\usepackage{color}
\usepackage{comment}
\usepackage{amsthm}% for \proof,\newtheorem
\usepackage{bm}
\usepackage{bbm} %\mathbbm{ABCabc\nabla}
\usepackage[bbgreekl]{mathbbol} %\mathbb{abcABC\nabla\bbalpha\bbbeta\Delta  }
\usepackage{authblk}
\usepackage{multirow}

\usepackage[english]{babel}
\usepackage{a4wide}
\usepackage[latin1]{inputenc}
\usepackage[T1]{fontenc}
\usepackage{color}
\usepackage{eurosym}
\usepackage[top=4cm,bottom=4cm,left=3cm,right=3cm]{geometry}
\usepackage{fancyhdr}
\usepackage{float}
\usepackage{amssymb,amsmath, amsthm}

\usepackage{mathenv}
\usepackage{dsfont}
\usepackage{amssymb}
\usepackage{latexsym}
\usepackage{amsmath}
\usepackage{color}
\usepackage{comment}
\usepackage{amsthm}% for \proof,\newtheorem
\usepackage{graphicx}% get graphs [dvips]
\newif\ifrs
\rstrue
\ifrs \usepackage{mathrsfs} \fi  % Use \mathscr{*}
\newif\ifcol
\coltrue %%% Activate "notice"
\newtheorem{theorem*}{Theorem}[section]

\newtheorem{note*}[theorem*]{Note}
\newtheorem{lemma*}[theorem*]{Lemma}
\newtheorem{definition*}[theorem*]{Definition}
\newtheorem{proposition*}[theorem*]{Proposition}
\newtheorem{corollary*}[theorem*]{Corollary}
\newtheorem{remark*}[theorem*]{Remark}
\newtheorem{example*}[theorem*]{Example}
\newtheorem*{conditionE'3}{Condition (E3')}

\numberwithin{equation}{section}
\newif\ifcol
\coltrue
\ifcol
\newcommand{\colorr}{\color[rgb]{0.8,0,0}}

\newcommand{\colorn}{\color[rgb]{1,1,1}}

\else
\newcommand{\colorr}{\color{black}}% {{\color[rgb]{0.8,0,0}}
\newcommand{\colorn}{\color{black}}% {\color[rgb]{1,1,1}}
\fi
\excludecomment{en-text}
\includecomment{jp-text}
\includecomment{comment}
\input nakamacro27.tex

\newcommand{\var}{\operatorname{Var}}

\newcommand{\espc}{\esp_{\theta_0,i,n}}
\newcommand{\probac}{\proba_{\theta_0,i,n}}

\newcommand{\esps}{\esp_{\theta_0,n}}
\newcommand{\probas}{\proba_{\theta_0,n}}

\DeclareMathOperator{\ind}{\perp \!\!\! \perp}

\newcommand{\blockd}{(i-1)\Delta_n}
\newcommand{\blockf}{i\Delta_n}

\newtheorem*{conditionCstar}{[$\textnormal{C}^*$]}

\begin{document}
 \title{ Statistical inference for the doubly stochastic self-exciting process}
\author{ Simon Clinet$^{1,2}$  and Yoann Potiron$^{3}$ \\ 
$^1$ \small Graduate School of Mathematical Sciences, University of Tokyo: 3-8-1 Komaba, Meguro-ku, \\Tokyo 153-8914, Japan. Email: simon@ms.u-tokyo.ac.jp, website : http://www.ms.u-tokyo.ac.jp/\char`\~simon/ \\
$^2$ CREST, Japan Science and Technology Agency, Japan.\\
$^3$ Faculty of Business and Commerce, Keio University. 2-15-45 Mita, Minato-ku, Tokyo, 108-8345, Japan. Phone:  +81-3-5418-6571. Email: potiron@fbc.keio.ac.jp, website: http://www.fbc.keio.ac.jp/\char`\~ potiron}
\date{ This version: \today}
\maketitle

\begin{abstract}
We introduce and show the existence of a Hawkes self-exciting point process with exponentially-decreasing kernel and where parameters are time-varying. The quantity of interest is defined as the integrated parameter $T^{-1}\int_0^T\theta_t^*dt$, where $\theta_t^*$ is the time-varying parameter, and we consider the high-frequency asymptotics. To estimate it na\"{i}vely, we chop the data into several blocks, compute the maximum likelihood estimator (MLE) on each block, and take the average of the local estimates.  The asymptotic bias explodes asymptotically, thus we provide a non-na\"{i}ve estimator which is constructed as the na\"{i}ve one when applying a first-order bias reduction to the local MLE. We show the associated central limit theorem. Monte Carlo simulations show the importance of the bias correction and that the method performs well in finite sample, whereas the empirical study discusses the implementation in practice and documents the stochastic behavior of the parameters.

\end{abstract}


\section{Introduction} \label{introduction}
In high-frequency data, market events are observed more often than ever. As an example, the correlation between the timing of those events and other financial quantities, such as asset price, volatility and microstructure noise has become of special interest. Also, financial agents can model the order book to predict key quantities, such as the volume of trades in the next hour. For all those reasons, models for inter-arrival times, also called duration models, are needed. As a pioneer work, \cite{engle1998autoregressive} introduced the autoregressive conditional duration (ACD) model. Other references include and are not limited to \cite{van2001duration}, \cite{zhang2001nonlinear}, as well as \cite{bowsher2007modelling}, \cite{fernandes2006family}, and more recently \cite{renault2014dynamic} and \cite{potiron2017estimation}. 

\smallskip
The cited work is partly based on the self-exciting Hawkes point process introduced in \cite{HawkesPointSpectra1971} and \cite{hawkes1971spectra}. In that model, the intensity of the point process $N_t$  is defined as $\lambda (t) := \nu + \int_0^t \phi_{t-s} dN_s$, where the baseline $\nu > 0$. Self-exciting processes are very popular to model phenomena mainly because future events can be boosted by past events. In the high-frequency finance literature, \cite{toke2011market} documented such time-clustering property in the order flow of several stocks. Other examples of application can be found in \cite{embrechts2011multivariate}, \cite{bacry2013modelling}, \cite{bacry2013some}, \cite{ZhengRoueffAbergel2014}, \cite{jaisson2015limit}, etc. Also, \cite{bacry2015hawkes} offers a general overview of the Hawkes process applications in finance.

\smallskip
 We restrict our attention to the case with exponential exciting function $\phi_t = a e^{- bt}$, as studied in \cite{OakesMarkovianHawkes1975}. Time-varying parameter extensions have already been considered taking the locally stationary processes approach in \cite{roueff2016locally}, and restricting to the baseline time-varying case in \cite{fox2016modeling}, \cite{chen2013inference} and \cite{chen2016nonparametric}. Our approach is much in line with the latter couple of work in that we consider the high-frequency point of view. In \cite{chen2013inference} the authors allow the background parameter (as they call it) $\nu$ to be time-varying to incorporate intraday seasonality and consider the ACD model with time-varying background parameter. They illustrate that on data the ACD performs better when allowing for time-varying background, and that as it was already well-documented in \cite{engle1998autoregressive} the background parameter is moving a lot intraday. 
 
 \smallskip
 This calls into question what happens to the other two parameters $a$ and $b$ when sampling at the ultra high frequency? Do they look constant intraday? In our empirical study, we document that they are moving intraday just as the background parameter does although the intraday seasonality pattern isn't as clear. Indeed from one day to the next, the paths are very much different and although intraday seasonality can definitely be considered as one factor, it seems that it can't solely explain such behavior. Correspondingly we introduce a self-exciting process with stochastic time-varying parameters $\theta_t^* := (\nu_t^*, a_t^*, b_t^*)$. The new object of interest is defined as the integrated parameter 
\begin{eqnarray}
\label{integratedparameter}
\Theta := \frac{1}{T} \int_0^T \theta_t^* dt,
\end{eqnarray}
where $T > 0$ is the horizon time. 

\smallskip
To estimate the integrated parameter (\ref{integratedparameter}), we choose to do locally MLE estimations, which was studied in a parametric context  in \cite{Clinet20171800}, and whose numerical computation can be consulted in \cite{ozaki1979maximum}. Specifically, if we consider $B_n := n h_n^{-1}$ regular non-overlapping blocks of observation with time length $\Delta_n := T h_n n^{-1}$, the estimator of (\ref{integratedparameter}) is defined as
\begin{eqnarray}
\label{newestimator} \widehat{\Theta}_n := \frac{1}{B_n} \sum_{i=1}^{B_n} \widehat{\Theta}_{i,n},
\end{eqnarray}
where $\widehat{\Theta}_{i,n}$ corresponds to the MLE applied to the market events on the $i$th block, $n$ corresponds to the number of events' order between $0$ and $T$ (typically the expected number of events) and the block size $h_n$ stands for the number of events in a block's order (typically the expected number of events on a block). The idea to use a Riemann sum of local estimates in high-frequency finance problems is very common, and can be found for example in \cite{jacod2013quarticity} or \cite{mykland2009inference}. Our own recent work includes \cite{clinet2017efficient}. The more general literature on local parametric approaches, when not considering the high-frequency data case, includes \cite{fan1996local}, but also \cite{hastie1993varying},  the locally stationary processes of \cite{dahlhaus1997fitting}, etc. 

\smallskip
The first contribution of this paper is to obtain conditions on the stochastic parameter $\theta_t^*$ and the block size $h_n$ under which we can show a local central limit theorem (LCLT) in high-frequency asymptotics, and finiteness of moments of order $2 \kappa > 2$. The technique used, namely Quasi Likelihood Analysis (QLA) whose most general and powerful formulation can be consulted in \cite{YoshidaPolynomial2011}, is not problem-specific and can very much be applied to different models. For this part, blocks with $h_n$ which goes to infinity very slowly will be preferred, as the block length $\Delta_n$ will be smaller, and thus the parameter $\theta_t^*$ almost constant on each block. In particular, if $\theta_t^*$ is not constant, we obtain that a necessary condition is 
\begin{equation}
\label{localass}
h_n = o (\sqrt{n}).
\end{equation}

\smallskip
The second issue that this work is addressing is the asymptotic bias generated by $\widehat{\Theta}_n$. Even in the simple parametric case, note that the bias of the MLE on each block $\widehat{\Theta}_{i,n}$ is of order $h_n^{-1}$, and thus that the bias of $\widehat{\Theta}_n$ is also of the same order $h_n^{-1}$. The asymptotic bias, i.e. the bias of the scaled error $\sqrt{n} \big( \widehat{\Theta}_n - \Theta \big)$, is thus of order $\sqrt{n} h_n^{-1}$. If we want to obtain no asymptotic bias, we thus need to assume that 
\begin{equation}
\label{biasass}
\sqrt{n} = o (h_n).
\end{equation}
Thus, for that part, the block size $h_n$ should be as large as possible. 

\smallskip
In view of the necessary conditions (\ref{localass}) and (\ref{biasass}), there is no hope to obtain any $h_n$ for which the asymptotic bias of $\widehat{\Theta}_n$ will vanish. For that reason, we derive the one-order bias-corrected parametric MLE. Correspondingly, we define $\widehat{\Theta}_{i,n}^{(BC)}$ as the bias-corrected MLE when fitted to the observations on the $i$th block. Moreover, the bias-corrected estimator of (\ref{integratedparameter}) is defined as 
\begin{eqnarray}
\label{bcnewestimator} \widehat{\Theta}_n^{(BC)} := \frac{1}{B_n} \sum_{i=1}^{B_n} \widehat{\Theta}_{i,n}^{(BC)}.
\end{eqnarray}
We provide conditions under which $\widehat{\Theta}_n^{(BC)}$ has no asymptotic bias.
Finally, the global central limit theorem (GCLT) is obtained as an immediate consequence of the finiteness of moments of order $2\kappa$, the LCLT and the fact that the asymptotic bias of $\widehat{\Theta}_{n}^{(BC)}$ is null.  

\smallskip
The following section provides the setup, Section $3$ develops the statistical underpinning for the time-varying self-exciting process case and Section $4$ introduces the general model. In Section $5$, we discuss the main results. We give some practical guidance about the implementation of the statistical procedure in Section $6$. We also carry out numerical simulations in Section $7$, and give an empirical illustration on real tick-by-tick data in section $8$. Finally, Section $9$ concludes. Proofs can be found in the Annex.

\section{The setup}
\label{setup}

In this work, the terminology "market event" should be understood as possibly corresponding to a time of trade, bid or ask order (limit or market), an order of cancellation, the time of a price change, etc. We need to introduce some notation first, that will be used throughout this work. For any stochastic process $X_t$, we define $\F^X = (\calf_t^X)_{t\in [0,T]}$, where $\calf_t^X = \sigma \{ X_s, 0 \leq s \leq t \}$ designates the canonical filtration generated by $X_t$. We assume that $N_t^n$ is a point process, which counts the number of events on $[0,t]$. It means that $dN_t^n = 1$ if there is a market event at time $t$ and $dN_t^n = 0$ if not. Moreover, we assume that there is no jump at time $0$ and thus that $dN_0^n = 0$. Correspondingly, we define the intensity of market events $\lambda_*^n (t)$. The intensity process can be thought as the instantaneous expected number of events, i.e. $\lambda_*^n(t) dt = \esp\l[dN_t^n | \calf_t^{(\theta^*, N^n)}\r]$, where $\calf_t^{(\theta^*, N^n)}$ is the filtration generated by $\calf_t^{\theta^*}$ and $\calf_t^{N^n}$. For definitions, the reader can consult \cite{daley2007introduction} or \cite{JacodLimit2003} for more general results about the compensator of a point process. 

\smallskip There are commonly two ways to make the number of events go to infinity. The low-frequency asymptotics assume that $T \rightarrow \infty$. \cite{Clinet20171800} took this approach in an ergodic framework. On the contrary, the high-frequency point of view (also sometimes called heavy traffic asymptotics) assumes that $T$ is fixed, and that the number of events explodes on $[0, T]$. We adopt the latter approach and further consider a sequence of intensities such that $\esp [ \lambda_*^n (t) ]$ is exactly of order $n$, with $n \rightarrow \infty$. This yields a number of observations $N_T^n$ of order $n$, so that we are in the classical framework of the large-sample theory.   

\section{Outline of the problem: an illustrative example}
\label{outline}
We start our theoretical exposition by the introduction of a point process toy model which provides an insight on the difficulties to overcome when considering the self-exciting model case. For the sake of simplicity, we stay at a heuristic level. The continuous parameter $\theta_t^*$ is assumed to be 1-dimensional throughout the rest of this section. The parameter $\theta_t^*$ is also restricted to belong to a compact set $K = \big[ \underline{\theta}, \overline{\theta} \big]$, where $\underline{\theta} > 0$. Moreover, $\theta_t^*$ is assumed to be adapted to some filtration $\calf_t$, and to satisfy uniformly in $0 \leq s < t \leq T$ that $\esp_{s} \big[ |\theta_t^* - \theta_s^*|^p \big] = O_\proba((t-s)^p)$, where $\esp_s [ . ]$ denotes the conditional expectation with respect to $\calf_s$. Finally, we assume that the process $N_t^n$ is adapted to $\calf_t$ and follows the dynamic of a doubly stochastic Poisson process (or Cox process) whose underlying stochastic intensity is assumed to be defined as $\lambda_*^n(t) = n \sqrt{\theta_t^*}$. 

\smallskip
The estimation procedure $\widehat{\Theta}_n$ follows \cite{potiron2016estimating}. We are interested in assessing the GCLT $\sqrt{n}(\widehat{\Theta}_n - \Theta) \to^d V_T^{\half} \mathcal{N} (0,1)$, where the asymptotic random variance $V_T = T^{-1} \int_0^T v_t dt$ is independent from $\mathcal{N} (0,1)$. Since the parameter $\theta_t^*$ is smooth, we obtain 
\bea 
\label{outlineparamsmooth}
\Theta = \inv{B_n} \sum_{i=1}^{B_n} \theta_{(i-1)\Delta_n}^*  + O_\proba(\Delta_n).
\eea
Consequently, the GCLT will follow if we can prove that 
\bea
\label{outlineCLT}
\frac{\sqrt{n}}{B_n} \sum_{i =1}^{B_n} \big( \widehat{\Theta}_{i,n}-\theta_{(i-1)\Delta_n}^* \big) \to^d V_T^{\half} \mathcal{N} (0,1).
\eea
We focus on how to obtain (\ref{outlineCLT}) in this simple toy model. To do that, we rewrite the left hand-side of (\ref{outlineCLT}) as a sum of a martingale triangular array and an array of biases. Formally, (\ref{outlineCLT}) is expressed as
\bea 
\underbrace{\frac{\sqrt{n}}{B_n}\sum_{i =1}^{B_n}{M_{i,n}}}_{S_{n}^{(M)}} + \underbrace{\frac{\sqrt{n}}{B_n}\sum_{i=1}^{B_n}{B_{i,n}}}_{S_{n}^{(B)}}  \to^d V_T^{\half} \mathcal{N} (0,1),
\label{cltCondition}
\eea
where $M_{i,n} = \widehat{\Theta}_{i,n} - \theta_{(i-1)\Delta_n}^* - \esp_{(i-1)\Delta_n}\big[\widehat{\Theta}_{i,n}-\theta_{(i-1)\Delta_n}^*\big]$ and $B_{i,n} = \esp_{(i-1)\Delta_n}\big[\widehat{\Theta}_{i,n}-\theta_{(i-1)\Delta_n}^*\big]$. Our strategy to show (\ref{cltCondition}) relies thus on exploiting the  martingale decomposition on the left hand side of (\ref{cltCondition}) to show that the covariances between blocks are negligible. More precisely, we want to prove that $S_{n}^{(M)} \to^d V_T^{\half} \mathcal{N} (0,1)$ on the one hand, and that $S_{n}^{(B)} \to^\proba 0$ on the other hand.  To show the former statement, classical sufficient conditions (see for instance Corollary 3.1 of pp. 58-59 in \cite{hall1980martingale}, or also Theorem VIII.3.33 in \cite{JacodLimit2003}) will hold if\footnote{The reader can find more details in Section \ref{GCLTproof}} uniformly in $i \in \{1,...,B_n \}$ we can show that
\bea 
\label{outlinemom}
\esp_{(i-1)\Delta_n}\l[\l(\sqrt{h_n}\l(\widehat{\Theta}_{i,n}-\theta_{(i-1)\Delta_n}^*\r)\r)^2\r] = v_{(i-1)\Delta_n} +o_\proba(1), 
\eea 
and for some $\kappa > 1$ that
\bea 
\label{outlinelin}
\esp_{(i-1)\Delta_n}\l[\l|\sqrt{h_n}\l(\widehat{\Theta}_{i,n}-\theta_{(i-1)\Delta_n}^*\r)\r|^{2\kappa}\r] = O_\proba(1).
\eea 
If we show the LCLT, i.e. the convergence of $\sqrt{h_n}(\widehat{\Theta}_{i,n}-\theta_{(i-1)\Delta_n}^*) \to^d v_{(i-1)\Delta_n}^{\frac{1}{2}} \mathcal{N} (0,1)$ uniformly in the block number $i \in \{ 1, \cdots, B_n \}$, we can deduce from (\ref{outlinelin}) that (\ref{outlinemom}) holds. This will be our strategy to show that $S_{n}^{(M)} \to^d V_T^{\half} \mathcal{N} (0,1)$. Moreover, to obtain the GCLT (\ref{cltCondition}), we also need to show that the array of biases vanishes asymptotically. Accordingly, we will look at how to obtain those three conditions (boundedness of local moments of order $2\kappa$, LCLT and no asymptotic bias) in the toy model.\\
\smallskip 

To fix ideas, we provide one way, which turns out to be helpful when estimating (\ref{integratedparameter}), to obtain asymptotic properties of the MLE in the parametric case when the intensity of the point process $N_t^n$ is defined as $\lambda_{*}^n (t) := n \sqrt{\theta^*}$. The log-likelihood of the parametric model can be expressed up to a constant additive term as 
\bea 
l_n(\theta) = \text{log}\l(\sqrt{\theta}\r)N_T^n - n\sqrt{\theta}T,
\eea
whose maximizer $\hat{\theta}_n$ admits the explicit form
\bea 
\label{outlineMLEexplicitform}
\hat{\theta}_n = \l(\frac{N_T^n}{nT}\r)^2.
\eea
If we introduce the martingale $\tilde{N}_t^n = N_t^n - n\sqrt{\theta^*}t$, we can rewrite $\hat{\theta}_n$ as a function of $\tilde{N}_T^n$:
\bea 
\hat{\theta}_n = \theta^*  + \frac{2\sqrt{\theta^*}}{nT}\tilde{N}_T^n + \l(\frac{\tilde{N}_T^n}{nT}\r)^2.
\label{closedForm}
\eea
As a consequence of classical limit theorems on martingales (see, e.g., Theorem 2.28 of p. 152 in \cite{mykland2012econometrics} if we interpolate $\tilde{N}_t^n$ as a continuous martingale, or the more general Theorem IX.7.3 in p. 584 of \cite{JacodLimit2003}), we obtain the CLT
 $$\sqrt{n} (\hat{\theta}_n -\theta^*) \to^d T^{-\half}\Gamma(\theta^*)^{-\half} \mathcal{N} (0,1),$$
where the Fisher information has the form $\Gamma(\theta^*) = \inv{4}(\theta^*)^{-\frac{3}{2}}$. We also have the stronger statement that for any $p \geq 1$: 
\bea 
\esp \l[ \l(\sqrt{n} (\hat{\theta}_n -\theta^*)\r)^p \r] \to \esp\l[\l(T^{-\half}\Gamma(\theta^*)^{-\half}\xi\r)^p\r],
\eea
where $\xi$ follows a $\mathcal{N} (0,1)$. Finally, we can also compute in (\ref{closedForm}) the finite-sample bias of the MLE 
\bea 
\esp \big[ \hat{\theta}_n - \theta^* \big] = \frac{\sqrt{\theta^*}}{nT}.
\label{biasConstant}
\eea

\smallskip
We are now back to the time-varying parameter model case $\lambda_*^n(t) = n \sqrt{\theta_t^*}$. In that case, we adapt the definition of the martingale as $\tilde{N}_t^n = N_t^n - n \int_0^t \sqrt{\theta_s^*} ds$. Working out from the explicit form (\ref{outlineMLEexplicitform}), the local MLE can be expressed as 
\bea \nonumber
\widehat{\Theta}_{i,n} = \l(\frac{\tilde{N}_{\blockf }^n-\tilde{N}_{\blockd }^n}{h_nT}\r)^2 +\frac{2}{h_n^2T^2}(\tilde{N}_{\blockf}^n - \tilde{N}_{\blockd}^n )\int_{\blockd }^{\blockf }{n\sqrt{\theta_s^*}ds} \\ \label{outlineMLEform} +  \inv{h_n^2T^2}\l(\int_{\blockd }^{\blockf }{n\sqrt{\theta_s^*}ds}\r)^2.
\eea
In view of (\ref{outlineMLEform}) and under the assumption that $h_n = o (n^{\frac{2}{3}})$, it is easy to obtain the LCLT with local conditional variance $v_{s} = T^{-1}\Gamma \big(\theta_{s}^* \big)^{-1}$ and the boundedness of moments of order $2\kappa$. It remains to control the array of biases $S_n^{(B)}$. Calculation gives us
\bea 
B_{i,n} = \frac{\sqrt{\theta_{(i-1)\Delta_n}^*}}{h_nT} + O_\proba(\Delta_n),
\label{biasMoving}
\eea
where the residual term $O_\proba(\Delta_n)$, which was not part of the parametric bias (\ref{biasConstant}), is due to the deviation of $\theta_t^*$. In order to obtain no asymptotic bias, we assume that $\sqrt{n} = o (h_n)$. Consequently, if we assume that $h_n = n^{1/\delta}$ with $\frac{3}{2} < \delta < 2$, we can prove the GCLT  with asymptotic variance $V_T = T^{-2}\int_0^T{\Gamma\l(\theta_t^*\r)^{-1}dt}$ in this toy model . This is a simple example where no further bias correction is needed to obtain the GCLT. However, in the time-varying self-exciting model, we will require to bias correct the estimator. This could be done in this simple setting via
\bea 
\widehat{\Theta}_{i,n}^{(BC)} = \widehat{\Theta}_{i,n} - \frac{\sqrt{\widehat{\Theta}_{i,n}}}{h_nT}.
\eea
%A calculation would show that if we use $\widehat{\Theta}_{i,n}^{(BC)}$ the lower bound on $h_n$ can be lowered from $n^{\half}$ to $n^{\frac{1}{3}}$. 

\section{The model} \label{model}
 We introduce in this section the time-varying self-exciting process, which will also be called the doubly stochastic Hawkes process, in analogy with the doubly stochastic Poisson process introduced in \cite{cox1955some}. We first recall the definition of the non time-varying self-exciting point process. In the parametric case, the point process $N_t^{P,n}$ can be defined via its intensity function 
\bea
\label{intensityPHP}
\lambda_{*}^{P,n} (t) = n\nu^{*} + \int_0^{t-}{na^{*} e^{-nb^{*}(t-s)} dN_s^{P,n}}, 
\eea
where $\theta^* = (\nu^*, a^*, b^*)$ is the 3-dimensional parameter. The self-excitation property can be read directly from the intensity form $\lambda_{*}^{P,n}$ in (\ref{intensityPHP}). Indeed, a market event arriving at time $t$ will immediately boost the intensity, with an additional factor of magnitude $na^*$, favoring the occurrence of new events in the close future. The excitation then exponentially fades away after a time of order $\l(nb^*\r)^{-1}$. We explain now our choice regarding the asymptotics. First, we assume that the baseline intensity is proportional to $n$ to boost the average rate of spontaneous events. Moreover, we assume that the excitation variables are of magnitude $(na^*,nb^*)$ in order to preserve the proportionality between the typical excitation time after a market event, $\l(nb^*\r)^{-1}$, and the average inter-arrival time between two spontaneous events $(n\nu^*)^{-1}$. To wrap it up, $N_t^n$ is a self-exciting process with parameters $(n \nu^*, n a^*, nb^*)$. Note that other choices can lead to fairly different asymptotics such as the ones in \cite{ogihara2015quasi} where authors suggested a model with baseline $n\nu^*$ but a constant excitation kernel of the form $a^*e^{-b^*t}$.       

\smallskip
We consider now the time-varying case. We assume that the 3-dimensional time-varying parameter process $\theta_t^*$ is component-wise positive and is confined into the interior of a compact space $K$. This implies the existence of two non-negative vectors $\underline{\theta}$ and $\overline{\theta}$ such that $0 < \underline{\theta} \leq \theta \leq \overline{\theta}$ for any 
$\theta \in K$, where the inequalities should be read component-wise. Moreover, we assume that $N_t^n$ admits the $\calf_t$-stochastic intensity $\lambda_{*}^n(t)$ defined as 
\bea 
\label{intensityTVHP}
\lambda_{*}^n(t) = n\nu_{t}^{*} + \int_0^{t-}{na_{s}^{*} e^{-nb_{s}^{*}(t-s)} dN_s^n}, \text{    } t \in (0,T],
\eea
where $N_t^n$ and $\theta_t^*$ are adapted to $\calf_t$\footnote{The formal definition of $\calf_t$ can be found in Section \ref{mainResults}.}, and $N_0^n = 0 $ a.s. The time-varying model (\ref{intensityTVHP}) is a natural time-varying parametric model extension of (\ref{intensityPHP}). It is constructed in the same spirit as for the doubly stochastic Poisson process, in the sense that conditionally on the path of $\theta_t^*$, $N_t^n$ is distributed as a standard inhomogeneous Hawkes process. The formal definition of such a property along with the existence of the doubly stochastic Hawkes process can be found in Theorem \ref{lemmaExistence}. Finally, note that the time-varying parameter model (\ref{intensityTVHP}) is more general than the parametric model (\ref{intensityPHP}). In particular, the intensity between two market events is not exponentially decreasing, but rather a sum of decreasing exponential functions, each one with its own starting point and decreasing rate. 

\section{Main Results}\label{mainResults}
\subsection{Preliminary results}
We present in this section general results for the doubly stochastic Hawkes process. We start by stating basic conditions on a given parameter process $\theta_t$ that ensure the existence of the related doubly stochastic Hawkes process. 
\bd
\im[[E\!\!]]
    \bd 
        \im[{\bf (i)}]  \bea 
                         r := \sup_{t \in [0,T]} \frac{a_t}{b_t} < 1 \text{    } \proba-a.s.
                         \label{condQuotient}
                    \eea 
        \im[{\bf (ii)}]  \bea 
                        \int_0^T \nu_s ds < +\infty \text{    } \proba-a.s.
                        \label{condNu}
                     \eea
   
    \ed 
\ed
First, note that (\ref{condQuotient}) is not harmful. Indeed, the corresponding condition for the existence of the parametric model is $\frac{a}{b} < 1$. Moreover, when estimating parameters by local MLE, we need $\theta_t$ to be contained within a compact set. Thus, (\ref{condNu}) will be verified automatically in that context. The next theorem shows the existence of the doubly stochastic Hawkes process associated with the process $\theta_t$. We recall that $\calf_t^\theta$ designates the canonical filtration associated with $\theta_t$. Moreover, the following bigger filtration $\calf_t$ is introduced for the construction of the doubly stochastic Hawkes process. We define the filtration as $\calf_t = \calf_t^{(\theta, \overline{N})} = \calf_t^{\theta} \vee \calf_t^{\overline{N}}$, where $\overline{N}_t = \overline{N}([0,t] \times \reels)$ is a  Poisson process of intensity $1$ on $\reels^2$ which is independent from $\theta_t$.  

\begin{theorem*} (Existence)
 Under \textnormal{[E]}, there exists a point process $N_t$ adapted to $\calf_t$ such that its $\calf_t$-intensity has the representation
\bea 
\lambda(t) = \nu_{t} + \int_0^{t-}{a_s e^{-b_s(t-s)} dN_s}.
\label{eqRecursive}
\eea
Moreover, conditionally on the path of $\theta_t$, $N_t$ is distributed as a standard Hawkes process with inhomogeneous deterministic parameter  $\theta_t$, that is 
\bea 
\esp \l[  f(N) \l| \calf_T^{\theta} \r. \r] = \esp \l[ f\l(N^{\tilde{\theta}}\r) \r]_{|\tilde{\theta} = \theta},
\label{doublyStoProp}
\eea
for any continuous bounded function $f$, and where $N_t^{\tilde{\theta}}$ is a doubly stochastic Hawkes process with underlying deterministic process $\tilde{\theta}_t$.

\label{lemmaExistence}
\end{theorem*}

From now on we assume that $\theta_t^*$ satisfies Condition [E]. Under this assumption, since $N_t^n$ is a time-varying self-exciting process with parameters $(n \nu_t^*, n a_t^*, n b_t^*)$, $N_t^n$ is well-defined and adapted to $\calf_t$.

\smallskip
We describe the statistical procedure, provide a formal definition of the local MLE $\widehat{\Theta}_{i,n}$ as well as its first order bias-corrected version $\widehat{\Theta}_{i,n}^{(BC)}$. We state their asymptotic properties, including the main result of this paper which is the GCLT for $\widehat{\Theta}_n^{(BC)}$ in Theorem \ref{GCLT}. Recall that we have chopped our observations into $B_n$ time blocks of the form $((i-1)\Delta_n, i\Delta_n]$. For any $i \in \{1,...,B_n\}$ and any $\theta \in K$, we consider the regression family of intensities
\bea 
\lambda^{i,n}(t,\theta) = n\nu + \int_{(i-1)\Delta_n}^{t-}{nae^{-nb(t-s)}dN_s^n},
\label{candidateIntensity}
\eea
defined for $t \in ((i-1)\Delta_n, i\Delta_n]$. We now define the Quasi Log Likelihood\footnote{The model is by definition misspecified and thus $l_{i,n}$ is not the log likelihood function of the model.} on the $i$-th block as
\bea 
l_{i,n}(\theta) = \int_{(i-1)\Delta_n}^{i \Delta_n}{\text{log}\l(\lambda^{i,n}(t,\theta)\r)dN_t^n} - \int_{(i-1)\Delta_n}^{i \Delta_n}{\lambda^{i,n}(t,\theta)dt}.
\label{QLL}
\eea
We take the local MLE $\widehat{\Theta}_{i,n}$ as one maximizer of the Quasi Log Likelihood on the $i$-th block  defined as
\bea 
l_{i,n}\l(\widehat{\Theta}_{i,n}\r) = \max_{\theta \in K}l_{i,n}(\theta). 
\label{defLocalMLE}
\eea 
Looking at the form of (\ref{candidateIntensity}), (\ref{QLL}) and (\ref{defLocalMLE}), we can see that $\lambda^{i,n}$, $l_{i,n}$ and $\widehat{\Theta}_{i,n}$ are functions of the $i$-th block's events\footnote{Note that this doesn't mean that $\widehat{\Theta}_{i,n}$ are uncorrelated.}. In particular, we don't take account for the possible preexcitation induced by past events in the expression of the candidate intensity (\ref{candidateIntensity}), as the lower bound of the integral is fixed to $\blockd$. Asymptotically, such approximation is valid because the exponential form of the excitation kernel along with the order of the excitation parameters $(na_t^*, nb_t^*)$ induce a weak-enough influence of the past events on the actual stochastic intensity $\lambda_*^n(t)$. 

\smallskip
In what follows we specify the form of $h_n$ and assume the existence of an exponent $\delta > 1$ such that 
\bea 
h_n = n^{1/ \delta}.
\label{formH}
\eea
We will also have to specify the smoothness of the process $\theta^{*}$ using the following quantities. First, define the regularity modulus of order $p \in \naturels - \{0\}$, at time $t \in [0,T]$ and value $\theta \in K$ as 
\bea 
w_p(t,\theta,r) = \esp\l[\l.\sup_{h \in [0, r \wedge (T - t)]} |\theta_{t+h}^{*} - \theta_t^{*}|^p\r| \calf_t, \theta_t^{*} = \theta \r], \text{   } r > 0.
\eea
We then define the global regularity modulus as 
\bea 
w_p(r) = \sup_{(t,\theta) \in [0,T] \times K} w_p(t,\theta,r), \text{  } r>0. 
\eea 

We introduce the following conditions needed to obtain the LCLT and the boundedness of moments.
\bd
\im[[C\!\!]] 
  \bd 

  \im[{\bf (i)}] There exists an exponent $\gamma \in  (0,1]$, such that for  $r \to 0$, we have
  \bea
  w_p(r)= O_\proba \l(r^{\gamma p}\r).
  \eea 
  \im[{\bf (ii)}] $\delta$ and $\gamma$ satisfy the relation 
                   \bea 
                    \delta > 1+ \inv{\gamma}.
                   \eea
  \im[{\bf (iii)}] The excitation parameters $a_t^*$ and $b_t^*$ satisfy 
                    \bea \label{condBoundedMoments} c :=  \sup_{(t,n) \in [0,T] \times \naturels} \int_0^t{n a_s^* e^{-n b_s^*(t-s)}ds}  < 1 \text{    } \proba-a.s.
                     \eea
  \ed
\ed
Note that the conditional expectation $\esp\l[ .| \calf_s, \theta_s^{*} = \theta \r]$ refers to the operator $\esp[.|\calf_s]$ conditioned on $\theta_s^* = \theta$. By definition, for a $\calf$-measurable random variable $X$, if we write $G_X(\theta) = \esp\l[ X| \calf_s, \theta_s^{*} = \theta \r]$, the relationship between both expecations can be expressed as $\esp[X|\calf_s]=G_X(\theta_s^{*}) $. The justification of the existence of $\esp\l[ .| \calf_s, \theta_s^{*} = \theta \r]$ can be found in Section \ref{sectionLCLT}. Condition [C]-{\bf (i)} quantifies the regularity of the process $\theta_t^*$ through the regularity exponent $\gamma$. A natural example of a process satisfying [C]-{\bf (i)} is the drift function, i.e. of the form 
\begin{eqnarray}
\label{drift}
\theta_t^* = \theta_0^* + \int_0^t{u_s^* ds}, 
\end{eqnarray}                           
where $u^*$ is a stochastic process that takes its values in a compact subset of $\reels^3$. Another example is a smoothed version of the Brownian motion that can be obtained as follows. Take some $\tau > 0$, a positive vector $\theta^{(M)} \in \reels^3$, a positive diagonal matrix $\sigma = \text{diag}(\sigma^{\nu},\sigma^{a},\sigma^{b})$ and consider the process 
\bea 
\theta_t^{*} = \theta^{(M)} + \frac{\sigma} {\tau}\int_{t-\tau}^t{W_s ds},
\label{smoothBM}
\eea 
where $(W_t)_{t \in \l[-\tau,T\r]}$ is a 3-dimensional standard Brownian motion. One can confine $\theta_t^{*}$ in a compact space by stopping the process $W$ when it reaches some critical value. This second example is useful to model the stochastic component of the parameter as a nuisance process, and we use (\ref{smoothBM}) in our simulation study. Note that the smaller $\tau$, the less auto-correlated $\theta_t^{*}$ will be, and that we would be back to a Brownian motion in the limit $\tau \to 0$. For both examples (\ref{drift}) and (\ref{smoothBM}) we have $\gamma = 1$, but note that the correlation structure of $\theta_t^{*}$ may be very complex though (to do so, we can take any process $u_t^{*}$ which has a complex correlation structure).\\

\smallskip 

Condition [C]-{\bf (ii)} controls the lower bound of $h_n$ and is necessary to derive the LCLT and the local boundedness of moments. In particular, as $\gamma \leq 1$, [C]-{\bf (ii)} implies that $h_n = o(\sqrt n)$. This was stated in (\ref{localass}). Finally [C]-{\bf (iii)} is an additional condition that ensures the existence of moments of $N^n$. We can see that  [C]-{\bf (iii)} is automatically satisfied if $\underline{a} \leq a^* \leq \overline{a}$, $\underline{b} \leq b^* \leq \overline{b}$ and $\overline{a} < \underline{b}$.

\smallskip
We specify now the value of the exponent
\bea 
\kappa = \gamma(\delta -1) > 1,
\label{defKappa}
\eea
where the inequality is a direct consequence of [C]-{\bf (ii)}. For $\theta \in K$, the positive symmetric matrix $\Gamma(\theta)$ is defined as the asymptotic Fisher information of a parametric Hawkes process generated by $\theta$ and can be found in (\ref{paramFisher}). The next theorem encompasses the LCLT and the local convergence of moments of order smaller than $2\kappa$ of the rescaled local MLE $\sqrt{h_n}\l(\widehat{\Theta}_{i,n} - \theta_{(i-1) \Delta_n}^*\r)$.    

\begin{theorem*} \label{conditionalCLT} (LCLT and boundedness of moments)
Let $L \in [0,2\kappa)$. Under \textnormal{[C]}, we have uniformly in $i \in \{ 1, \cdots, B_n \}$ that 
\bea 
	\esp_{(i-1) \Delta_n} \l[ f \l( \sqrt{h_n}\l(\widehat{\Theta}_{i,n} - \theta_{(i-1) \Delta_n}^*\r) \r) \r] =  \esp_{(i-1) \Delta_n} \l[ f \l( T^{-\half}\Gamma\l(\theta_{(i-1) \Delta_n}^*\r)^{-\frac{1}{2}}\xi \r) \r] + o_\proba(1)
	\label{convMoment2}
\eea
for any continuous function $f$ with $|f(x)| = O(|x|^{L})$ when $ |x| \to \infty$ , and such that $\xi$ follows a standard normal distribution and is independent of $\calf$.
\end{theorem*}
We now introduce the first-order bias-corrected local MLE for any $i \in \{1,...,B_n\}$ as 

\bea 
\widehat{\Theta}_{i,n}^{(BC)} = \widehat{\Theta}_{i,n} - \frac{b\l(\widehat{\Theta}_{i,n}\r)}{h_n T},
\label{localMLEBC}
\eea
where $b$ is defined in (\ref{paramb}), Section \ref{sectionMLEclassic}, and should be compared to its very similar form for the classical i.i.d case, see e.g. \cite{cox1968general}. We finally recall the definition of the global bias-corrected estimator that was introduced in (\ref{bcnewestimator}), i.e.
\begin{eqnarray}
\label{bcnewestimator2} \widehat{\Theta}_n^{(BC)} = \frac{1}{B_n} \sum_{i=1}^{B_n} \widehat{\Theta}_{i,n}^{(BC)}.
\end{eqnarray}
In the next theorem, the expression $x \wedge y$ stands for $\min \{x,y\}$.

\begin{theorem*} (bias correction) \label{biasCorrection2}
Let $\epsilon \in (0,1)$. The bias of the estimator $\widehat{\Theta}_{i,n}$ admits the expansion
\bea 
\esp_{(i-1)\Delta_n}\l[\widehat{\Theta}_{i,n} - \theta_{(i-1)\Delta_n}^*\r] = \frac{b\l(\theta_{(i-1)\Delta_n}^*\r)}{h_n T} + O_\proba\l(h_n^{-\epsilon \l(\kappa \wedge \frac{3}{2}\r)}\r), 
\label{biasExpansion2}
\eea 
uniformly in $i \in \{1,..., B_n\}$. Moreover, the estimator $\widehat{\Theta}_{i,n}^{(BC)}$ has the uniform bias expansion
\bea 
\esp_{(i-1)\Delta_n}\l[\widehat{\Theta}_{i,n}^{(BC)} - \theta_{(i-1)\Delta_n}^*\r] =  O_\proba \l(h_n^{-\epsilon \l(\kappa \wedge \frac{3}{2}\r)}\r) .
\label{EqBiasCorrection2}
\eea 
\end{theorem*}

Now our aim is to combine Theorem \ref{conditionalCLT} and Theorem \ref{biasCorrection2} to state the asymptotic properties of the global estimator. In the following there are two parts. The main one gives the GCLT when the parameter is assumed to be sufficiently smooth. The second part investigates what happens when the parameter is rough. 

\subsection{Global central limit theorem when parameters are smooth} 

In this section, we state an additional condition on $\delta$ and $\gamma$ so that $\widehat{\Theta}_n^{(BC)}$ is asymptotically unbiased. 

\bd
\im[[BC\!\!]] $\delta$ and $\gamma$ satisfy the relation 
\bea 
\frac{\gamma}{\gamma-\half} < \delta < 3.
\eea 
\ed
Intuitively, the left-hand side inequality in [BC] ensures that the size of each block is not too big so that the bias induced by the parameter process $\theta_t^*$ itself is negligible. On the contrary, the right-hand side inequality is a sufficient condition to keep under control the finite-sample bias of the local MLE by avoiding too small blocks. More precisely, Condition \textnormal{[BC]} implies in particular that the  exponent $\gamma \in \l(\frac{3}{4},1\r]$. Note that such condition excludes the class of It\^{o}-processes as a parameter process. Moreover, on $\l(\frac{3}{4},1\r]$ we have $\frac{\gamma}{\gamma-\half} \geq 1 + \inv{\gamma}$ with equality for $\gamma =1$, and thus \textnormal{[BC]} is a stronger condition than \textnormal{[C]}-{\bf (ii)}. For instance, in the Lipschitz case $\gamma = 1$,  \textnormal{[BC]} (and thus \textnormal{[C]}-{\bf (ii)}) are satisfied for $ 2 < \delta < 3$. This means by definition of $\delta$ that $h_n$ must be taken so that $n^{\inv{3}} = o(h_n)$ and $h_n = o(n^\half)$. \\%The case where $\gamma \notin \l(\frac{3}{4},1\r]$ is treated below.\\
\medskip

We finally state the main result of this work which investigates the limit error of the bias-corrected estimator $\widehat{\Theta}_{n}^{(BC)}$.
\begin{theorem*} (GCLT)
\label{GCLT}
 Assume that \textnormal{[C]} and \textnormal{[BC]} hold. Then, $\calf_T^{\theta^{*}}$-stably in law as $n \rightarrow \infty$,
\begin{eqnarray}
 \label{GCLTeq} \sqrt{n} \big( \widehat{\Theta}_n^{(BC)} - \Theta \big) \rightarrow \l(T^{-2} \int_0^T \Gamma(\theta_s^*)^{-1} ds \r)^{\frac{1}{2}} \mathcal{N} (0,1),
\end{eqnarray}
where $\mathcal{N} (0,1)$ is independent from the $\sigma$-field $\mathcal{F}_T^{\theta^{*}}$.
\end{theorem*}

\begin{remark*}
(convergence rate) The convergence rate in Theorem (\ref{GCLT}) is the same as in the parametric case. We also  conjecture that the asymptotic variance is the non-parametric efficient bound.
\end{remark*}
 
%The proof of the consistency of (\ref{AV}) can be found in Section \ref{proofAV}.

\begin{remark*} (robustness to jumps in the parameter process) We assume that we add a jump component to the parameter process
\begin{eqnarray}
 \label{jumpmodel}
 \theta_t^* = \theta_t^{(C)} + \theta_t^{(J)},
\end{eqnarray}
where $\theta_t^{(J)}$ denotes a $3$-dimensional finite activity jump process and $d\theta_t^{(J,k)}$ is either zero (no jump) or a real number indicating the size of the jump at time $t$ for $k=1, 2, 3$. We further assume that there is no initial jump, i.e. $J_0 = 0$. Moreover, we assume that $J_t$ is a general Poisson process independent from the other quantities. Under similar assumptions, the results of this work can be adapted.
\end{remark*}
\begin{remark*} (mutually exciting process) The proofs can be adapted to a multidimensional Hawkes process. Investigating the corresponding conditions is beyond the scope of this paper.
\end{remark*}

\subsection{What happens in the rough parameter case?}
In this section, we are interested in the asymptotic properties of our estimators when the regularity condition $\gamma \in \l (\frac{3}{4},1\r]$ fails. We first give a theoretical argument to show that the bound $3/4$ can be lowered to $1/2$, although the some of the corresponding bias theoretical formula terms would be too involved to be of any practical interest. Nonetheless the bias can be computed with Monte Carlo methods (see Section \ref{goal} in our numerical study for more details). We then provide the expected convergence rate of the consistency for both the naive and the first order bias-corrected estimators. 

\smallskip
When $\gamma \notin \l(\frac{3}{4},1\r]$, Theorem \ref{GCLT} fails in general. This is due to the bias expansion obtained in Theorem \ref{biasCorrection2}, (\ref{EqBiasCorrection2}), whose order in $h_n^{-\kappa \wedge \frac{3}{2}}$ can be dominated by $n^{-\half}$ only if $\gamma > 3/4$. Nevertheless, we can expect that correcting for the bias to a higher order improves the rate of convergence in (\ref{EqBiasCorrection2}). Thus we would obtain a corresponding central limit theorem even for $\gamma \leq 3/4$. A closer investigation to the proofs shows that if one conducts the bias correction up to order $q \in \naturels - \{0\}$, conditions [C]-{\bf(ii)} and [BC] are respectively replaced by $\delta > 1 + (q+1)/(2\gamma)$ and $\gamma/(\gamma-1/2) < \delta < 2+q$, so that the GCLT becomes valid under the weaker condition $\gamma \in \l(\half + \frac{1}{2(1+q)},1\r]$. For $q \to +\infty$, the asymptotic admissible interval becomes thus $\l(\half,1\r]$, so it is theoretically possible to construct an asymptotically normal estimator for any $\gamma > 1/2$.

\smallskip
When $\gamma \in (0, \half]$, we can't use the same martingale approach. In general, the bias induced by the parameter process in the expansion (\ref{EqBiasCorrection2}) cannot be corrected without some information on the distribution of $\theta_t^{*}$. We can't show that the bias is of the right order, because this would imply a choice of $\delta$ such that $\gamma < \delta(\gamma -1/2) $ which is not possible if $\gamma$ reaches the critical value $\gamma = 1/2$. Investigating if other approaches yield a better estimate of the bias under additional specification on the structure of $\theta_t^{*}$ is beyond the scope of this work. 

\smallskip
We now turn to the convergence rate of our estimators when the central limit theorem fails. We can prove that both estimators $\widehat{\Theta}_n$ and $\widehat{\Theta}_n^{(BC)}$ are consistent. Indeed, it turns out that the first order bias corrected version is $n^\alpha$-consistent\footnote{An estimator $\bar{\Theta}_n$ is said to be $a_n$-consistent if $a_n(\bar{\Theta}_n-\Theta) = O_\proba(1).$ } for any  $\alpha  \in \big(0,{\frac{\gamma}{1+\frac{2}{3}\gamma}}\big)$, whereas the naive estimator is only $n^\alpha$-consistent for any  $\alpha  \in \big(0,{\frac{\gamma}{1+\gamma}}\big)$. More specifically, we have the following result. The proposition can be showed following a similar reasoning as in the proof of Theorem \ref{GCLT}.

\begin{proposition*} (consistency)\label{propConsistency}
For any $\alpha \in \l(0 , \frac{\gamma}{1+\gamma}\r)$, the choice $\delta \in \l(1+\inv{\gamma}, \inv{\alpha}\r)$ gives 
$$ n^\alpha(\widehat{\Theta}_n - \Theta) \to^\proba 0.$$
Moreover, if $\gamma \in (0,3/4]$, for any $\alpha \in \l(0,\frac{\gamma}{1+\frac{2}{3}\gamma}\r)$, the choice $\delta \in \l( (1+\inv{\gamma}) \vee \frac{\gamma}{\gamma-\alpha}, \frac{3}{2\alpha}\r)$ gives 
$$ n^\alpha(\widehat{\Theta}_n^{(BC)} - \Theta) \to^\proba 0.$$
\end{proposition*}
In particular we can see that $\widehat{\Theta}_n$ is already almost $\sqrt n$-consistent when $\gamma =1$ without any bias correction. In a similar way, the case $\gamma = 3/4$ also yields an almost $\sqrt n$-consistent bias corrected estimator $\widehat{\Theta}_n^{(BC)}$ as was expected. Again, knowing if the bounds for $\alpha$ given in Proposition \ref{propConsistency} are optimal is beyond the scope of this paper.

\section{Statistical implementation}
\label{statisticalimplementation}
In this section, we give some practical guidance to the above theory including a studentized version of the GCLT. Actually, on real data, the quantity of interest is $n \Theta$ whereas $n$ is (usually) unknown. This doesn't prevent us from obtaining a studentized version of the GCLT. A feasible procedure consists in estimating directly $n\Theta$ in place of estimating  $\Theta$. When properly divided by $n$, this yields the same estimate as the non feasible procedure, i.e. we have $n\widehat{\Theta} = \widehat{n\Theta}$, where $\widehat{\Theta}$ is the naive or the bias-corrected estimator. Indeed note that a maximizer $\widehat{\Theta}_{i,n}$ of $l_{i,n}(\theta)$ is equal to $n^{-1}\tilde{\Theta}_{i,n}$, where $\tilde{\Theta}_{i,n}$ is a maximizer of $l_{i,n}(n^{-1}\theta)$, which corresponds to the ordinary quasi-likelihood (i.e. with disregard for the actual value of $n$). 

\smallskip
Now we provide an estimator (up to a scaling factor) of the asymptotic variance $V_T = T^{-2} \int_0^T \Gamma(\theta_s^*)^{-1} ds$, which also requires no information on the value of $n$. For any $i \in \{ 1, \cdots, B_n \}$, we estimate the contribution of the $i$-th block by the formula  
\bea 
\widehat{C}_{i,n} := - \l[\partial_{\xi}^2 l_{i,n}\l(n^{-1}\xi\r)_{|\xi = n \widehat{\Theta}_{i,n}}\r]^{-1} .
\label{fisherEstim}
\eea
The term $\partial_{\xi}^2 l_{i,n}\l(n^{-1}\xi\r)$ doesn't depend on $n$ (when $h_n$ is chosen) and corresponds precisely to the Hessian matrix at point $\xi$ of the likelihood function of a Hawkes model when one disregards the value of $n$. In particular, this implies that $\widehat{C}_{i,n}$ can be computed. The asymptotic variance is then estimated, up to a scaling factor, as the weighted sum
\begin{eqnarray}
\label{AV} \widehat{C}_n = \frac{1}{B_n^2} \sum_{i=1}^{B_n} \widehat{C}_{i,n}. 
\end{eqnarray} 
The next proposition states the consistency of $n^{-1}\widehat{C}_n$ towards $V_T$ along with a corresponding studentized version of Theorem \ref{GCLT}, which is a corollary to the stable convergence in the GCLT. 
\begin{proposition*} \label{propStudent}
We have
$$ n^{-1}\widehat{C}_n \to^\proba V_T.$$
Moreover, we have the convergence in distribution
\begin{eqnarray}
\label{studentization}
 \widehat{C}_n^{-1/2}\l( n\widehat{\Theta}_n^{(BC)} - n\Theta \r) \rightarrow  \mathcal{N} (0,1).
 \end{eqnarray}
\end{proposition*}
Note that $\widehat{C}_n$ is the asymptotic variance of the dispersion between the estimated value of the scaled integrated parameter $n\widehat{\Theta}_n^{(BC)}$ and the target $n\Theta$ itself. In particular feasible asymptotic confidence intervals can be constructed from the data.

\smallskip
As the value of $n$ is unknown, which value to choose for $h_n$? One idea is to normalize the value of $\theta_t^*$ so that the expected number of events between 0 and $T$ is roughly one when parameters are equal to $\theta_t^*$ (by analogy with other models in high frequency data where $n$ corresponds exactly to the size of the sample data, as when estimating volatility from log-price returns observed regularly at times $iT/n$). This amounts to taking $n \approx N_T$ in practice. Although not perfect this provides guidance to the choice of $h_n$, which is assumed to be $n^{1/3} = o (h_n)$ and $h_n= o (n^{1/2})$ for a regular process ($\gamma =1$). In our numerical study, we have $N_T \approx 27,300$ which amounts to taking $n = 27,300$. This gives us $n^{1/3} \approx 30$ and $n^{1/2} \approx 165$, and correspondingly we look at different $h_n = 136.5, 273, 546$ which are of the same order. In our empirical study, we consider $h_n = \sqrt{n}, 2\sqrt{n}, 4\sqrt{n}, 8\sqrt{n}, 16\sqrt{n}$.

\section{Numerical simulations} \label{Simulations}
\subsection{Goal of the study} \label{goal}
In this section, we report the numerical results which assess the central limit theory of 
\begin{eqnarray*}
 Z_n^{(BC)} = \widehat{C}_n^{-1/2}\l(  n\widehat{\Theta}_n^{(BC)} - n\Theta \r)  \rightarrow  \mathcal{N} (0,1)
 \end{eqnarray*} 
 in a finite sample context for several time-varying parameter models. In addition, we report the behavior of the studentized naive estimator 
 \begin{eqnarray*}
 Z_n = \widehat{C}_n^{-1/2}\l(  n\widehat{\Theta}_n - n\Theta \r).
 \end{eqnarray*}
 Finally, we compare the performance of $\widehat{\Theta}_n^{(BC)}$ and $\widehat{\Theta}_n$ with two concurrent methods which are
\begin{enumerate}
\item The MLE on $[0,T]$ when considering that the parameters are not time-varying on $[0,T]$.
\item The time-varying baseline intensity MLE from \cite{chen2013inference} (CH) that assumes that $\nu_t^* = f (t, \theta)$ with $f$ being a polynomial of order $3$. More specifically in this setting the MLE estimates $(\theta,a,b)$ where $a$ and $b$ are assumed to be constant over time.
\end{enumerate}
The local log-likelihood functions and local variance estimators are computed implementing the formula obtained in \cite{ozaki1979maximum}. To compute $\widehat{\Theta}_n^{(BC)}$, we can either implement the function defined in (\ref{localMLEBC}) or carry out Monte-Carlo simulations to compute $b_{i,n} (\theta)$ for any $\theta$ prior to the numerical study. We choose the latter option as this allows to get also rid of bias terms which appears in the Taylor expansion in a higher order than $1$. Indeed, although those terms vanish asymptotically, they can pop up in a finite sample context. To be more specific, we first compute the sample mean for a grid of parameter values $\theta$ and a grid of block length $\Delta$ with 100,000 Monte Carlo paths of the parametric model, that we denote $b (\theta, \Delta)$. Then on each block, we estimate the bias by $b(\widehat{\Theta}_{i,n} , \Delta_{i,n} )$.

\subsection{Model design}
We consider that $T=21,600$ seconds, which corresponds to the period of activity for one working day from 10am to 4pm. The market events are chosen to correspond to trades. The $N_t^n$ process is generated using a time-varying version of the algorithm described in \cite{ozaki1979maximum} (Section 4, pp. 148-149). The integrated parameter is set to $n \Theta \approx (.8, 11, 30)$, which are comparable values to our empirical results. This yields on average $N_T^n \approx 27,300$ trades a day.

\smallskip
We consider three deterministic and one stochastic models for the time-varying parameter. The first two settings are toy models. Model I is a linear trend with $\theta_t^* = \theta^{(M)} + A (-1 +2 \frac{t}{T})$, where the non-random target value $\theta^{(M)} = (.8, 11, 30)$ and the amplitude is set to $A= (.5, 4, 10)$. This means that $\theta_t^*$ takes values in $\big[.3, 1.3 \big] \times \big[7, 15\big] \times \big[20, 40\big]$, which is comparable to the daily variation in our empirical results. In Model II $\theta_t^*$ oscillates around  $\theta^{(M)} = (.8, 11, 30)$ and has the form $\theta_t^* = \theta^{(M)} + A \cos ( \frac{t}{T} 2 \pi)$, in particular implying that the range of taken values is the same as in Model I.

\smallskip
Model III is taken directly from the literature. We keep $a_t^*$ and $b_t^*$ constant whereas $\nu_t^*$ follows a usual intraday pattern, so that CH is well specified for this model. As pointed out in \cite{engle1998autoregressive} (see discussions in Section 5-6 and Figure 2), the expected duration before the next trade tends to follow a U-shape intraday pattern. This diurnal effect motivated \cite{chen2013inference} (see Section 5, pp. 1011-1017) to model a Hawkes process where $\nu_t^*$ is time-varying with a quadratic form. The model is written as $\nu_t^* = e^{\beta_1} + \{ e^{\beta_2} + e^{\beta_3} \}^2 (t/T - e^{\beta_2}/(e^{\beta_2} + e^{\beta_3}))^2$.  We fit the model to the empirical intraday mean and find $\beta_1 \approx -.84$, $\beta_2 \approx -.26$ and $\beta_3 \approx -.39$, which implies that $T^{-1} \int_0^T \nu_t^* dt \approx .61$. The other two parameters $(a_t^*, b_t^*) = (11, 30)$ are assumed to be constant.

\smallskip
Model IV is an extension of Model III based on more realistic considerations where $a_t^*$ and $b_t^*$ also feature intraday seasonality. In addition, we allow for additive stochastic component in the three parameters . We assume that 
\begin{eqnarray*}
\nu_t^* & = & e^{\beta_1^{\nu}} + \{ e^{\beta_2^{\nu}} + e^{\beta_3^{\nu}} \}^2 (t/T - e^{\beta_2^{\nu}}/(e^{\beta_2^{\nu}} + e^{\beta_3^{\nu}}))^2 + \sigma^{\nu} \tilde{W}_t^{\nu}, \\
a_t^* & = & e^{\beta_1^{a}} + \{ e^{\beta_2^{a}} + e^{\beta_3^{a}} \}^2 (t/T - e^{\beta_2^{a}}/(e^{\beta_2^{a}} + e^{\beta_3^{a}}))^2 + \sigma^{a} \tilde{W}_t^{a}, \\
b_t^* & = & e^{\beta_1^{b}} + \{ e^{\beta_2^{b}} + e^{\beta_3^{b}} \}^2 (t/T - e^{\beta_2^{b}}/(e^{\beta_2^{b}} + e^{\beta_3^{b}}))^2 + \sigma^{b} \tilde{W}_t^{b}, 
\end{eqnarray*}
where $(\beta_1^{\nu}, \beta_2^{\nu}, \beta_3^{\nu}) \approx (-0.84, -0.26, -0.39)$, $(\beta_1^{a}, \beta_2^{a}, \beta_3^{a}) \approx (2.35, -0.05, 0.40)$, and $(\beta_1^{b}, \beta_2^{b}, \beta_3^{b})$ $\approx (3.66,$ $-0.33, 0.67)$ were obtained when fitted to the respective parameter intraday mean. We also set $(\sigma^{\nu}, \sigma^{a}, \sigma^{b}) = (0.8/(6T^{1/2}), 11/(6T^{1/2}), 30/(6T^{1/2}))$ and $\tilde{W}_t = (\tilde{W}_t^{\nu}, \tilde{W}_t^{a}, \tilde{W}_t^{b}) = \int_{t-1}^t W_s ds$ with $W_t = (W_t^{\nu}, W_t^{a}, W_t^{b})$ a standard 3-dimensional Brownian motion. This means that the standard deviation of the noise factor is roughly equal to $1/6$ the value of the parameter at time $T$. Also, we cap the possible value taken by $\nu_t^*$ so that the intensity parameter stays bigger than $0.2$. It is clear that in all those models the parameter is smooth enough to satisfy the assumptions of the GCLT.

\smallskip
Finally, we look at several values for $h_n = 136.5, 273, 546$, which correspond respectively to block lengths of size $Th_n/n=$ 108, 216 and 432 seconds. We have $N_T \approx 27,300$ which means that we should take $n = 27,300$ as explained in Section 
\ref{statisticalimplementation}. According to the theory we would expect to need that $h_n$ is of the same order as $n^{1/3}$ and  $n^{1/2}$, which are approximately equal to 30 and 165, thus our choices for $h_n$ seem coherent with the theory.

\subsection{Results}
Table \ref{montecarlonaive} shows the Monte Carlo results for the feasible statistic of the naive estimator. For all the models, the value of the bias is striking as it is of the same magnitude as the standard deviation. This indicates that the bias do play a crucial role in finite sample too.

\smallskip
Table \ref{montecarloBC} shows the result for the bias-corrected estimator. In addition, Figure \ref{QQplot} provides the associated QQ-plot. In this case, the sample mean is very close to 0 indicating that our proposed reduction method is working well. The standard deviation obtained for the intensity parameter is close to 1, but it is bigger for the other parameters $a^*$ and $b^*$. Correspondingly, the asymptotics are slightly underestimating the mass of the distribution in the tails. The reason for this is probably that it is difficult to accurately estimate the variance of the parametric model on the small blocks.

\smallskip
Table \ref{performance} shows the performance of the estimators with concurrent approaches. It is clear that regardless of the model at hand the bias-corrected local approach performs better than the MLE and the CH. In Model III, CH performs a MLE with no misspecification. Although CH performs much better than in the other models, it still doesn't outperform the local approach, which indicates that the local approach can performs better even on standard parametric models. Both estimators are badly biased in case of misspecification for them. More surprisingly, although Model III follows a model included in  \cite{chen2013inference} and thus CH performs a MLE with no misspecification in that specific case, the estimates are still biased (although very good for $a^*$ and with a smaller standard deviation).  

\smallskip
Finally, we can see that the naive estimator is more biased when $h_n$ is smaller, which is in line with what we expected. The bias-correction is performed better with bigger $h_n$, although too big of a $h_n$ will tend to bias-correct less efficiently (due to the fact that parameters are moving too much on a bigger block). This can be seen in Table \ref{performance} as the bias-corrected estimator seems to perform slightly better with 4 minute block than 7 minute block.

\begin{table}
\centering
\caption{Finite sample properties of $Z_n$ for several models$^\dag$} 
\label{montecarlonaive}
\begin{tabular}{@{}rccccccccc@{}}
\toprule
\toprule
\multicolumn{1}{l}{Param.} & \multicolumn{1}{l}{Mean} & \multicolumn{1}{l}{Stdv.} & \multicolumn{1}{l}{RMSE} & \multicolumn{1}{l}{0.5\%} & \multicolumn{1}{l}{2.5\%} & \multicolumn{1}{l}{5\%} & \multicolumn{1}{l}{95\%} & \multicolumn{1}{l}{97.5\%} & \multicolumn{1}{l}{99.5\%} \\ \toprule
\multicolumn{10}{l}{Model I} \\
$\nu^*$ & 0.69& 1.02 & 1.23 & 0.00 & 0.40 & 1.00 & 82.60 & 89.70 & 97.00\\
$a^*$ & 0.77 & 1.10 & 1.34 & 0.40 & 1.00 & 1.50 & 78.00 & 86.50 & 95.90\\
$b^*$ & 1.37 & 1.24 & 1.85 & 0.20 & 0.60 & 1.00 & 58.20 & 67.60 & 83.90\\
\multicolumn{10}{l}{Model II} \\
$\nu^*$ & 0.71& 1.02&1.24&0.00&0.60&1.10&81.20&89.40&97.60\\
$a^*$ & 0.71 & 1.14 & 1.34 & 0.30 & 1.40 & 2.10 & 79.50 & 86.70 & 94.30\\
$b^*$ & 1.33 & 1.27 & 1.83 & 0.10 & 0.40 & 1.10 & 60.00 & 69.60 & 85.40\\
\multicolumn{10}{l}{Model III} \\
$\nu^*$ & 0.80 & 1.05 & 1.32 & 0.10 & 0.30 & 0.90 & 78.40 & 86.30 & 95.30\\
$a^*$ & 0.78 & 1.14 & 1.38 & 0.00 & 1.20 & 1.70 & 78.10 & 85.70 & 94.10\\
$b^*$ & 1.43 & 1.24 & 1.89 & 0.00 & 0.20 & 0.40 & 55.70 & 66.10 & 80.60\\
\multicolumn{10}{l}{Model IV} \\
$\nu^*$ & 0.83 & 0.99 & 1.29 & 0.00 & 0.20 & 0.50 & 79.70 & 87.10 & 96.30  \\
$a^*$ & 1.05 & 1.04 & 1.48 & 0.00 & 0.60 & 0.90 & 71.60 & 80.90 & 93.00\\
$b^*$ & 1.62 & 1.10 & 1.96 & 0.00 & 0.10 & 0.20 & 52.20 & 63.20 & 79.70\\

\bottomrule

\end{tabular}

\scriptsize $^\dag$This table shows summary statistics and empirical quantiles benchmarked to the $N$(0,1) distribution for the feasible Z-statistics related to the naive  estimator with $h_n=273$ (which corresponds to a 216 second block length). The simulation design is Model I-IV with $M = 1,000$ Monte-Carlo simulations. 
\end{table}

\begin{table}
\centering
\caption{Finite sample properties of $Z_n^{(BC)}$ for several models$^\dag$} 
\label{montecarloBC}
\begin{tabular}{@{}rccccccccc@{}}
\toprule
\toprule
\multicolumn{1}{l}{Param.} & \multicolumn{1}{l}{Mean} & \multicolumn{1}{l}{Stdv.} & \multicolumn{1}{l}{RMSE} & \multicolumn{1}{l}{0.5\%} & \multicolumn{1}{l}{2.5\%} & \multicolumn{1}{l}{5\%} & \multicolumn{1}{l}{95\%} & \multicolumn{1}{l}{97.5\%} & \multicolumn{1}{l}{99.5\%} \\ \toprule
\multicolumn{10}{l}{Model I} \\
$\nu^*$ & -0.01 & 1.02 & 1.02 & 0.40 & 2.80 & 5.70 & 94.70 & 97.20 & 99.30\\ 
$a^*$ & 0.00 & 1.12 & 1.12 & 1.10 & 4.00 & 7.60 & 93.50 & 96.70 & 99.10\\
$b^*$ & -0.02 & 1.30 & 1.30 & 2.60 & 6.60 & 10.70 & 90.10 & 93.40 & 98.60\\
\multicolumn{10}{l}{Model II} \\
$\nu^*$ & 0.02&1.02&1.02& 0.90 & 2.70 & 5.20 & 95.70 & 98.10 & 99.40\\
$a^*$ & -0.07 & 1.16 & 1.16 & 2.00 & 5.00 & 9.00 & 92.10 & 95.70 & 99.00\\
$b^*$ & -0.08 & 1.32 & 1.33 & 4.20 & 7.60 & 11.50 & 90.50 & 94.30 & 98.20\\
\multicolumn{10}{l}{Model III} \\
$\nu^*$ & 0.00 & 1.05 & 1.05 & 0.10 & 2.90 & 6.20 & 94.30 & 96.90 & 99.50\\
$a^*$ & -0.02 & 1.15 & 1.15 & 1.90 & 4.80 & 8.80 & 91.00 & 95.90 & 99.00 \\
$b^*$ & -0.06 & 1.29 & 1.30 & 4.00 & 7.30 & 11.00 & 90.80 & 94.50 & 98.30\\
\multicolumn{10}{l}{Model IV} \\
$\nu^*$ & 0.07 & 0.99 & 1.00 & 0.30 & 2.30 & 4.60 & 94.20 & 97.60 & 99.70\\
$a^*$ & -0.04 & 1.05 & 1.05 & 1.20 & 3.50 & 5.40 & 94.90 & 97.30 & 99.30\\
$b^*$ & -0.07 & 1.15 & 1.15 & 1.40 & 5.70 & 9.10 & 92.00 & 95.80 & 99.30\\
\bottomrule

\end{tabular}

\scriptsize $^\dag$This table shows summary statistics and empirical quantiles benchmarked to the $N$(0,1) distribution for the feasible Z-statistics related to the bias-corrected estimator with $h_n=273$ (which corresponds to a 216 second block length). The simulation design is Model I-IV with $M = 1,000$ Monte-Carlo simulations. 
\end{table}

\begin{figure}
\includegraphics[width=1\linewidth]{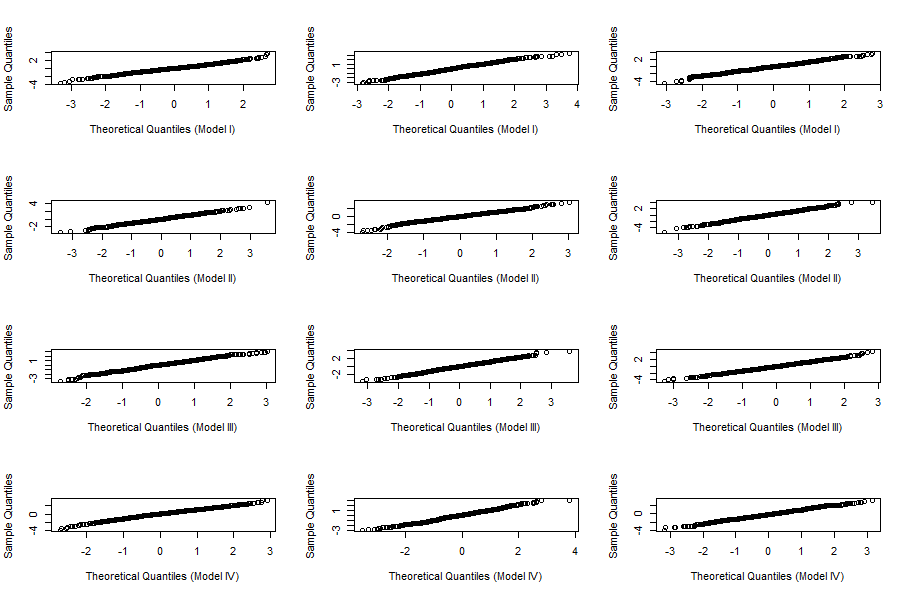}
\centering
\caption{QQ-plots benchmarked to the $N$(0,1) distribution for the feasible Z-statistics related to the bias-corrected estimator with $h_n=273$ (which corresponds to a 216 second block length). The left column corresponds to $\nu_t^*$, the column in the middle for $a_t^*$ and the right column for $b_t^*$. The simulation design is Model I-IV with $M = 1,000$ Monte-Carlo simulations. }
\label{QQplot}
\end{figure}

\begin{table}
\centering
\caption{Performance of 8 estimators for several models$^\dag$} 
\label{performance}
\begin{tabular}{@{}lcccccc@{}}
\toprule
\toprule
\multicolumn{2}{l}{} & \multicolumn{2}{l}{$\nu^*$} & \multicolumn{2}{l}{$a^*$} & \multicolumn{1}{l}{$b^*$} \\ \toprule
\multicolumn{1}{l}{Est.} & \multicolumn{1}{l}{Mean} & \multicolumn{1}{l}{Stdv.} &  \multicolumn{1}{l}{Mean} & \multicolumn{1}{l}{Stdv.} &  \multicolumn{1}{l}{Mean} & \multicolumn{1}{l}{Stdv.} \\ \toprule
\multicolumn{7}{l}{Model I} \\
naive 2m& 0.009 & 0.007 & 0.286 & 0.196 & 1.239 & 0.630\\
BC 2m& 0.000 & 0.007 & -0.012 & 0.201 & -0.124 & 0.658\\
naive 4m& 0.005 & 0.007 & 0.134 & 0.189 & 0.546 & 0.495 \\
BC 4m& 0.000 & 0.007 & 0.002 & 0.192 & 0.005 & 0.503\\
naive 7m& 0.002 & 0.007 & 0.068 & 0.188 & 0.271 & 0.477\\
BC 7m& 0.000 & 0.007 & 0.004 & 0.189 & 0.010 & 0.481\\
MLE & -0.011 & 0.006 & 0.489 & 0.198 & 0.485 & 0.494\\
CH & 0.018 & 0.010 & 0.424 & 0.438 & 1.378 & 0.942\\
\multicolumn{7}{l}{Model II} \\
naive 2m& 0.009 & 0.007 & 0.287 & 0.213 & 1.346 & 0.734\\
BC 2m& 0.000 & 0.007 & -0.018 & 0.218 & -0.078 & 0.744\\
naive 4m& 0.005 & 0.007 & 0.126 & 0.198 & 0.538 & 0.516\\
BC 4m& 0.000 & 0.007 & -0.009 & 0.201 & -0.019 & 0.525\\
naive 7m & 0.002 & 0.007 & 0.057 & 0.196 & 0.245 & 0.490\\
BC 7m & 0.000 & 0.007 & -0.009 & 0.197 & -0.022 & 0.494\\
MLE & -0.017 & 0.006 & 0.708 & 0.214 & 0.666 & 0.516 \\
CH & -0.063 & 0.012 & 0.294 & 0.443 & -0.265 & 1.097\\
\multicolumn{7}{l}{Model III} \\
naive 2m& 0.009 & 0.006 & 0.348& 0.235& 1.474& 0.734\\
BC 2m& 0.000 & 0.006 & -0.009 & 0.241 & -0.108 & 0.742\\
naive 4m& 0.005 & 0.005 & 0.158 & 0.227 & 0.645 & 0.568 \\
BC 4m& 0.000 & 0.005 & -0.003 & 0.241 & -0.015 & 0.578\\
naive 7m& 0.002 & 0.006 & 0.074 & 0.225 & 0.316 & 0.543\\
BC 7m& 0.000 & 0.006 & -0.004 & 0.227 & -0.020 & 0.548\\
MLE & -0.004 & 0.006 & -0.081 & 0.220 & -0.572 & 0.519\\ 
CH & -0.009 & 0.005 & -0.002 & 0.213 & -0.082 & 0.454\\
\multicolumn{7}{l}{Model IV} \\
naive 2m& 0.009 & 0.006 & 0.624 & 0.332 & 3.533 & 2.579\\ 
BC 2m & 0.000 & 0.006 & -0.011 & 0.311 & -0.199 & 2.333\\
naive 4m & 0.005 & 0.006 & 0.276 & 0.274 & 1.389 & 1.022 \\
BC 4m & 0.000 & 0.006 & -0.006 & 0.270 & -0.008 & 0.967\\
naive 7m & 0.003 & 0.006 & 0.133 & 0.265 & 0.655 & 0.889\\
BC 7m & 0.000 & 0.006 & 0.000 & 0.265 & 0.013 & 0.876\\
MLE & -0.004 & 0.006 & 0.005 & 0.286 & 0.924 & 0.959\\
CH & -0.012 & 0.009 & -0.857 & 0.699 &  -3.802 & 2.430\\
\bottomrule

\end{tabular}

\scriptsize $^\dag$This table shows the statistic $\widehat{\Theta} - \Theta$ where $\widehat{\Theta}$ is equal to the naive estimator and the bias-corrected estimator  with $h_n=136.5, 273, 546$ (this corresponds respectively to a 108 second block (roughly 2 minutes), 216s (roughly 4m) and 432s (roughly 7m)), the MLE and the CH. The simulation design is Model I-IV with $M = 1,000$ Monte-Carlo simulations. 
\end{table}

\section{Empirical study}
In this section, we implement local MLE on intraday transaction (corresponding to trade) times of Apple (APPL) shares carried out on the NASDAQ in 2015. Our aim is twofold. First, using relatively large (30 minute) local blocks, we document about seasonality and intraday variability in the parameters. Second, we implement the naive and the bias-corrected estimator. We
exclude January 1, the day after Thanksgiving and December 24 which are less active. This leaves us with 251 trading days of data. To prevent from opening and closing effect, we consider transactions that were carried out between 9:30am and 3:30pm, which corresponds to 5 full hours of trading. The number of daily trades is on average 15,000 with more than 50,000 trades for the most active days and slightly more than 3,000 for the least active days.

\smallskip
In Figure \ref{plotnuintraday}-\ref{plotbintraday}, we document the intraday variation of the three parameters. To do that, we divide the 5 hours of trading into 10 blocks of 30 minutes. On each block, we fit the MLE and obtain the corresponding estimates. We also estimate the standard deviation, which allows us to build 95\% confidence intervals. Given how volatile the estimates are with respect to their own confidence interval, it is clear that neither the parametric model nor the seasonal component model can be satisfactory to fit such data. This time-varying tendency of parameter intraday values was consistently observed across most of the trading days in 2015. The behavior is heterogeneous in the three parameters. The seasonal model seems to do a decent job for the intensity parameter\footnote{probably even better if we add a "day effect" in the model} although the shape of the parameter is very particular for each different day. The seasonal tendency is less clear for the other two parameters. $a^{*}$ tends to oscillate not too far around the seasonal path with a behavior which is day specific, whereas $b^{*}$ can really go far off from one side or the other with no specific pattern. For all those reasons we believe that including in the model both a seasonal and a stochastic effect is more realistic.

\begin{figure}
\includegraphics[width=1\linewidth]{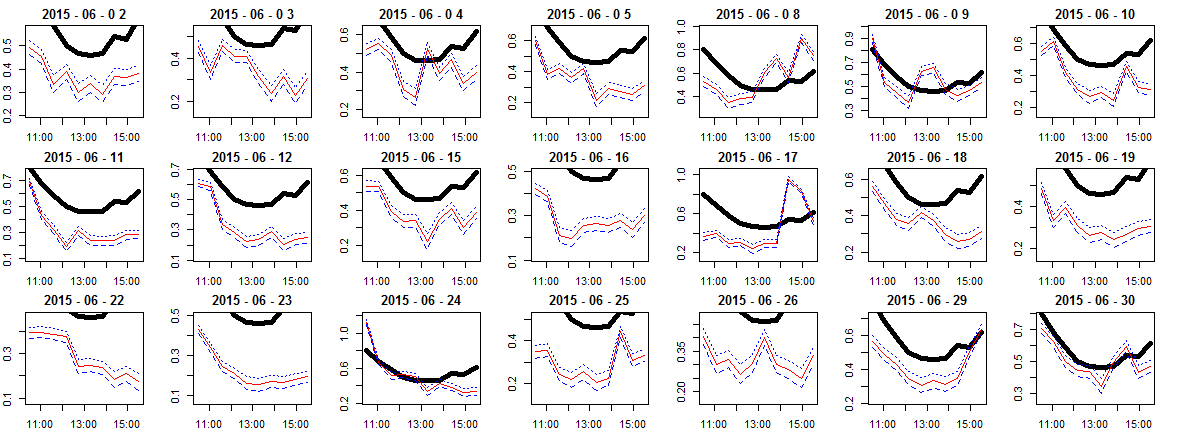}
\centering
\caption{Local estimated $\nu^*$ parameter on 30 minute long blocks in June 2015. The two dashed lines correspond to the 95\% confidence intervals. The thick line stands for the seasonality intraday effect, estimated as a temporal local mean across all the trading days in 2015. }
\label{plotnuintraday}
\end{figure}

\begin{figure}
\includegraphics[width=1\linewidth]{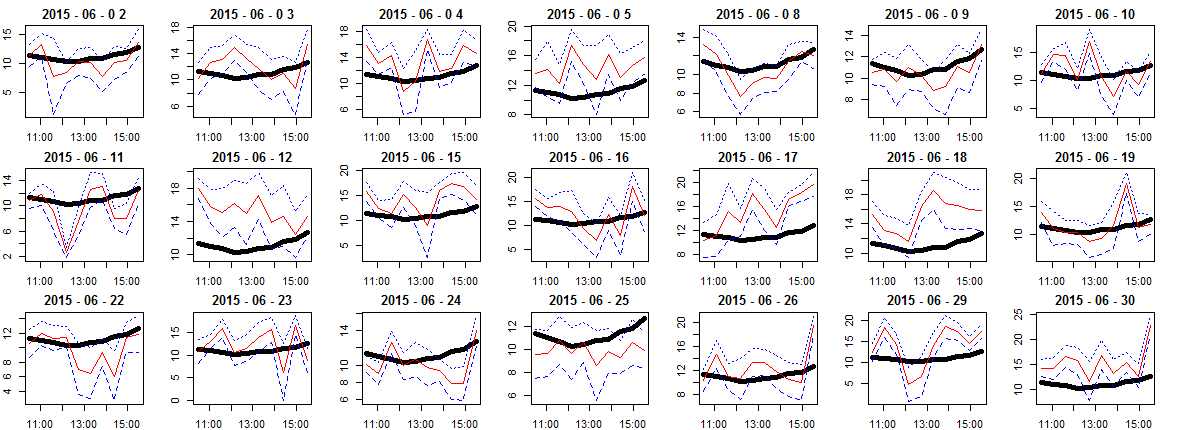}
\centering
\caption{Local estimated $a^*$ parameter on 30 minute long blocks in June 2015. The two dashed lines correspond to the 95\% confidence intervals. The thick line stands for the seasonality intraday effect, estimated as a temporal local mean across all the trading days in 2015.}
\label{plotaintraday}
\end{figure}

\begin{figure}
\includegraphics[width=1\linewidth]{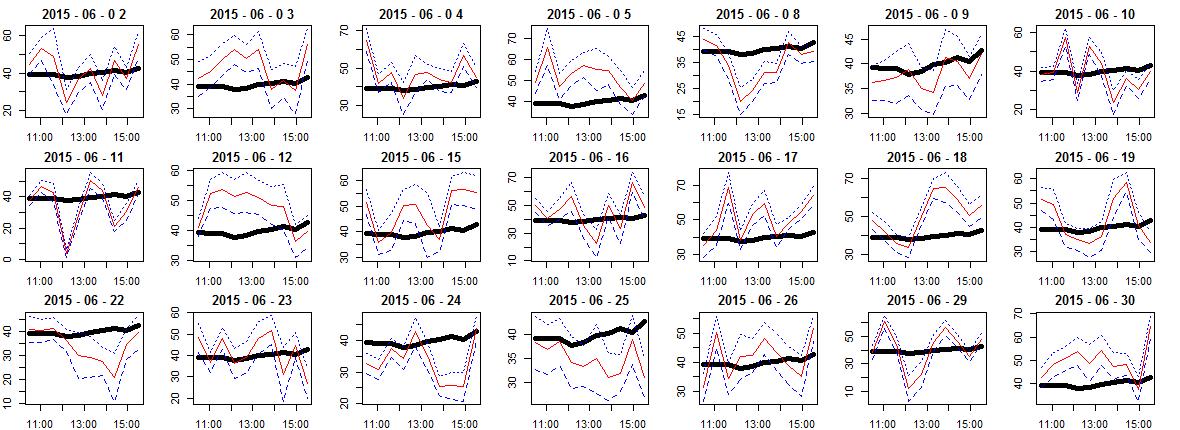}
\centering
\caption{Local estimated $b^*$ parameter on 30 minute long blocks in June 2015. The two dashed lines correspond to the 95\% confidence intervals. The thick line stands for the seasonality intraday effect, estimated as a temporal local mean across all the trading days in 2015.}
\label{plotbintraday}
\end{figure}

\smallskip
In Table \ref{empwork}, we report statistics of the implemented estimators. As our method is non parametric, the assumption of any particular parametric model for the time-varying parameter is not required. Overall we find that the daily estimates are on average roughly equal to $(0.56, 11, 40)$, with a standard deviation around $(0.24, 2, 8)$. The results are in line with the numerical study. We implemented five levels of $h_n = \sqrt{n}, 2 \sqrt{n}, 4 \sqrt{n}, 8 \sqrt{n},$ $16 \sqrt{n}$ that we denote respectively the corresponding bias-corrected estimators BC 1-5. We can see that BC 2-4 are highly correlated, whereas BC 1 and BC 5 are slightly less correlated. This is probably due to the fact that $h_n$ can be too small on non active days in the case of BC 1 and too big when considering BC 5. This shows that the local method seems robust to a wide range of possible tuning parameter $h_n$. Furthermore, the mean of the MLE and the CH are very different from the one of BC. This is most likely explained by the strong bias obtained in our numerical study. Among those two estimators it is not surprising to find that the MLE is more in line with the local estimates than the CH as the MLE is a "local estimate" in the degenerate case $h_n = n$. 
\begin{table}
\centering
\caption{Summary statistics for 12 estimators$^\dag$} 
\label{empwork}
\begin{tabular}{@{}lccccccccc@{}}
\toprule
\toprule
\multicolumn{2}{l}{} & \multicolumn{3}{l}{$\nu^*$} & \multicolumn{3}{l}{$a^*$} & \multicolumn{2}{l}{$b^*$} \\ \toprule
\multicolumn{1}{l}{Est.} & \multicolumn{1}{l}{Mean} & \multicolumn{1}{l}{Stdv.} & \multicolumn{1}{l}{Corr.(,BC 3)} & \multicolumn{1}{l}{Mean} & \multicolumn{1}{l}{Stdv.} & \multicolumn{1}{l}{Corr.(,BC 3)} & \multicolumn{1}{l}{Mean} & \multicolumn{1}{l}{Stdv.} & \multicolumn{1}{l}{Corr.(,BC 3)} \\ \toprule
naive 1& 0.57 & 0.24 & $\approx$ 1 & 12.48 & 2.18 & 0.94 & 49.80 & 10.33 & 0.78\\
BC 1& 0.56 & 0.24 & $\approx$ 1 & 11.16 & 2.17 & $\approx$ 1 & 40.24 & 9.54 & 0.98\\
naive 2& 0.57 & 0.24 & $\approx$ 1 & 11.78 & 2.00 & 0.98 & 44.74 & 8.20 & 0.89\\
BC 2& 0.56 & 0.24 & $\approx$ 1 & 11.12 & 2.00 & $\approx$ 1 & 39.92 & 7.79 & 0.99\\
naive 3& 0.56 & 0.24 & $\approx$ 1 & 11.42 & 1.96 & 0.99 & 42.14 & 7.58 & 0.92 \\
BC 3& 0.56 & 0.24 & 1 & 11.10 & 1.96 & 1 & 40.16 & 7.48 & 1\\
naive 4& 0.56 & 0.24 & $\approx$ 1 & 11.22 & 1.96 & $\approx$ 1 & 40.66 & 7.40 & 0.96\\
BC 4& 0.56 & 0.24 & $\approx$ 1 & 11.07 & 1.96 & $\approx$ 1 & 39.72 & 7.36 & 0.94\\
naive 5& 0.56 & 0.24 & $\approx$ 1 & 11.07 & 1.97 & 0.99 & 39.60 & 7.50 & 0.94  \\
BC 5 & 0.56 & 0.24 & $\approx$ 1 & 11.01 & 1.97 & 0.98 & 39.14 & 7.50 & 0.92\\
MLE & 0.55 & 0.23 & $\approx$ 1 & 10.78 & 2.11 & 0.95 & 36.77 & 8.55 & 0.91\\
CH & 0.55 & 0.22 &  0.99 & 11.69 & 1.44 & 0.63 & 40.50 & 4.85 & 0.63\\
\bottomrule

\end{tabular}

\scriptsize $^\dag$Sample mean, standard deviation and correlation with BC 3 for the naive estimators (naive 1-5) and the bias-corrected estimators (BC 1-5) with respectively $h_n=\sqrt{n}, 2 \sqrt{n}, 4 \sqrt{n}, 8 \sqrt{n}, 16 \sqrt{n}$, the MLE and the CH implemented for APPL in 2015. 
\end{table}

\section{Conclusion}
We have introduced a time-varying parameter extension of the Hawkes process with exponential exciting function. We have also provided an estimator, along with its central limit theorem, of the integrated parameter. We have seen on numerical simulations that this is of particular interest to the practitioner because some concurrent methods (e.g. MLE applied to all the observations) are biased. Finally, our empirical study points out the possible presence of variability in the parameter in addition to seasonal effects. 

\smallskip
There are some questions left to explore such as what would happen to the local MLE in the case of a kernel with a fatter tail, such as a polynomial decreasing kernel. As far as the authors know, no convergence of moments of the rescaled MLE has been investigated even in the parametric case. Also, optimality of the tuning parameter $h_n$ could be investigated, and we could potentially allow for time-varying tuning parameter.

\smallskip
Finally, we point out that the method can be extended to estimate more general key quantities than the integrated parameter, such as functional of the parameter  $T^{-1}\int_0^T{f_s(\theta_s^*) ds}$. In particular, the GCLT for weighted versions $T^{-1}\int_0^T{\theta_s^*w_sds}$ where $w_s$ is a weight process chosen by the practitioner may be derived by a similar reasoning.   

\section{Appendix} \label{Annex}

\subsection{The standard MLE for the parametric Hawkes process} \label{sectionMLEclassic}

We briefly introduce the standard maximum likelihood estimation procedure for the parametric Hawkes process with exponential kernel $\phi_t = ae^{-bt}$ in the long run (also called low-frequency) asymptotics, that is when we consider observations of a Hawkes process $N^P$ on the time interval $[0,T]$ with $T \to \infty$. We define several deterministic key quantities, such as the Fisher information matrix, as time average limits of quantities which depend on the point process $N_t^{P}$. 

\smallskip
The regression family is defined for each $\theta \in K$ as 
\bea 
\lambda(t,\theta) = \nu + \int_0^{t-}{ae^{-b(t-s)}dN_s^P}.
\eea
We assume that there exists an unknown parameter $\theta^* \in K$ such that the $\calf_t^{N^P}$-intensity of $N_t^P$ is expressed as 
\bea 
\lambda_*^{P}(t) = \lambda(t,\theta^*).
\eea
The log-likelihood process is, up to a constant term,
\bea
l_T(\theta) = \int_0^T{\text{log}\l(\lambda(t,\theta)\r)dN_t^P} - \int_0^T{\lambda(t,\theta)dt}.
\eea
The MLE $\hat{\theta}_T$ is a maximizer of $l_T(\theta)$. We define 
\bea 
\Gamma_T(\theta^*) = - \inv{T} \partial_\theta^2 l_T(\theta^*) \in \reels^{3 \times 3},
\eea
\bea 
K_T(\theta^*) = \inv{T} \partial_\theta^3 l_T(\theta^*) \in \reels^{3 \times 3 \times 3},
\eea
\bea 
M_T(\theta^*) = \int_0^T{\frac{\partial_\theta \lambda(t,\theta^*)}{\lambda(t,\theta^*)}\{dN_t^P - \lambda(t,\theta^*)dt\}}, 
\eea 
and for any indices $k,l,m \in \{0,1,2\}$,
\bea
C_T(\theta^*)_{k,lm} = \inv{T}  \int_{0}^{T}{\partial_{\theta,k}\lambda(t,\theta^*) \partial_{\theta,lm}^2 \text{log} \l(\lambda(t,\theta^*)\r)dt},
\eea  
and 
\bea 
Q_T(\theta^*)_{k,lm} = - \frac{M_T(\theta^*)_k }{T}  \int_{0}^{T}{\frac{\partial_\theta \lambda(t,\theta^*)_l\partial_\theta \lambda(t,\theta^*)_m}{\lambda(t,\theta^*)}dt }.
\eea 
 The three time-averaged quantities $\Gamma_T$, $K_T$ and $C_T$ admit deterministic limiting values when $T \to \infty$ because the process $N^P$ is exponentially mixing. Indeed, a slight generalization of Lemma 6.6 in \cite{Clinet20171800} shows that the vector process $(\lambda(t,\theta^*), \partial_\theta(t,\theta^*),$ $\cdots, \partial_\theta^3(t,\theta^*))$ satisfies the mixing condition [M2] defined on p. 14 in the cited paper, which in turn implies the existence of $\Gamma(\theta^*) \in \reels^{3 \times 3}$, and $K(\theta^*)$, $C(\theta^*) \in \reels^{3 \times 3 \times 3}$ such that for any $\epsilon \in (0,1)$ and any integer $p \geq 1$,
\bea 
\esp \l|\Gamma_T(\theta^*) -  \Gamma(\theta^*) \r|^p = O\l(T^{-\epsilon\frac{p}{2}}\r),
\label{paramFisher}
\eea
\bea 
\esp \l|K_T(\theta^*) -  K(\theta^*) \r|^p = O\l(T^{-\epsilon\frac{p}{2}}\r),
\label{paramK}
\eea
and
\bea 
\esp \l|C_T(\theta^*) -  C(\theta^*) \r|^p = O\l(T^{-\epsilon\frac{p}{2}}\r),
\label{paramC}
\eea
where $|x|$ stands for $\sum_i |x_i|$ for any vector or a matrix $x$. Moreover, it is also an easy consequence of the mixing property along with the fact that $M_T(\theta^*)$ is a martingale that we have the convergence 
\bea 
\esp\l[Q_T(\theta^*) - Q(\theta^*)\r] = O\l(T^{-\frac{\epsilon}{2}}\r),
\label{paramQ}
\eea 
for some $Q(\theta^*) \in \reels^{3 \times 3 \times 3}$. Note that $\Gamma(\theta^*)$ is the asymptotic Fisher information. In particular, in \cite{Clinet20171800} the authors have shown the convergence of moments of the MLE (see Theorem 4.6), 
\bea 
\esp\l[ f\l(\sqrt T (\hat{\theta}_T - \theta^*)\r)\r] \to \esp\l[f\l(\Gamma(\theta^*)^{-\half} \xi\r)\r],
\label{paramConvMoment}
\eea
where $f$ can be any continuous function of polynomial growth, and $\xi$ follows a standard normal distribution. Also, it is easy to see that the convergences in (\ref{paramFisher})-(\ref{paramConvMoment}) hold uniformly in $\theta^* \in K$ under a mild change in the proofs of \cite{Clinet20171800}. The result (\ref{paramConvMoment}) should be compared to Theorem \ref{conditionalCLT}.  Finally, from $\Gamma$, $K$, $C$ and $Q$ we define for any $k \in \{0,1,2\}$
\bea 
b(\theta^*)_k = \half \Gamma(\theta^*)^{jk}\Gamma(\theta^*)^{lm}(K(\theta^*)_{jlm}+ 2\l\{C(\theta^*)_{l,jm}+Q(\theta^*)_{l,jm}\r\})
\label{paramb}
\eea 
with implicit summation of repeated indices. The function $b$ appears in the expression of the expansion of the bias of the local MLE in Section \ref{sectionBiasCorrection}. 

\subsection{Construction of the doubly stochastic Hawkes process}
We establish the existence of the doubly stochastic self-exciting process under very general conditions on the parameter process. We also provide the boundedness of moments of various stochastic integrals with respect to such point process when the parameter is assumed to take its values in a compact space. We follow the same procedure as in \cite{BremaudStability1996} for the construction of a Hawkes process, that is, we show the existence of the doubly stochastic Hawkes process by a fixed point argument. In what follows we let $\calb=(\Omega,\calf,\F,\proba)$, $\F=(\calf_t)_{t\in [0,T]}$, $\calf = \calf_T$ be a stochastic basis such that the filtration $\F$ is generated by the three-dimensional predictable process $(\theta_s)_{s \in [0,T]} = (\nu_s,a_s,b_s)_{s \in [0,T]}$ which is component-wise non-negative, and by a Poisson process $\overline{N}$ of intensity $1$ on $\reels^2$ which is independent of $\theta$. In other words, $\calf_t = \calf_t^{(\theta, \overline{N} )}$. In the following, properties such as predictability or adaptivity will automatically refer to $\F$. Before we turn to the existence of the self-exciting doubly stochastic process, we recall a key result for martingales.

\begin{lemma*}
 Let $\F = (\calf_t)_{t \in [0,T]}, \calf = \calf_T$ be a filtration and $\calg$ a $\sigma$-field that is independent of $\calf$. Consider also the extended filtration defined by $\calh_t = \calf_t \vee \calg$. Then any square integrable $\calf_t$-martingale $M$ is also a $\calh_t$-martingale. In particular, for any $\calh_t$-predictable process $u$ such that $\int_0^T{u_s d\langle M,M\rangle_s}$ is integrable, $\esp[\int_0^T{u_s dM_s} | \calg ] = 0$. 
 \label{lemmaMartingale}
\end{lemma*}

\begin{proof}
Let $M$ defined as in the lemma and write for $0 \leq s \leq t \leq T$,
\beas
\esp[M_t | \calh_s] &=& \esp[M_t | \calf_s \vee \calg] \\
&=& \esp[M_t | \calf_s] \\
&=& M_s,
\eeas
since $\calg \ind M_t$ and $\calg \ind \calf_s$. It follows that $\int_0^t{u_s dM_s}$ is a $\calh_t$-martingale, the second part of the lemma follows.
\end{proof}

We now show the existence of the doubly stochastic Hawkes process associated to $\theta$.

\begin{proof}[Proof of Theorem \ref{lemmaExistence}]
We apply a fixed point argument using integrals over the two-dimensional integer measure $\overline{N}(dt,dx)$. Let us first define $\lambda^0(t) = \nu_t$ and $N^0$ the point process defined as 
\bea 
N_t^0 = \iint_{[0,t] \times \reels}{\mathbb{1}_{\{0 \leq  x \leq \lambda^0(s) \}} \overline{N}(ds,dx)}.
\eea
It is immediate to see that $\lambda^0(t)$ is the $\calf_t$-intensity of $N_t^0$. We then define recursively the sequence of $\calf_t$-adapted point processes $N^n$ along with their stochastic intensities $\lambda^n$ as 
\begin{subequations}
\begin{align}
        \label{lambdan1} \lambda^{n+1}(t) &= \nu_t + \int_{0}^{t-}{a_s e^{-b_s(t-s)}dN_s^n},\\
        N_t^{n+1} &= \iint_{[0,t] \times \reels}{\mathbb{1}_{\{0\leq x \leq \lambda^{n+1}(s)\}} \overline{N}(ds,dx)}.
\end{align} 
\label{systemHawkes}
\end{subequations}
Note that both $\lambda^n$ and $N^n$ are increasing with $n$ and thus both converge point-wise to some limiting values $\lambda$ and $N$ that take their values on $[0, +\infty]$. Moreover, $N$ counts the points of $\overline{N}$ which belong to the positive domain under the curve $t \mapsto \lambda(t)$ by an immediate application of the monotone convergence theorem. Let's now introduce the sequence of processes $\rho^n$ defined as $\rho_t^n = \esp \l[ \lambda^n(t)- \lambda^{n-1}(t) \l|\r. \calf_T^{\theta}  \r]$. Then 
\beas 
\rho_t^{n+1} &=& \esp \l[\l. \int_0^{t}{a_s e^{-b_s(t-s)} \l(\lambda^{n}(s)-\lambda^{n-1}(s)\r)ds} \r| \calf_T^{\theta} \r]\\
&=& \int_0^{t}{a_s e^{-b_s(t-s)} \esp \l[\l. \lambda^{n}(s)-\lambda^{n-1}(s)\r| \calf_T^{\theta} \r]ds}\\
&=&\int_0^{t}{a_s e^{-b_s(t-s)} \rho_s^{n}ds}, 
\eeas
where we used Fubini's theorem in the second equality. Also, the first equality is obtained by Lemma \ref{lemmaMartingale} applied to the compensated measure $\overline{N}(ds,dz) - ds \otimes dz$ and the  independence between $\calf_T^{\overline{N}}$ and $\calf_T^{\theta}$. Thus, setting $\Phi_t^n = \int_0^t{\rho_s^n ds}$, we have by Fubini's theorem
\beas 
\Phi_t^{n+1} = \int_0^t{ \l\{ \int_0^{t-s}{a_s e^{-b_s u}du} \r\}\rho_s^n ds }.
\eeas
Note that $\int_0^{t-s}{a_s e^{-b_s u}du} \leq \frac{a_s}{b_s} \leq r < 1$ by condition (\ref{condQuotient}). Therefore, $\Phi_t^{n+1} \leq r \Phi_t^n$, and thus the application of the monotone convergence theorem to the sequence $\l(\sum_{k=0}^{n}{\Phi_t^k}\r)_n$ yields 
\bea 
\label{unmoinsr}
\esp \l[\l.\int_0^t \lambda(s)ds \r| \calf_T^{\theta} \r] &\leq& \int_0^t {\nu_s ds} + r \esp \l[\l.\int_0^t \lambda(s)ds \r| \calf_T^{\theta} \r].
\eea
A straightforward rearrangement of the terms in (\ref{unmoinsr}) gives us that 
$$\esp \l[\l.\int_0^t \lambda(s)ds \r| \calf_T^{\theta} \r] \leq (1-r)^{-1}\int_0^t{\nu_s ds} < \infty \text{    } \proba-a.s.$$ 
where the last inequality is a consequence of condition (\ref{condNu}). In particular, we deduce that $\int_0^t \lambda(s)ds $ and $N_t$ are both finite almost surely. We need to show that $\lambda(t)$ satisfies  (\ref{eqRecursive}). By mononicity, we deduce by taking the limit $n \to +\infty$ in (\ref{lambdan1}) that 
\bea 
\lambda(t) = \nu_t + \int_{0}^{t-}{a_s e^{-b_s(t-s)}dN_s}.
\eea
Finally, we show how to obtain (\ref{doublyStoProp}). As $\overline{N}$ and $\calf_T^{\theta}$ are independent, it still holds that conditioned on $\calf_T^{\theta}$, $\overline{N}$ is a Poisson process of intensity $1$. From the representation of $N$ as an integral over $\overline{N}$ we conclude that (\ref{doublyStoProp}) holds, and this completes the proof. 
\end{proof}

We now adapt well-known results on point processes to the case of the doubly stochastic Hawkes process, and derive some useful moments estimates for stochastic integrals with respect to $N$. Write $\overline{\Lambda}$ the compensating measure of $\overline{N}$, that is $\overline{\Lambda}(ds,dz) = ds \otimes dz$. Given a predictable function  $W$, write $ W \ast \overline{N}_t = \iint_{[0,t] \times \reels}{W(s,z)\overline{N}(ds,dz)} $, and the associated definition for $W \ast \overline{\Lambda}_t$. Predictable function and integral with respect to random measures definitions can be consulted in \cite{JacodLimit2003}, paragraph II.1. The following lemma is a straightforward adaptation of Lemma I.2.1.5 in \cite{jacod2011discretization}, using also Lemma \ref{lemmaMartingale} and (\ref{doublyStoProp}).

\begin{lemma*}
Let $W$ be a predictable function such that $ W^2 \ast \overline{\Lambda}_t < \infty$ almost surely. Then for any integer $p>1$, there exists a constant $K_p$ such that
$$\esp \l[ \l. \sup_{t \in [0,T]} \l|W \ast (\overline{N}-\overline{\Lambda})_t \r|^p \r|\calf_T^{\theta}\r] $$
$$\leq K_p \esp \l[\iint_{[0,T] \times \reels}{|W(s,z)|^p}dsdz + \l. \l(\iint_{[0,T] \times \reels}{W(s,z)^2}dsdz \r)^{\frac{p}{2}} \r|\calf_T^{\theta}\r]$$
\label{LemmaBDG}
\end{lemma*}

For any (random) kernel $\chi : (s,t) \to \chi(s,t)$, we say that $\chi$ is $\G$-predictable for some filtration $\G$ if for any $t \in [0,T]$ the process $\chi(.,t)$ is. For example the kernel $\chi(s,t) = a_s e^{-b_s(t-s)}$ is $\F^{\theta}$-predictable. Nonetheless,  we will also need to deal with other kernels in the course of the proofs. Consequently, we introduce the following lemma, which ensures the boundedness of moments of the doubly stochastic Hawkes process under the condition (\ref{condBoundedMoments}).

\begin{lemma*}
Under the condition $c :=  \sup_{t \in [0,T]}  \int_0^t{ a_s e^{- b_s(t-s)}ds}  < 1 \text{    } \proba-a.s.$, the counting process $N$ defined through (\ref{eqRecursive}) admits moments on $[0,T]$ that can be bounded by values independent from $T$. Moreover, for any $\F^{\theta}$-predictable kernel $\chi$ such that $\int_0^t{\chi(s,t)ds} $ is bounded uniformly in $t \in [0,T]$ independently from $T$, and for any predictable process $\psi$ that has uniformly bounded moments independently from $T$, we have

\bd
	\im[{\bf (i)}]  $\sup_{t \in [0,T]} \esp \l[\lambda(t)^p |\calf_T^{\theta}\r]^{\inv{p}} < Q_p$
	\im[{\bf (ii)}]  $\sup_{t \in [0,T]}  \esp \l[ \l(\int_0^{t}{\chi(s,t)  dN_s}\r)^p |\calf_T^{\theta}\r]^{\inv{p}} < Q_{p,\chi}$
	%\im[{\bf (iii)}] $\sup_{t \in [0,T]}  \l\| \int_0^{t}{\psi_s  dN_s} \r\|_p < Q_{p,\psi}\l(T^{\inv{2}} + T^{\inv{p}} + T \r)$
\ed 
where the constants $Q_p$, $Q_{p,\chi}$ are independent from $T$. 
%and $Q_{p,\psi}$
 \label{lemmaMomentsGeneral}
\end{lemma*}

\begin{proof}
We conduct the proof in three steps.\\
\smallskip
{\bf Step 1}. We prove that {\bf (i)} holds for $p=1$. We write 
\beas 
\esp [\lambda(t)|\calf_T^{\theta}] &=& \nu_{t} + \int_0^{t-}{a_{s} e^{-b_{s}(t-s)} \esp [\lambda(s)|\calf_T^{\theta}]ds} \\
&\leq& \overline{\nu} + \sup_{s \in [0,t]}\esp [\lambda(s)|\calf_T^{\theta}] \int_0^{t-}{a_{s} e^{-b_{s}(t-s)} ds}  \\
&\leq& \overline{\nu} + c \sup_{s \in [0,t]}\esp [\lambda(s)|\calf_T^{\theta}], 
\eeas
where we used condition (\ref{condBoundedMoments}) at the last step. Taking the supremum over $[0,T]$ on both sides, we get
\bea 
\sup_{t \in [0,T]} \esp [\lambda(t)|\calf_T^{\theta}] \leq (1-c)^{-1}\overline{\nu}.
\label{Eqnorm1}
\eea
In particular this proves the case $p=1$, since the right hand side of (\ref{Eqnorm1}) is independent from $T$.

{\bf Step 2}.  We prove that {\bf (i)} holds for any integer $p>1$. Note that it is sufficient to consider the case $p = 2^q$, $q >0$. We  thus prove our result by induction on $q$. The initialisation case $q=0$ has been proved in Step 1. Note that for any $\epsilon >0$, 
\beas 
\esp [\lambda(t)^p|\calf_T^{\theta}] &\leq& (1+\epsilon^{-1})^{2^q-1}\overline{\nu} + (1+ \epsilon)^{2^q-1}\esp\l[\l.\l(\int_0^{t-}{a_{s} e^{-b_{s}(t-s)} dN_s}\r)^p \r| \calf_T^{\theta} \r], \\
\eeas
where we have used the inequality $(x+y)^{2^q} \leq (1+\epsilon)^{2^q-1}x^{2^q}+(1+\epsilon^{-1})^{2^q-1}y^{2^q}$ for any $x,y,\epsilon >0$. Now, for a fixed $t \in [0,T]$, define $W(s,z) = a_{s} e^{-b_{s}(t-s)} \mathbb{1}_{\{0 \leq z \leq \lambda(s)\}}$, and note that
\beas 
\esp\l[\l.\l(\int_0^{t-}{a_{s} e^{-b_{s}(t-s)} dN_s}\r)^p \r| \calf_T^{\theta} \r] &=& \esp \l[  \l.\l(W \ast \overline{N}_t \r)^p \r| \calf_T^{\theta} \r]\\
&\leq& (1+\epsilon^{-1})^{2^q-1}\esp \l[ \l. \l(W \ast (\overline{N}-\overline{\Lambda})_t \r)^p \r| \calf_T^{\theta}\r]\\ & & + (1+\epsilon)^{2^q-1}\esp \l[ \l. \l(W \ast \overline{\Lambda}_t \r)^p \r| \calf_T^{\theta} \r].   \\
\eeas
We apply now Lemma \ref{LemmaBDG} to get
\beas 
\esp \l[  \l(W \ast (\overline{N}-\overline{\Lambda})_t \r)^p | \calf_T^{\theta} \r] &\leq& K_p \esp\l[\l.\iint_{[0,T] \times \reels}{|W(s,z)|^p}dsdz + \l(\iint_{[0,T] \times \reels}{W(s,z)^2}dsdz \r)^{\frac{p}{2}}\r| \calf_T^{\theta} \r]\\
&=&  K_p \esp \l[\l.\int_0^{t-}{a_{s}^p e^{-pb_{s}(t-s)} \lambda(s)ds}+ \l(\int_0^{t-}{a_{s}^2 e^{-2b_{s}(t-s)} \lambda(s)ds}\r)^{\frac{p}{2}} \r| \calf_T^{\theta} \r]. 
\eeas
We easily bound the first term by the induction hypothesis by some constant $\frac{A_p}{2}$. For the second term, an elementary application of H\"{o}lder's inequality shows that for any $k >1$ and any non-negative functions $f,g$, $(\int{fg})^k \leq (\int{f^k g})(\int{g})^{k-1}$. This along with the induction hypothesis leads to a similar bound for the second term. On the other hand, we have
\beas
 \esp \l[ \l.\l(W \ast \overline{\Lambda}_t \r)^p \r|\calf_T^{\theta} \r] &=& \esp\l[\l.\l(\int_0^{t-}{a_{s} e^{-b_{s}(t-s)} \lambda(s)ds}\r)^p \r| \calf_T^{\theta} \r]. \\
 \eeas
  We apply again the same H\"{o}lder's inequality as above with functions $f(s) = \lambda(s)$ and $g(s) = a_s e^{-b_s(t-s)}$ to get
 \beas 
 \esp\l[\l.\l(\int_0^{t-}{a_{s} e^{-b_{s}(t-s)} \lambda(s)ds}\r)^p \r| \calf_T^{\theta} \r] &\leq& c^{p-1}\esp\l[\l.\int_0^{t-}{a_{s} e^{- b_{s}(t-s)}\lambda(s)^pds} \r| \calf_T^{\theta} \r]\\
&=& c^{p-1}\int_0^{t-}{a_{s} e^{-b_{s}(t-s)} \esp \l[\lambda(s)^p| \calf_T^{\theta} \r]ds} \\
&\leq& c^p \sup_{s \in [0,t]}{\esp \l[  \lambda(s)^p| \calf_T^{\theta} \r]}
\eeas 
Finally, we have shown that 

\beas 
\esp [\lambda(t)^p|\calf_T^{\theta}] &\leq& (1+\epsilon^{-1})^{2^q-1}\overline{\nu} + (1+ \epsilon)^{2^q-1}(1+\epsilon^{-1})^{2^q-1}A_p + (1+ \epsilon)^{2^q}c^p \sup_{s \in [0,t]} \esp [\lambda(s)^p|\calf_T^{\theta}].
\eeas
This yields, taking supremum over the set $[0,T]$ and taking $\epsilon >0$ small enough so that $(1+ \epsilon)^{2^q} c^p <1$, 

\beas 
\sup_{t \in [0,T]} \esp [\lambda(t)^p|\calf_T^{\theta}]\l(1-(1+ \epsilon)^{2^q}c^p\r) &\leq& (1+\epsilon^{-1})^{2^q-1}\overline{\nu} + (1+ \epsilon)^{2^q-1}(1+\epsilon^{-1})^{2^q-1}A_p, 
\eeas
and dividing by $\l(1-(1+ \epsilon)^{2^q}c^p\r) $ on both sides we get the result.\\
{\bf Step 3}. It remains to show {\bf (ii)} and {\bf (iii)}. But note that they are direct consequences of the boundedness of moments of $\lambda$ along with Lemma \ref{LemmaBDG}. 

\end{proof}

%\begin{remark*} 
% Note that as the right hand-sides of Equations {\bf (i)}, {\bf (ii)} and {\bf (iii)} in Lemma \ref{lemmaMomentsGeneral} are deterministic constants, by Jensen's inequality the result still holds with actual expectations in place of conditional expectations. 
%\end{remark*}

\subsection{LCLT and boundedness of moments of order $2\kappa$ } \label{sectionLCLT}

We focus on asymptotic properties of the local maximum likelihood estimator $\widehat{\Theta}_{i,n}$ of our model on each block $i \in \{1, \cdots, B_n\}$. Recall that we are given the global filtration $\calf_t = \calf_t^{(\theta^{*}, \overline{N})}$ that bears a sequence of doubly stochastic Hawkes processes $(N_t^n)_{t \in [0,T]}$. We perform maximum likelihood estimation on each time block $\big( (i-1)\Delta_nT, i \Delta_nT \big]$, $i \in \{ 1, \cdots, B_n \}$ on the regression family of a parametric Hawkes process and show the local central limit theorem for every local estimator $\widehat{\Theta}_{i,n}$ of $\theta_{(i-1)\Delta_n}^{*}$, uniformly in the block index $i$. In addition, we show that all moments up to order $2\kappa > 2$ of the rescaled estimators $\sqrt{h_n}\big( \widehat{\Theta}_{i,n}- \theta_{(i-1)\Delta_n}^{*} \big)$ are convergent uniformly in $i$. 
 
\smallskip
Instead of deriving the limit theorems directly on each block, we show that by a well-chosen time change it is possible to reduce our statistical problem to a long-run framework. Such procedure is based on the following elementary lemma.
\begin{lemma*}
Let $(N_t)_t$ be a point process adapted to a filtration $\calf_t$, with $\calf_t$-stochastic intensity $\lambda(t)$. For $\gamma > 0$, consider $N_t^\gamma = N_{\gamma t}$, which is adapted to $\calf_t^\gamma = \calf_{\gamma t}$. Then, $N_t^\gamma$ admits $\lambda^\gamma(t) = \gamma \lambda(\gamma t)$ as a $\calf_t^\gamma$-stochastic intensity. Moreover, if $N_t$ is a doubly stochastic Hawkes process with parameter process $(\theta_s)_s$, $N_t^\gamma$ has the distribution of a Hawkes process of parameter $(\gamma \theta_{\gamma s})_s$, that is,
\bea 
\lambda^\gamma(t) = \gamma \nu_{\gamma t} + \int_0^{t-}{\gamma a_{\gamma s} e^{-\gamma b_{\gamma s} (t-s)}dN_s^\gamma}.
\eea

\label{lemmaTimeChange}
\end{lemma*}

\begin{proof}
First note that $N_t^\gamma = N_{\gamma t}$ is compensated by $\int_0^{\gamma t}{\lambda(s)ds}$. By a simple change of variable $u = \gamma^{-1} s$ this integral can be written as $\int_0^t{\gamma \lambda( \gamma u) du} $ which proves the first part of the lemma. In the doubly stochastic Hawkes case, let us write the integral form of the time-changed intensity and apply once again the change of variable $u = \gamma^{-1} s$,
\beas
\lambda^\gamma(t) &=& \gamma \lambda(\gamma t)\\
&=& \gamma \nu_{\gamma t} + \int_0^{\gamma t-}{\gamma a_{ s} e^{-b_{ s}(\gamma t -s)} dN_s}\\
&=& \gamma \nu_{\gamma t} + \int_0^{ t-}{\gamma a_{\gamma u} e^{-\gamma b_{\gamma u}(t -u)} dN^\gamma_{ u}},\\
\eeas 
and we are done.
\end{proof}

By virtue of Lemma \ref{lemmaTimeChange}, for any block index $i \in \{1, \cdots , B_n\}$, we consider the time change $\tau_i^n : t \mapsto n^{-1}t + (i-1)\Delta_n $ and the point process $(N_s^n)_{\{s \in ((i-1)\Delta_n, i \Delta_n]\}}$ in order to get a time changed point process $N^{i,n}$ defined on the time set $[0, h_n T]$ by the formula $N_t^{i,n} = N_{\tau_i^n(t)}^{n} - N_{\blockd}^n$. Such process is adapted to the filtration $\calf_t^{i,n}= \calf_{\tau_i^n(t)}$, for  $t \in [0,  h_nT]$. The parameter processes are now $(\theta_t^{i,n,*})_{\{t \in [0,  h_n T]\}} =(\theta_{\tau_i^n(t)}^{*})_{\{t \in [0,  h_n T]\}}$ whose canonical filtration can be expressed as $\calf_t^{\theta^{i,n,*}} =\sigma\{\theta_{s}^{i,n,*}| 0 \leq s \leq t\}$, for $t \in [0,h_n T]$. Finally note that the
$\calf_t^{i,n}$-stochastic intensities are now of the form
\bea 
\lambda_{*}^{i,n}(t) &=& \nu_{t}^{i,n,*} + \int_0^{t-}{a_{s}^{i,n,*} e^{-b_{s}^{i,n,*}(t-s)} dN_s^{i,n}} + R_{i,n}(t),
\label{newIntensity}
\eea
where $R_{i,n}(t)$ is the $\calf_0^{i,n}$-measurable residual process defined by the relation
\bea 
R_{i,n}(t) = \int_0^{(i-1)\Delta_n -}{na_{s}^{*} e^{-nb_{s}^{*}(\tau_i^n(t) -s)} dN_s^n}.
\eea
$R_{i,n} (t)$ should be interpreted as the pre-excitation induced by the preceding blocks. Note that in view of the exponential form of the kernel $\phi_t = a e^{- bt}$ assumption, $R_{i,n}(t)$ can be bounded by 
\bea 
R_{i,n}(t) \leq e^{-\underline{b}t} R_{i,n}(0)
\label{preExcitationExpo}
\eea
%For ease of notation, we now assume that we have fixed some $i \in \{1,..,\Delta_n^{-1}\}$, and we drop the index $i$ in all expressions. The reader should be easily convinced that all the results are obtained uniformly in the block index, and the value of the residual $R_n (t)$.
%The stochastic intensities are thus rewritten as
 %\bea 
 %\lambda_{*}^n(t) = \nu_{t}^{n,*} + \int_0^{t-}{a_{s}^{n,*} e^{-b_{s}^{n,*}(t-s)} dN_s^n} + R_n(t)
 %\eea
Note that all the processes $N^{i,n}$ can be represented as integrals over a sequence of Poisson processes $\overline{N}^{i,n}$ of intensity $1$ on $\reels^2$ as follows:
\bea 
N_t^{i,n} = \iint_{[0,t] \times \reels_+}{\mathbb{1}_{\{0 \leq z \leq \lambda_{*}^{i,n}(s) \}} \overline{N}^{i,n}(ds,dz)}.
\eea
Indeed, $\overline{N}^{i,n}$ is the time-space changed version of the initial Poisson process $\overline{N}$ defined by $\overline{N}^{i,n}(A \times B) = \overline{N}\l( \tau_i^n(A) \times nB\r)$ for $A$ and $B$ any two Borel sets of $\reels$. In the time-changed representation, we define the regression family of stochastic intensities 
\bea 
\tilde{\lambda}^{i,n}(t,\theta) = \nu + \int_0^{t-}{ae^{-b(t-s)}dN_s^{i,n}},
\eea
which is related to $\lambda^{i,n}$ (see (\ref{candidateIntensity})) by $\tilde{\lambda}^{i,n}(t,\theta) = n^{-1} \lambda^{i,n}(\tau_i^n(t), \theta)$.
Also, the Quasi Log Likelihood process defined in (\ref{QLL}) on the $i$-th block has now the representation (up to the constant term $\text{log}(n) N_{h_nT}^{i,n}$)
\bea 
l_{i,n}(\theta) = \int_0^{h_n T}{\text{log}(\tilde{\lambda}^{i,n}(t,\theta))dN_t^{i,n}} - \int_0^{h_n T}{\tilde{\lambda}^{i,n}(t,\theta)dt},
\eea
Note that in our case, the true underlying intensity, $\lambda_{*}^{i,n}$ does not belong to the regression family $(\tilde{\lambda}^{i,n}(.,\theta))_{\theta \in K}$ for two reasons : the parameter process $\theta^{*}$ is not constant on the $i$-th block, and the regression family does not take into account the existence of a pre-excitation term in (\ref{newIntensity}). We are in a mispecified case, but we wish to take advantage of the continuity of the process $\theta^{*}$ to show that the asymptotic theory still holds, that is, the MLE tends to the value $\theta_0^{i,n,*} = \theta_{(i-1)\Delta_n}^{*}$ which is the value of the process $\theta^{*}$ at the beginning of the $i$-th block. The procedure is thus asymptotically equivalent to performing the MLE on the model whose stochastic intensity is in the regression family with true value $\theta = \theta_{(i-1)\Delta_n}^{*}$. To formalize such idea, we introduce an auxiliary model corresponding to the parametric case generated by the true value $\theta_{(i-1)\Delta_n}^{*}$. More precisely, we introduce the constant parameter Hawkes process $N^{i,n,c}$ generated by $\overline{N}^{i,n}$ and the initial value $\theta_0^{i,n,*}$, whose stochastic intensity satisfies 
\bea 
\lambda^{i,n,c}(t) = \nu_0^{i,n,*} + \int_0^{t-}{a_0^{i,n,*}e^{-b_0^{i,n,*}(t-s)} dN_s^{i,n,c}}.
\eea
Moreover, we assume that $N_t^{i,n,c}$ has the representation
\bea
N_t^{i,n,c} = \iint_{[0,t] \times \reels_+}{\mathbb{1}_{\{0 \leq z \leq \lambda^{i,n,c}(s) \}} \overline{N}^{i,n}(ds,dz)}.
\eea
Note that $N_t^{i,n,c}$ is unobserved and just used as an intermediary to derive the asymptotic properties of the MLE, by showing systematically that any variable $N^{i,n}$, $\tilde{\lambda}^{i,n}$, $l_{i,n}$, etc. is asymptotically very close to its counterpart that is generated by the constant parameter model.

\smallskip
For reasons that will become apparent later, it is crucial to localize the pre-excitation $R_{i,n}(0)$ and bound it by some deterministic value $M_n$ that depends solely on $n$ and such that $M_n = O(n^q)$ for some $q > 1$. To reduce our local estimation problem to the case of a parametric Hawkes process, we will also need to condition with respect to the initial value of the parameter process. We will thus use extensively the conditional expectations $\esp[ . \mathbf{1}_{\{R_{i,n}(0) \leq M_n\}} | \calf_0^{i,n}, \theta_0^{i,n,*} = \theta_0 ]$, that we denote  by  $\espc$, and whose existences are justified by a classical regular distribution argument\footnote{This is a consequence to the fact that $K \subset \reels^3$ is a Borel space.}  (see for instance Section $4.3$ (pp. $77-80$) in \cite{breiman1992probability}). In the same spirit, for a measurable set $A \in \calf$, $\probac[A]$ should be understood as $\espc[\mathbf{1}_A]$. Finally we will need frequently to take supremum over the quadruplet     $(\theta_0,i,n,t)$. For that reason we introduce the notation $\E = \{\l. (\theta_0,i,n,t) \in K \times \naturels^2 \times  \reels_+ \r | 1 \leq i \leq B_n \text{  and  } 0 \leq t \leq h_n T\} $. When $n \in \naturels$ is fixed, we define $\E_n$ the subset of $\E$ as $\E_n = \{\l. (\theta_0,i,t) \in K \times \naturels \times  \reels_+ \r | 1 \leq i \leq B_n \text{  and  } 0 \leq t \leq h
_n T\} $. In the same spirit, it is also useful when truncation arguments appear, to consider in the previous equation the subset of $\E_n$ for which we have the stronger condition $h_n^\alpha T \leq t \leq h_n T$ where $\alpha \in (0,1)$ that we denote by $\E_n^\alpha$. The next lemma states the uniform boundedness of the moments of $\lambda_{*}^{i,n}$ and $\lambda^{i,n,c}$, along with $\mathbb{L}^p$ estimates for stochastic integrals over $N^{i,n}$ and $N^{i,n,c}$.

\begin{lemma*}
We have, for any integer $p \geq 1$ and any $\F^{\theta^{i,n,*}}$-predictable kernel $\chi$ such that $\int_0^t{\chi(s,t)ds} $ is bounded uniformly in $t \in [0,h_n T]$ independently from $T$ and $n$, %and for any predictable process $\psi$ that has uniformly bounded moments independently of $T$ and $n$, we have:
\bd
		\im[{\bf (i)}] $\sup_{(\theta_0,i,n,t) \in \E} \espc\l|\lambda_{*}^{i,n}(t) \r|^p \leq M_p \text{  }\proba\text{-a.s.}$
		\im[{\bf (ii)}] $\sup_{(\theta_0,i,n,t) \in \E} \espc\l| \int_0^{t}{\chi(s,t)  dN_s^{i,n}} \r|^p < M_{p,\chi} \text{  }\proba\text{-a.s.}$
		\im[{\bf (iii)}] $\sup_{(\theta_0,i,n,t) \in \E} \espc\l|\lambda^{i,n,c}(t) \r|^p <M_p \text{  }\proba\text{-a.s.}$
		\im[{\bf (iv)}] $\sup_{(\theta_0,i,n,t) \in \E} \espc\l| \int_0^{t}{\chi(s,t)  dN_s^{i,n,c}} \r|^p < M_{p,\chi} \text{  }\proba\text{-a.s.}$
		
	\ed
where $M_p$ and $M_{p,\chi}$ are finite constants depending respectively solely on $p$ and on $p$ and $\chi$.
\label{boundedMoments}
\end{lemma*}

\begin{proof}
This is a straightforward adaptation of the proof of Lemma \ref{lemmaMomentsGeneral}, with the conditional expectation $\esp[ . \mathbf{1}_{\{R_{i,n}(0) \leq M_n\}} | \calf_0^{i,n} \vee \calf_{h_n T}^{\theta^{i,n,*}}, \theta_0^{i,n,*} = \theta_0 ]$. The presence of $\mathbf{1}_{\{R_{i,n}(0) \leq M_n\}}$ along with the exponential decay in (\ref{preExcitationExpo}) show clearly that the result still holds, uniformly in the quadruplet $(\theta_0,i,n,t)$. By an immediate application of Jensen's inequality, this is still true replacing $\calf_0^{i,n} \vee \calf_{h_n T}^{\theta^{i,n,*}}$ by the smaller filtration $\calf_0^{i,n}$, that is, for the operator $\espc$.
\end{proof}
Before we turn to estimating the distance between the two models, we state a technical lemma. 
\begin{lemma*}
Let $h : s \mapsto ae^{-bs}$, and let $f$,$g$ be two non-negative functions satisfying the inequality $f \leq g + f \ast h$ where $(f \ast h)(t) = \int_0^t{f(t-s)h(s)ds}$ is the usual convolution. Then we have the majoration for any $t \geq 0$
$$ f(t)  \leq g(t) + a\l(g \ast e^{(a-b).} \r)(t)$$
\label{lemmaConvolution}
\end{lemma*}

\begin{proof}
Iterating the inequality we get for any $n \in \naturels^*$ 
\bea 
f \leq g + g \ast \sum_{k = 1}^n{h^{\ast(k) }} + f \ast h^{\ast(n+1)}.
\label{eqConv}
\eea
We fix $t \geq 0$, and note that by a straightforward computation, for any integer $k \geq 1$ we have $h^{\ast(k)}(t) = \frac{t^{k-1}}{(k-1)!}a^{k}e^{-bt}$. We deduce that 
\beas 
f \ast h^{\ast(n+1)}(t) &=& \int_0^t{f(t-s) \frac{s^n}{n!}a^{n+1}e^{-bs}ds }\\  
&\leq& \frac{t^{n}}{n!}a^{n+1} \int_0^t{f(s)ds} \to 0
\eeas 
as $n \to +\infty$. We also have for any integer $n \geq 1$
\beas 
\sum_{k = 1}^n{h^{\ast(k) }}(t) &=& \sum_{k = 1}^n \frac{t^{k-1}}{(k-1)!}a^{k}e^{-bt} \\
&\leq&a e^{(a-b)t}
\eeas
and thus we get the result by taking the limit $n \to +\infty$ in (\ref{eqConv}) evaluated at any point $t \geq 0$.
\end{proof}

In what follows, we quantify the local error between the doubly stochastic model and its constant parameter approximation. We recall the value of the key exponent $\kappa = \gamma(\delta -1)$ that has been introduced in (\ref{defKappa}), and which plays an important role in the next results as it proves to be the rate of convergence of one model to the other in power of $h_n^{-1}$, where $h_n$ is proportional to the typical size of one block after our time change.  Recall that $\gamma$ represents the regularity exponent in time of $\theta$ while $\delta$ controls the size of small blocks compared to $n$ by the relation $h_n = n^{1/\delta}$. Note that by (\ref{defKappa}) we have $\kappa >1$. The next lemma shows that the models $(N^{i,n,c}, \lambda^{i,n,c})$ and $(N^{i,n},\lambda_{*}^{i,n})$ are asymptotically close in the $\mathbb{L}^p$ sense. The proof follows the same path as the proof of Lemma \ref{lemmaMomentsGeneral}.

\begin{lemma*}
Let $\alpha \in (0,1)$ be a truncation exponent, and $\epsilon \in (0,1)$. We have, for any $p \geq 1$, any deterministic kernel $\chi$ such that $\int_0^t{\chi(s,t)ds} $ is bounded uniformly in $t \in \reels_+$, and  any predictable process $\l(\psi_s\r)_{s \in \reels_+}$ whose moments are bounded :
\bd
		\im[{\bf (i)}] $ \sup_{(\theta_0,i,t) \in \E_n^\alpha} \espc\l|\lambda^{i,n,c}(t) - \lambda_{*}^{i,n}(t) \r|^p = O_{\proba} \l(h_n^{-\kappa}\r)$
		\im[{\bf (ii)}] $\sup_{\theta_0 \in K, 1 \leq i \leq B_n} \espc\l| \int_{h_n^\alpha T}^{h_n T}{\psi_s \{dN_s^{i,n,c} - dN_s^{i,n}\}} \r|^p = O_{\proba} \l( h_n^{p-\epsilon\kappa}\r)$
		\im[{\bf (iii)}] $ \sup_{\theta_0 \in K, 1 \leq i \leq B_n} \espc\l| \int_{h_n^\alpha T}^{h_n T}{\chi(s,h_n T) \{dN_s^{i,n,c} - dN_s^{i,n}\}} \r|^p = O_{\proba} \l( h_n^{-\kappa}\r)$
	\ed
\label{boundedDeviation}
\end{lemma*}

\begin{remark*}
For $p=1$, if we recall that $\Delta_n = h_n n^{-1}T$ and $\kappa = \gamma(\delta-1)$, we get a typical  deviation in $h_n^{-\kappa} = T^{-\gamma}\Delta_n^\gamma$ between the real model and its constant parameter approximation. This is not very surprising since on one block the parameter process $\theta^{*}$ has exactly a deviation of that order. For $p >1$, the situation is fairly different. One would expect a deviation of the same order of that of the parameter process, that is of order $h_n^{-\kappa p} = T^{-\gamma p}\Delta_n^{\gamma p}$. But as it is shown in the previous lemma, deviations between the two models are quite weaker since the deviation remains of order $h_n^{-\kappa} = T^{-\gamma}\Delta_n^\gamma$ for any $p$. This loss is due to the point process structure and the shape of its related Burkholder-Davis-Gundy type inequality (see Lemma \ref{LemmaBDG}). This is the same phenomenon as in the following fact. For a Poisson process $N$ of intensity $\lambda$, we have $\esp[|N_t -\lambda t|^p] \sim \alpha_p t $ when $t \to 0$, i.e. a rate of convergence which is linear regardless of the moment chosen.    
\end{remark*}

\begin{proof}
We will show by recurrence on $q \in \naturels$ that for every $p$ of the form $p = 2^q$, we have the majoration for $n \in \naturels$, $t \in [0, h_n T]$ and uniformly in $(\theta_0,i)$,
\bea 
 \espc \l|\lambda^{i,n,c}(t) - \lambda_{*}^{i,n}(t) \r|^{2^q} \leq L_{n,q} + M_{n,q} e^{-\underline{b}(1-r) t},
\label{majExpo}
\eea
where $L_{n,q}$ and $M_{n,q}$ depend on $n$ and $q$ only, $L_{n,q} = O_{\proba}(h_n^{-\kappa})$, and  $M_{n,q}$ is of polynomial growth in $n$. Note that then {\bf (i)} will be automatically proved since
by taking the supremum over the set $[h_n^\alpha T, h_nT]$ and using the estimate $M_{n,q} e^{-\underline{b}(1-r) h_n^\alpha T} = o_{\proba}(h_n^{-\kappa}) $ we get 

$$  \espc  |\lambda^{n,c}(t) - \lambda_{*}^n(t)|^p = O_{\proba}(h_n^{-\kappa}) $$
uniformly over the set $\E_n^\alpha$.

{\bf Step 1}. We show our claim in the case $q = 0$, that is $p=1$. Write 
\beas
|\lambda_{*}^{i,n}(t)-\lambda^{i,n,c}(t)  | &\leq& |\nu_{t}^{i,n,*} - \nu_0^{i,n,*}| + \l| \int_0^{t-}{\l(a_{s}^{i,n,*} e^{-b_{s}^{i,n,*}(t-s)} - a_{0}^{i,n,*} e^{-b_{0}^{i,n,*}(t-s)}\r)}dN_s^{i,n}\r| \\
&+& \l|\int_0^{t-}{a_{0}^{i,n,*} e^{-b_{0}^{i,n,*}(t-s)}\l(dN_s^{i,n,c} - dN_s^{i,n}\r)}\r| +  R_{i,n}(t) \\
&\leq& A_{i,n}(t) + B_{i,n}(t) +C_{i,n}(t) +  R_{i,n}(t)
\eeas

The (uniform) majoration $ \espc A_{i,n}(t) =O_{\proba}(h_n^{-\kappa})$ is an immediate consequence of [C]-{\bf (i)}. By the inequality 
\bea 
|ae^{-bt} - a^{'}e^{-b^{'}t}|  \leq \l(|a-a^{'}| + |b-b^{'}| \r)e^{-\underline{b}t}
\label{majorationg}
\eea
for any $(\nu,a,b), (\nu^{'},a^{'},b^{'}) \in K$, we can write
\beas 
\espc B_{i,n}(t) &\leq& \espc \int_0^{t-} {\l(|a_{s}^{i,n,*}-a_0| + |b_{s}^{i,n,*}-b_0| \r)e^{-\underline{b}(t-s)}dN_s^{i,n}} \\
&\leq& \sqrt{\espc\l|\sup_{s \in [0,t]} {\l(|a_{s}^{i,n,*}-a_0| + |b_{s}^{i,n,*}-b_0| \r)} \r|^2 \espc \l|\int_0^{t-}{e^{-\underline{b}(t-s)}dN_s^{i,n}}\r|^2},
\eeas
where Cauchy-Schwartz inequality was applied in the last inequality. Note that the right term is almost surely dominated by a constant by Lemma \ref{boundedMoments} and thus the uniform majoration $ \espc B_{i,n}(t) =O_{\proba}(h_n^{-\kappa})$ follows from [C]-{\bf (i)}. Finally, for $C_{i,n}(t)$, write 
\bea 
\espc C_{i,n}(t) &\leq& \espc \int_0^{t -}{a_{0} e^{-b_{0}(t-s)}d\l|N^{i,n,c} - N^{i,n}\r|_s}
\eea 
where $d\l|N^{i,n,c} - N^{i,n}\r|_s$ is the integer measure which counts the jumps that don't belong to both $dN^{i,n,c}$ and $dN^{i,n}$, i.e. the points of $\overline{N}^{i,n}$ that lay between the curves $t \to \lambda_*^{i,n}(t)$ and $t \to \lambda^{i,n,c}(t)$. A short calculation shows that this counting process admits $|\lambda^{i,n,c}(s) - \lambda_{*}^{i,n}(s)|$ as stochastic intensity. We compute now:
\beas 
\espc C_{i,n}(t) &\leq& \espc \int_0^{t-}{a_{0} e^{-b_{0}(t-s)}|\lambda^{i,n,c}(s) - \lambda_{*}^{i,n}(s)|ds} \\
&=& \int_0^{t-}{a_{0} e^{-b_{0}(t-s)} \espc|\lambda^{i,n,c}(s) - \lambda_{*}^{i,n}(s)|ds}. \\
\eeas
So far we have shown that there exists a sequence $L_n$ such that $L_n = O(h_n^{-\kappa})$ and such that the function $f(t) = \espc |\lambda^{i,n,c}(t) - \lambda_{*}^{i,n}(t)|$ satisfies the inequality \bea 
f(t) \leq L_n + R_{i,n}(t) + f \ast h (t),
\eea
where $h$ is the kernel defined as $h: t \mapsto a_0e^{-b_0t}$. By Lemma \ref{lemmaConvolution}, this yields
\bea 
f(t) \leq L_n + R_{i,n}(t) + \int_0^t{\{L_n + R_{i,n}(s)\} a_0e^{(a_0 - b_0)(t-s)}ds   }.
\eea 
Now recall that $b_0 - a_0 > \underline{b}(1-r)$ and that on the set $\{R_{i,n}(0) \leq M_n\}$, we have $R_{i,n}(s) \leq M_n e^{-\underline{b}s} < M_n e^{-\underline{b}(1-r)s}$ to get
\beas 
f(t) &\leq& (1+(1-r)^{-1})L_n + R_{i,n}(t) + \int_0^t{\{M_n e^{-\underline{b}(1-r)s}\} a_0e^{\underline{b}(1-r)(t-s)}ds   }\\
&\leq&(1+(1-r)^{-1})L_n + (1+ \overline{a}t) M_n e^{-\underline{b}(1-r)t}. 
\eeas
If we recall that in the above expression $f(t)$ stands for $\espc |\lambda^{i,n,c}(t) - \lambda_{*}^{i,n}(t)|$, such uniform estimate clearly proves (\ref{majExpo}) in the case $q=1$.

{\bf Step 2}. We prove the result for any $q\in \naturels^{*}$.  Let the expression $f(t)$ stands for $\espc |\lambda^{i,n,c}(t) - \lambda_{*}^{i,n}(t)|^p$. With similar notations as for the previous step, we have for any $\eta >0$ 
\beas
f(t) = \espc|\lambda^{i,n,c}(t) - \lambda_{*}^{i,n}(t)|^p &\leq& \espc \l|A_{i,n}(t) + B_{i,n}(t) +C_{i,n}(t) +  R_{i,n}(t) \r|^p\\
&\leq& (1+\eta^{-1})^{2^q-1}\espc \l| A_{i,n}(t) + B_{i,n}(t) + R_{i,n}(t) \r|^p\\
& & + (1+\eta)^{2^q-1} \espc C_{i,n}(t)^p
\eeas
It is straightforward to see that similar arguments to the previous case lead to the uniform estimate
\beas
\espc  A_{i,n}(t)^p + \espc B_{i,n}(t)^p = O_{\proba}\l(h_n^{-\kappa p}\r).
\eeas
 Now, define $W(s,z)= a_0 e^{-b_0(t-s)} |\mathbb{1}_{\{0 \leq z \leq \lambda^{i,n,c}(s) \}} - \mathbb{1}_{\{0 \leq z \leq \lambda_{*}^{i,n}(s) \} }|$ to get
\beas 
\espc\l[C_{i,n}(t)^p\r] &=& \espc \l[  \l(W \ast \overline{N}_t \r)^p\r]\\
&\leq& (1+\eta^{-1})^{2^q-1}\espc\l[  \l(W \ast (\overline{N}-\overline{\Lambda})_t \r)^p \r]+ (1+\eta)^{2^q-1}\espc \l[ \l(W \ast \overline{\Lambda}_t \r)^p \r],\\
\eeas
and apply Lemma \ref{LemmaBDG} to get
\beas 
\espc \l[  \l(W \ast (\overline{N}-\overline{\Lambda})_t \r)^p \r] &\leq& K_p\l( \espc\l[\iint_{[0,T] \times \reels}{|W(s,z)|^p}dsdz \r] \r.\\& &+ \l. \espc\l[\l(\iint_{[0,T] \times \reels}{W(s,z)^2}dsdz \r)^{\frac{p}{2}} \r]\r)\\
&=&  K_p \l( \espc \l[\int_0^{t-}{a_{0}^p e^{-pb_{0}(t-s)} |\lambda^{i,n,c}(s) - \lambda_{*}^{i,n}(s)| ds} \r]\r. \\
& & +  \l.\espc \l[\l(\int_0^{t-}{a_{0}^2 e^{-2b_{0}(t-s)} |\lambda^{i,n,c}(s) - \lambda_{*}^{i,n}(s)| ds}\r)^{\frac{p}{2}} \r]\r),
\eeas
which is easily bounded as in (\ref{majExpo}) using the induction hypothesis. Note that here the presence of the integral term in $|\lambda^{i,n,c}(s) - \lambda_{*}^{i,n}(s)|$ is the major obstacle to getting the stronger estimate $O_{\proba}\l(h_n^{-\kappa p}\r)$ that one would expect. Finally the term

$$\espc \l[ \l(W \ast \overline{\Lambda}_t \r)^p \r] = \espc \l[ \l(\int_0^{t-}{a_{0} e^{-b_{0}(t-s)} |\lambda^{i,n,c}(s) - \lambda_{*}^{i,n}(s)| ds}\r)^p \r]$$
is treated exactly in the same way as for the proof of Lemma \ref{lemmaMomentsGeneral}, to get the bound 
\bea 
\espc \l[ \l(W \ast \overline{\Lambda}_t \r)^p \r] &\leq& c_q f \ast h(t),
\eea
where again $h : s \mapsto a_0e^{-b_0s}$, and $c_q <1$ if $\eta$ is taken small enough. We have thus shown that $f$ satisfies a similar convolution inequality as for the case $q=1$ and we can apply Lemma \ref{lemmaConvolution} to conclude.

{\bf Step 3}. It remains to show {\bf (ii)} and  {\bf (iii)}. They are just  consequences of the application of Lemma \ref{LemmaBDG} to the case $W_\psi(s,z) = \psi_s |\mathbb{1}_{\{0 \leq z \leq \lambda^{n,c}(s) \}} - \mathbb{1}_{\{0 \leq z \leq \lambda_{*}^n(s) \} }|$  and $W_\chi(s,z) = \chi(s,t) |\mathbb{1}_{\{0 \leq z \leq \lambda^{n,c}(s) \}} - \mathbb{1}_{\{0 \leq z \leq \lambda_{*}^n(s) \} }|$ along with H\"{o}lder's inequality. 
\end{proof}

We are now ready to show the uniform asymptotic normality of the MLE by proving that any quantity related to the estimation is asymptotically very close to its counterpart for the constant parameter model $(N^{i,n,c},\lambda^{i,n,c})$. To this end we introduce the fake candidate intensity family and the fake log-likelihood process, as 
\bea 
\lambda^{i,n,c}(t,\theta) = \nu + \int_0^{t-}{ae^{-b(t-s)}dN_s^{i,n,c}}
\eea 
and
\bea 
l_{i,n}^{c}(\theta) = \int_0^{h_n T}{\text{log}(\lambda^{i,n,c}(t,\theta))dN_t^{i,n,c}} - \int_0^{h_n T}{\lambda^{i,n,c}(t,\theta)dt},
\eea
for any $\theta \in K$. Note that $\lambda^{i,n,c}(t, \theta_0^{i,n,*}) = \lambda^{i,n,c}(t)$ by definition. Those quantities, which are all related to $(N^{i,n,c},\lambda^{i,n,c})$, are unobserved.\\

As a consequence of the previous lemma we state the uniform $\mathbb{L}^p$ boundedness of the candidate intensity families, along with estimates of their relative deviations. 
\begin{lemma*}
Let $\alpha \in (0,1)$. We have for any integer $p \geq 1$ and any $j \in \naturels$ that

\bd 

\im[{\bf (i)}]  $ \sup_{(\theta_0,i,n,t) \in \E} \espc \sup_{\theta \in K}\l|  \partial_\theta^j  \tilde{\lambda}^{i,n}(t,\theta) \r|^p \leq K_j \text{   }\proba\text{-a.s.}$
\im[{\bf (ii)}] $ \sup_{(\theta_0,i,n,t) \in \E} \espc \sup_{\theta \in K} \l|  \partial_\theta^j  \lambda^{i,n,c}(t,\theta) \r|^p \leq K_j \text{   }\proba\text{-a.s.}$
\im[{\bf (iii)}] $\sup_{(\theta_0,i,t) \in \E_n^\alpha} \espc \sup_{\theta \in K} \l|   \partial_\theta^j \tilde{\lambda}^{i,n}(t,\theta) - \partial_\theta^j  \lambda^{i,n,c}(t,\theta)  \r|^p  = O_{\proba}(h_n^{-\kappa})$

\ed
where the constants $K_j$ depend solely on $j$.
\label{lemmaBoundedDerivatives}
\end{lemma*}
\begin{proof}
Note that the derivatives of $\tilde{\lambda}^{i,n}(t,\theta)$ can be all bounded uniformly in $\theta$ by linear combinations of terms of the form $\overline{\nu}$ or $\int_0^{t-}{(t-s)^j e^{-\underline{b}(t-s)} dN_s^{i,n}}$, $j \in \naturels$. The boundedness of moments of those terms uniformly in $n \in \naturels$ and in the time interval $[0,h_nT]$ is thus the consequence of Lemma \ref{lemmaMomentsGeneral} {\bf (ii)} with $\chi(s,t) = (t-s)^j e^{-\underline{b} (t-s)}$, and consequently {\bf (i)} follows. {\bf (ii)} is proved in the same way. Finally we show {\bf (iii)}. Note that $\sup_{\theta \in K} | \partial_\theta^j \tilde{\lambda}^{i,n}(t,\theta) - \partial_\theta^j  \lambda^{i,n,c}(t,\theta) |$ can be bounded by linear combinations of terms of the form $\int_0^{t-}{(t-s)^j e^{-\underline{b}(t-s)} d|N^{i,n} - N^{i,n,c}|_s}$. The $\mathbb{L}^p$ estimate of such expression is then easily derived by a truncation argument and Lemma \ref{boundedDeviation} {\bf (iii)}.
\end{proof}
We now follow similar notations to the ones introduced in \cite{Clinet20171800},  and consider the main quantities of interest to derive the properties of the MLE. We define for any $(\theta, \theta_{0}) \in K^2$,
\bea 
\mathbb{Y}_{i,n}(\theta, \theta_0) = \inv{h_nT}(l_{i,n}(\theta) - l_{i,n}(\theta_0)),
\eea 
\bea
\Delta_{i,n}(\theta_0) &=& \inv{\sqrt{h_nT}} \partial_\theta l_{i,n}(\theta_0), 
\label{defDeltaBlock}
\eea
and finally 
\bea
\Gamma_{i,n}(\theta_0) = -\inv{h_nT}\partial_{\theta}^2 l_{i,n}(\theta_0).
\eea
We define in the same way $\mathbb{Y}_{i,n}^c$,  $\Delta_{i,n}^c$, and $\Gamma_{i,n}^c$. We introduce for the next lemma the set $\I = \{(\theta_0,i,n) \in K \times \naturels^2 | 1 \leq i \leq B_n \}$.  

\begin{lemma*}
Let $\epsilon \in (0,1)$, and $L \in (0, 2\kappa)$. For any $p \in \naturels^*$, for any $\epsilon \in (0,1)$, we have the estimates
\bea 
 \sup_{\theta_0 \in K, 1 \leq i \leq B_n}  \espc \l|\Delta_{i,n}(\theta_0) - \Delta_{i,n}^c(\theta_0) \r|^L \to^\proba 0,
\label{eqDeltaDev}
\eea
\bea
\sup_{\theta_0 \in K, 1 \leq i \leq B_n}  \espc \l[ \sup_{\theta \in K} |\mathbb{Y}_{i,n}(\theta,\theta_0)-\mathbb{Y}_{i,n}^c(\theta,\theta_0)|^p \r] = O_\proba\l( h_n^{-\epsilon\kappa}\r), 
\label{eqY}
\eea
\bea
\sup_{\theta_0 \in K, 1 \leq i \leq B_n} \espc \l|\Gamma_{i,n}(\theta_0) - \Gamma_{i,n}^c(\theta_0) \r|^p = O_\proba\l(  h_n^{-\epsilon\kappa}\r), 
\label{eqGamma}
\eea
\bea
\sup_{(\theta_0,i,n) \in \I} \espc \left| h_n^{-1} \sup_{\theta \in K}|\partial_\theta^3l_{i,n}(\theta)| \right|^p < K \text{   } \proba\text{-a.s.} 
\label{eqPartial}
\eea

\label{lemma4Inequalities}
\end{lemma*}

\begin{proof}

Let us show (\ref{eqDeltaDev}). We can express the equation in (\ref{defDeltaBlock}) and its counterpart for the constant model as 
\bea 
\Delta_{i,n}(\theta_0) = \inv{\sqrt{h_nT}}\l\{ \int_0^{h_nT}{ \frac{\partial_\theta\tilde{\lambda}^{i,n}(s,\theta_0)}{\tilde{\lambda}^{i,n}(s,\theta_0)} dN_s^{i,n}}-\int_0^{h_nT}{\partial_\theta \tilde{\lambda}^{i,n}(s,\theta_0)ds}\r\}
\eea
and 
\bea 
\Delta_{i,n}^c(\theta_0) = \inv{\sqrt{h_nT}}\l\{ \int_0^{h_nT}{ \frac{\partial_\theta\lambda^{i,n,c}(s,\theta_0)}{\lambda^{i,n,c}(s,\theta_0)} dN_s^{i,n,c}}-\int_0^{h_nT}{\partial_\theta \lambda^{i,n,c}(s,\theta_0)ds}\r\}.
\label{martingaleDeltaC}
\eea
By Lemma \ref{boundedMoments} {\bf (i)} and {\bf (iii)}, and Lemma \ref{lemmaBoundedDerivatives} {\bf (i)} and {\bf (ii)} and the presence of the factor $\inv{\sqrt{h_nT}}$, it is possible to replace the lower bounds of those integrals by $h_n^\alpha T$ for some $\alpha \in (0, \half)$. Thus the difference $\sqrt{h_nT}(\Delta_{i,n}(\theta_0)-\Delta_{i,n}^c(\theta_0))$ is equivalent to the sum of the three terms 

\beas 
\int_{h_n^\alpha T}^{h_nT}{ \frac{\partial_\theta\tilde{\lambda}^{i,n}(s,\theta_0)}{\tilde{\lambda}^{i,n}(s,\theta_0)} (dN_s^{i,n}-dN_s^{i,n,c})} &+& \int_{h_n^\alpha T}^{h_nT}{ \l\{\frac{\partial_\theta\tilde{\lambda}^{i,n}(s,\theta_0)}{\tilde{\lambda}^{i,n}(s,\theta_0)} - \frac{\partial_\theta\lambda^{i,n,c}(s,\theta_0)}{\lambda^{i,n,c}(s,\theta_0)} \r\} dN_s^{i,n,c}}\\
&+&\int_{h_n^\alpha T}^{h_nT}{\{\partial_\theta \tilde{\lambda}^{i,n}(s,\theta_0)-\partial_\theta \lambda^{i,n,c}(s,\theta_0) \}ds}.
\eeas
We therefore apply Lemmas \ref{boundedDeviation} {\bf (ii)} and \ref{lemmaBoundedDerivatives} {\bf (i)} to the first term, Lemmas \ref{boundedMoments} {\bf (iii)} and \ref{lemmaBoundedDerivatives} {\bf (iii)} to the second term, and finally Lemma \ref{lemmaBoundedDerivatives} {\bf (iii)} to the last term to obtain the overall estimate 

\bea 
 \sup_{\theta_0 \in K, 1 \leq i \leq B_n}\espc \l|\Delta_{i,n}(\theta_0) - \Delta_{i,n}^c(\theta_0) \r|^L = O_\proba\l(h_n^{\frac{L}{2} -\epsilon \kappa}\r),
 \label{suchexpression}
\eea
for any $\epsilon \in (0,1)$. This tends to $0$ if we can find an $\epsilon $ such that $\frac{L}{2} -\epsilon \kappa <0$, and this can be done by taking $\epsilon $  sufficiently close to $1$ since $L < 2\kappa$. Equations (\ref{eqY}), (\ref{eqGamma}) and (\ref{eqPartial}) are proved similarly. 
\end{proof}

\begin{lemma*}
\label{lemmaConvConstant}
For any integer $p \geq 1$, there exists a constant $M$ such that 
\bea 
 \sup_{(\theta_0,i,n) \in \I}  \espc \l| \Delta_n^c(\theta_0) \r|^p < M \text{   } \proba\text{-a.s.}
\label{eqDeltaDevConstant}
\eea
Furthermore, there exists a mapping $(\theta,\theta_0) \to \mathbb{Y}(\theta,\theta_0)$ such that for any $\epsilon \in (0,1)$,
\bea
 \sup_{\theta_0 \in K, 1 \leq i \leq B_n} \espc\l[ \sup_{\theta \in K} |\mathbb{Y}_{i,n}^{c}(\theta,\theta_0)-\mathbb{Y}(\theta,\theta_0)|  \r] = O\l(h_n^{-\epsilon \frac{p}{2}} \r) \text{   } \proba\text{-a.s.}
 \label{eqYConstant}
\eea
Finally, for any $\theta_0 \in K$, and for any $\epsilon \in (0,1)$,
\bea 
\sup_{\theta_0 \in K, 1 \leq i \leq \Delta_{n}^{-1}}  \espc\l|\Gamma_{i,n}^c(\theta_0) - \Gamma(\theta_0)\r|^p = O\l(h_n^{-\epsilon \frac{p}{2}} \r) \text{   } \proba\text{-a.s.}
\label{eqGammaConstant}
\eea
where $\Gamma(\theta_0)$ is the asymptotic Fisher information matrix of the parametric Hawkes process regression model with parameter $\theta_0$ as introduced in (\ref{paramFisher}).

%Finally, Let $E \in \calf$. We have 
%	\bea 
%	 \sup_{\theta_0 \in \Theta}\l\{\espc \left[ f(\sqrt{h_n}(\Delta_n^c(\theta_0))\mathbf{1}_E \right] - \espc[f(\Gamma(\theta_0)^{\frac{1}{2}}\xi)\mathbf{1}_E ]\r\} \to 0,
%	\label{convMartingale}
%	\eea
%	for any continuous bounded function $f$, and such that $\xi$ follows a standard normal distribution and is independent of $\calf$. 

\end{lemma*}

\begin{proof}

Note that when $\theta_0^{i,n,*} = \theta_0$, the constant model $N^{i,n,c}$ is simply a parametric Hawkes process with parameter $\theta_0$, and is independent of the filtration $\calf_0^{i,n}$. Thus, by a regular distribution argument the operator $\espc$ acts as the simple operator $\esp$ for $N^{i,n,c}$ distributed as a Hawkes with true value $\theta_0$. It is straightforward to see that under a mild change in the proofs of Lemma 3.15 and Theorem 4.6 in \cite{Clinet20171800} those estimates hold uniformly in $\theta_0 \in K$ and in the block index. 
\end{proof}

\begin{theorem*} \label{convergence}
Let $L \in (0,2\kappa)$. We have 
	\bea 
	\label{momentconv}
	  \sup_{\theta_0 \in K, 1 \leq i \leq B_n}\l\{\espc \left[ f\l(\sqrt{h_n }(\widehat{\Theta}_{i,n}-\theta_0)\r) \right] - \esp\l[f\l( T^{-\half}\Gamma(\theta_0)^{-\frac{1}{2}}\xi\r) \r]\r\} \to^\proba 0,
	\label{convMoment}
	\eea
	for any continuous function $f$ with $|f(x)| = O(|x|^L)$ when $ |x| \to \infty$ , and such that $\xi$ follows a standard normal distribution.

\end{theorem*}

\begin{proof}
By (\ref{eqY}) and (\ref{eqYConstant}), we can define some number $\epsilon \in (0,1)$ such that 

\bea
 \sup_{\theta_0 \in K, 1 \leq i \leq B_n} h_n^{\epsilon(\frac{p}{2} \wedge \kappa )}\espc\l[ \sup_{\theta \in K} |\mathbb{Y}_{i,n}(\theta,\theta_0)-\mathbb{Y}(\theta,\theta_0)|^p  \r] \to^\proba 0,
 \label{eqY2}
\eea
and as $\widehat{\Theta}_{i,n}$ is also a maximizer of $\theta \to \mathbb{Y}_{i,n}(\theta,\theta_0)$, (\ref{eqY2}) implies the uniform consistency in the block index $i$ and the initial value of $\widehat{\Theta}_{i,n}$ to $\theta_0^{i,n,*}$, i.e.
\bea
\sup_{\theta_0 \in K, 1 \leq i \leq B_n} \probac\l[\widehat{\Theta}_{i,n} - \theta_0  \r] \to^\proba 0,
\eea
since $\mathbb{Y}$ satisfies the non-degeneracy condition [A4] in \cite{Clinet20171800}.  From (\ref{eqGamma}) and (\ref{eqGammaConstant}) we deduce 
\bea 
\sup_{\theta_0 \in K, 1 \leq i \leq B_n} h_n^{\epsilon(\frac{p}{2} \wedge \kappa )}\espc \l| \Gamma_{i,n}(\theta_0) - \Gamma(\theta_0) \r|^p  \to^\proba 0.
\label{eqGamma2}
\eea

By (\ref{eqDeltaDev}), $\Delta_{i,n}(\theta_0)$ and $\Delta_{i,n}^c(\theta_0)$ have the same asymptotic distribution, which is of the form $\Gamma(\theta_0)^{\half} \xi$, where $\xi$ follows a standard normal distribution. Following the proof of Theorem 3.11 in \cite{Clinet20171800}, we deduce that $\sqrt{h_n } (\widehat{\Theta}_{i,n}-\theta_0)$ converges uniformly in distribution to $T^{-\half}\Gamma(\theta_0)^{-\half} \xi$ when $\theta_0^{i,n,*} = \theta_0$, i.e.  
\bea 
\sup_{\theta_0 \in K, 1 \leq i \leq B_n}\l\{\espc \left[ f\l(\sqrt{h_n }(\widehat{\Theta}_{i,n}-\theta_0)\r) \right] - \esp\l[f\l(T^{-\half}\Gamma(\theta_0)^{-\frac{1}{2}}\xi\r) \r]\r\} \to^\proba 0,
\label{convDistribution}
\eea
for any bounded continuous function $f$ . 

\smallskip
Finally, we  extend (\ref{convDistribution}) to the case of a function of polynomial growth of order smaller than $L$. First note that by (\ref{eqDeltaDev}) and (\ref{eqDeltaDevConstant}) we have for any $L^{'} \in (L,2\kappa)$
\bea 
 \sup_{\theta_0 \in K, 1 \leq i \leq 
 B_n} \espc \l|\Delta_{i,n}(\theta_0)  \r|^{L^{'}} = O_\proba(1).
\label{eqDeltaDev2}
\eea
We now adopt the notations of \cite{YoshidaPolynomial2011} and define $\beta_1 = \frac{\epsilon}{2}$, $\beta_2=\frac{1}{2} - \beta_1$, $\rho =2$,  $0 < \rho_2 < 1- 2\beta_2$, $0< \alpha < \frac{\rho_2}{2}$, and $0<\rho_1<\min \{1, \frac{\alpha}{1-\alpha},\frac{2\beta_1}{1-\alpha}\} $ all sufficiently small so that $M_1 = L(1-\rho_1)^{-1} < L^{'} $, $ M_4 = \beta_1 L(\frac{2\beta_1}{1-\alpha} - \rho_1)^{-1} < 2\frac{ \gamma(\delta-1)}{2} = \kappa$, $M_2 = (\half -\beta_2)L(1- 2\beta_2 -\rho_2)^{-1} < \kappa$ and finally $M_3 = L\l(\frac{\alpha}{1-\alpha} - \rho_1 \r)^{-1} < \infty$. Then, by (\ref{eqY2}), (\ref{eqGamma2}), (\ref{eqDeltaDev2}) and finally (\ref{eqPartial}), conditions $[A1^{''}]$, $[A4^{'}]$, $[A6]$, $[B1]$ and $[B2]$ in \cite{YoshidaPolynomial2011} are satisfied. It is straightforward that we can apply a conditional version (with respect to the operator $\espc$) of Theorem 3 and Proposition 1 from \cite{YoshidaPolynomial2011}  to get that for any $p \leq L$,
\bea 
 \sup_{\theta_0 \in K, 1 \leq i \leq \Delta_n^{-1}} \espc  \l|\sqrt{h_n }\l(\widehat{\Theta}_{i,n}-\theta_0\r)\right|^p = O_\proba(1).
\eea
Such stochastic boundedness of conditional moments along with the convergence in distribution is clearly sufficient to imply the theorem.

\end{proof}
So far we have focused on the case where $R_{i,n}(0)$ is bounded by the sequence $M_n$. Nonetheless, the time-varying parameter Hawkes process has a residual which is a priori not bounded at the beginning of a block. In Theorem \ref{conditionalCLT}, we relax this assumption. In addition, we use regular conditional distribution techniques (see for instance Section $4.3$ (pp. $77-80$) in \cite{breiman1992probability}) to obtain (\ref{momentconv}) when not conditioning by any particular starting value of $\theta_t^*$. We provide the formal proof in what follows. Recall that $\esp_{(i-1) \Delta_n}$ stands for $\esp[.| \calf_0^{i,n}]$.
\begin{proof}[Proof of Theorem \ref{conditionalCLT}.]
We can decompose $\esp_{(i-1) \Delta_n} \Big[ f \big( \sqrt{h_n}(\widehat{\Theta}_{i,n} - \theta_{(i-1) \Delta_n}^*) \big) \Big]$ as
\begin{eqnarray}
\label{t1} & & \esp_{(i-1) \Delta_n} \l[ f \l( \sqrt{h_n }\l(\widehat{\Theta}_{i,n} - \theta_{(i-1) \Delta_n}^*\r) \r) \mathbf{1}_{ \{ R_{i,n}(0) \leq M_n \} } \r] \\ 
\label{t2} & + & \esp_{(i-1) \Delta_n} \l[ f \l( \sqrt{h_n}\l(\widehat{\Theta}_{i,n} - \theta_{(i-1) \Delta_n}^*\r) \r) \mathbf{1}_{ \l\{ R_{i,n}(0) > M_n \r\} } \r].
\end{eqnarray}
Let $\xi$ as in Theorem \ref{conditionalCLT}. On the one hand by a regular conditional distribution argument, if we define $G(\theta_0) = \espc \Big[ f \big( \sqrt{h_n}(\widehat{\Theta}_{i,n} - \theta_0 \big) \Big]- \esp \Big[ f \big( T^{-\half}\Gamma(\theta_0)^{-\frac{1}{2}}\xi \big) \Big] $, we can express uniformly in $i \in \{ 1, \cdots, B_n \}$ the quantity
\bea 
&&\esp_{(i-1) \Delta_n} \l[ f \l( \sqrt{h_n}\l(\widehat{\Theta}_{i,n} - \theta_{(i-1) \Delta_n}^*\r) \r) \mathbf{1}_{ \{ R_{i,n}(0) \leq M_n \} } - f \l( T^{-\half}\Gamma\l(\theta_{(i-1) \Delta_n}^*\r)^{-\frac{1}{2}}\xi \r) \r] 
\label{eqDiff}
\eea 
as $G\l(\theta_{(i-1) \Delta_n}^*\r)$ by definition of $\espc$ and because $\xi \ind \calf$. We note that
\bea 
\l|G\l(\theta_{(i-1) \Delta_n}^*\r)\r| \leq \sup_{\theta_0 \in K} \l|\espc \left[ f\l(\sqrt{h_n}\l(\widehat{\Theta}_{i,n}-\theta_0\r)\r) \right] - \esp\l[f\l(T^{-\half}\Gamma(\theta_0)^{-\frac{1}{2}}\xi\r) \r]\r|,
\label{proofSup}
\eea 

%\bea 
%\label{proofint}
 %\int_{K}{\l\{\espc \left[ f(\sqrt{h_n}(\widehat{\Theta}_{i,n}-\theta_0)) \right] - \esp[f(\Gamma(\theta_0)^{-\frac{1}{2}}\xi) ]\r\} \proba^{\theta_{(i-1) \Delta_n}^*}(d\theta_0)}.
%\eea 
%where $\proba^{\theta_{(i-1) \Delta_n}^*}$ is the distribution of the random variable $\theta_{(i-1) \Delta_n}^*$. 
take the sup over $i$ in (\ref{proofSup}), and in view of Theorem \ref{convergence}, we have shown that (\ref{eqDiff}) is uniformly of order $o_\proba(1)$.

\smallskip
On the other hand, (\ref{t2}) is bounded by $h_n^{L} Q  \mathbf{1}_{ \{ R_{i,n}(0) > M_n \} } $ for some $Q >0$, where we have used that $\widehat{\Theta}_{i,n}$ takes its values in a compact space.  By a straightforward computation it is easy to see that $\proba \l[ R_{i,n}(0) > M_n \r] \leq \proba \l[ \lambda_*^n((i-1)\Delta_n) > M_n \r]$, which in turn can be dominated easily with Markov's inequality by $M_n^{-1} \esp \l[ \lambda_*^n((i-1)\Delta_n) \r] = O(nM_n^{-1})$. We recall that $M_n$ is of the form $n^q$ where $q$ can be taken arbitrarily big, and we have thus shown that (\ref{t2}) vanishes asymptotically.
\end{proof}

\subsection{Bias reduction of the local MLE} \label{sectionBiasCorrection}
We go one step further and study the properties of the asymptotic conditional bias of the local MLE, i.e. the quantity $\esp_{(i-1)\Delta_n}\l[\widehat{\Theta}_{i,n} - \theta_{(i-1) \Delta_n}^*\r]$. We then derive the expression of a bias-corrected estimator $\widehat{\Theta}_{i,n}^{(BC)}$ whose expectation tends faster to $\theta_{(i-1) \Delta_n}^*$.  

\smallskip
We start by estimating the order of the bias of the local MLE. As the reader can see, the following computations are very involved. Therefore, in this section only, we adopt the following notation conventions. First, we drop the index reference $i$. Consequently, all the variables $N^n$,$ \lambda_*^n$,$ l_n$, $\esps$, etc. should be read $N^{i,n}$,$ \lambda_*^{i,n}$,$ l_{i,n}$, $\espc$, etc. All the results are implicitly stated uniformly in the block index. Second, for a random variable $Z$ that admits a first order moment for the operator $\esps$, we denote by $\overline{Z}$ its centered version, i.e. the random variable $Z - \esps[Z]$. We adopt Einstein's summation convention, i.e. any indice that is repeated in an expression is implicitly summed. For example the expression $a_{ij}b_j$ should be read $\sum_j a_{ij}b_j$. Finally, as in Section \ref{mainResults}, for a matrix $M$, we use superscripts to designate elements of its inverse, i.e. $M^{ij}$ stands for the element in position $(i,j)$ of $M^{-1}$ when it is well-defined, $M^{ij} = 0$ otherwise. 

\smallskip
By a Taylor expansion of the score function around the maximizer of the likelihood function, it is immediate to see that there exists $\xi_n \in [\widehat{\Theta}_n, \theta_0 ]$ such that 
\bea 
0 = \partial_\theta l_n(\widehat{\Theta}_n) = \partial_\theta l_n(\theta_0) + \partial_\theta^2 l_n(\theta_0)(\widehat{\Theta}_n- \theta_0) + \half \partial_\theta^3 l_n(\xi_n)(\widehat{\Theta}_n- \theta_0)^{\otimes 2},
\eea 
where $\partial_\theta^3 l_n(\xi_n)(\widehat{\Theta}_n- \theta_0)^{\otimes 2}$ is a compact expression for the vector whose $i$-th component is $\partial_{\theta,ijk}^3 l_n(\xi_n)(\widehat{\Theta}_n- \theta_0)_j(\widehat{\Theta}_n- \theta_0)_k$. Let $\epsilon \in (0,1)$. By application of Lemmas \ref{boundedDeviation} and \ref{lemmaBoundedDerivatives}, it still holds that
\bea 
 \partial_\theta l_n^c(\theta_0) + \partial_\theta^2 l_n^c(\theta_0)(\widehat{\Theta}_n- \theta_0) + \half\partial_\theta^3 l_n^c(\xi_n)(\widehat{\Theta}_n- \theta_0)^{\otimes 2} = O_{\proba}\l(h_n^{1-\epsilon\kappa}\r),
 \label{taylorConstant}
\eea 
where the residual term $O_{\proba}\l(h_n^{1-\epsilon\kappa}\r)$ admits clearly moments of any order with respect to $\esps$. We now apply the operator $\esps$, divide by $h_n T$ and obtain 
\beas 
  \esps[\overline{\Gamma}_n^c(\theta_0)(\widehat{\Theta}_n- \theta_0)]&+& \esps[\Gamma_n^c(\theta_0)]\esps[\widehat{\Theta}_n- \theta_0] - \esps\l[\frac{\partial_\theta^3 l_n^c(\xi_n)}{2h_n T}(\widehat{\Theta}_n- \theta_0)^{\otimes 2}\r] = O_\proba(h_n^{-\epsilon\kappa}),
  \label{espExpansion}
\eeas
where the expectation of the first term has vanished because of the martingale form of $\Delta_n^c(\theta_0)$ in (\ref{martingaleDeltaC}). The term $\esps[\Gamma_n^c(\theta_0)]\esps[\widehat{\Theta}_n- \theta_0]$ is of interest since it contains the quantity we want to evaluate. The first and the third terms have thus to be evaluated to derive an expansion of the bias. We start by the first term, i.e. the covariance between our estimator and $\Gamma_n^c(\theta_0)$. To compute the limiting value of such covariance, we consider the martingale $M_{n}^{c}(t,\theta_0) =   \int_0^{t}{ \frac{\partial_\theta\lambda^{n,c}(s,\theta_0)}{\lambda^{n,c}(s,\theta_0)} \{dN_s^{n,c} - \lambda^{n,c}(s,\theta_0) ds\}}$, and we define the empirical covariance processes $C_n^c(\theta_0)$ and $Q_n^c(\theta_0)$ whose components are, for any triplet $(i,j,k) \in \reels^3 \times \reels^3 \times \reels^3$,
 
 \beas 
 C_n^c(\theta_0)_{i,jk} = \inv{h_n T} \int_0^{h_n T}{\partial_{\theta,i}\lambda^{n,c}(s,\theta_0) \partial_{\theta,jk}^2 \text{log} \lambda^{n,c}(s,\theta_0)ds},
 \eeas
 and
 \beas 
 Q_n^c(\theta_0)_{i,jk} = -\frac{M_{n}^{c}(T,\theta_0)_i}{h_n T} \int_0^{h_n T}{  \frac{\partial_\theta \lambda^{n,c}(s,\theta_0)_j\partial_\theta \lambda^{n,c}(s,\theta_0)_k}{\lambda^{n,c}(s,\theta_0)} ds},
  \label{CovarianceExpression}
 \eeas
 
 We define in a similar way $C_n(\theta_0)$ and $Q_n(\theta_0)$. The next lemma clarifies the role of $C_n^c(\theta_0)+Q_n^c(\theta_0)$ and is a straightforward calculation.
 
 \begin{lemma*}
 We have 
\bea 
 \esps \l[ C_n^c(\theta_0)_{i,jk}+Q_n^c(\theta_0)_{i,jk} \r] = -\sqrt{h_n T}\esps \l[ \Delta_n^c(\theta_0)_i \Gamma_n^c(\theta_0)_{jk} \r].
 \eea
 \label{lemmaCovariance}
 \end{lemma*}

\begin{proof}
Note that for two $\mathbb{L}_2$ bounded processes $(u_s)_s$, $(v_s)_s$, we have 
\beas 
\l\langle \int_0^{.}{u_s \{dN_s^{n,c} - \lambda^{n,c}(s,\theta_0) ds\}},\int_0^{.}{v_s \{dN_s^{n,c} - \lambda^{n,c}(s,\theta_0) ds\}} \r\rangle_t = \int_0^{t}{u_s v_s  \lambda^{n,c}(s,\theta_0) ds}
\eeas
 Taking expectation, this yields
\beas  
 \esps\l[ \int_0^{t}{u_s \{dN_s^{n,c} - \lambda^{n,c}(s,\theta_0) ds\}}\int_0^{t}{v_s \{dN_s^{n,c} - \lambda^{n,c}(s,\theta_0) ds\}}\r] = \esps \l[ \int_0^{t}{u_s v_s  \lambda^{n,c}(s,\theta_0) ds }\r]
 \eeas
 Formula (\ref{CovarianceExpression}) is then obtained directly from the expression of $\Gamma_n^c(\theta_0)$ and $ \Delta_n^c(\theta_0)$. 
\end{proof}

Now, by the same argument as for the proof of (\ref{lemmaConvConstant}), we have for any integer $p \geq 1$ and any $\epsilon \in (0, 1)$, 
\bea
\sup_{\theta_0 \in K}  h_n^{\epsilon \frac{p}{2}} \esp_{\theta_0,n}\l|C_n^c(\theta_0) - C(\theta_0)\r|^p \to^\proba 0,
\label{eqCConstant}
\eea
and
\bea 
\sup_{\theta_0 \in K}  h_n^{\epsilon \frac{p}{2}} \l|\esp_{\theta_0,n}\l[Q_n^c(\theta_0) - Q(\theta_0)\r]\r|^p \to^\proba 0
\label{eqQConstant}
\eea 
where $C$ and $Q$ were defined respectively in (\ref{paramC}) and (\ref{paramQ}). Before we turn to the limiting expression of the term 
$$\esps[\overline{\Gamma}_n^c(\theta_0)(\widehat{\Theta}_n- \theta_0)]_i$$ 
in our expansion of the bias in terms of $C(\theta_0)+Q(\theta_0)$, we need to control the convergence of $\Gamma_n^c(\theta_0)^{-1}$ toward $\Gamma(\theta_0)^{-1}$. We define $c_0 = \min_{\theta_0 \in K} \min \{ c \in \reels_+ | \forall x \in \reels^3 - \{0\}, x^T \Gamma(\theta_0) x \geq c |x|_2^2   > 0 \}  $, the smallest eigenvalue of all the matrices $\Gamma(\theta_0)$.  We consider the sequence of events $\mathbb{B}_n(\theta_0) =  \{\forall x \in \reels^3 - \{0\}, x^T \Gamma_n^c(\theta_0) x \geq \frac{c_0}{2} |x|_2^2\}$, and their complements $\mathbb{B}_n(\theta_0)^c$. 
\begin{lemma*}
We have, for any integer $p \geq 1$ and any $\epsilon \in (0,1)$ that 
\bd
		\im[{\bf (i)}] $\sup_{\theta_0 \in K} \probas \l[ \mathbb{B}_n(\theta_0)^c\r]= O_\proba\l(h_n^{-\epsilon \frac{p}{2}}\r)$.
		\im[{\bf (ii)}] $\sup_{\theta_0 \in K}  h_n^{\epsilon \frac{p}{2}} \esps \l[ \l| \Gamma_n^c(\theta_0)^{-1} - \Gamma(\theta_0)^{-1} \r| \mathbf{1}_{\mathbb{B}_n}\r] \to^\proba 0$.
\ed
\label{lemmaConvInv}
\end{lemma*}

\begin{proof}
We start by showing {\bf (i)}. We recall that in our notation convention, the symbol $|x|$ stands for $\sum_i |x_i|$ for any vector or matrix. Clearly, we have that 
\bea 
\probas \l[ \mathbb{B}_n(\theta_0)^c\r] &\leq& \probas  \l\{\forall x \in \reels^3 - \{0\}, \frac{|x^T (\Gamma_n^c(\theta_0) - \Gamma(\theta_0)) x|}{|x|_2^2} > \frac{c_0}{2} \r\},
\label{dominationEigenValue}
\eea
and by equivalence of the norms $|M|$ and $\sup_{x \in \reels^3 - \{0\}}\frac{|x^T M x|}{|x|_2^2}$ on the space of symmetric matrices of $\reels^3$,  (\ref{dominationEigenValue}) implies the existence of some constant $\eta > 0$ such that 
\beas 
\probas \l[ \mathbb{B}_n(\theta_0)^c\r] &\leq& \probas  \l[|\Gamma_n^c(\theta_0) - \Gamma(\theta_0)| > \eta c_0 \r]\\
&\leq& (\eta c_0)^{-p} \esps  |\Gamma_n^c(\theta_0) - \Gamma(\theta_0)|^p,
\eeas
where Markov's inequality was used at the last step. {\bf (i)} thus follows from (\ref{eqGamma2}). Moreover, {\bf (ii)} is easily obtained using the elementary result $|A^{-1} - B^{-1}| = |B^{-1}(B-A)A^{-1}| \leq |A^{-1}|_{\infty}|B^{-1}|_{\infty} |B-A|$ applied to $ \Gamma_n^c(\theta_0)$ and $\Gamma(\theta_0)$ on the set $\mathbb{B}_n(\theta_0)$.
\end{proof}

\begin{lemma*}
Let $\epsilon \in (0,1)$ and $i \in \{0,1,2\}$. The following expansion holds.
\bea 
\esps[\overline{\Gamma}_n^c(\theta_0)(\widehat{\Theta}_n- \theta_0)]_i = -\frac{\Gamma(\theta_0)^{jk}\l\{C(\theta_0)_{k,ij}+Q(\theta_0)_{k,ij}\r\}}{h_nT} + O_\proba\l(h_n^{-\epsilon \l(\kappa \wedge \frac{3}{2}\r)}\r).
\eea
\label{ordre2}
\end{lemma*}

\begin{proof}
Note first that in view of Lemma \ref{lemmaConvInv} {\bf (i)} along with H\"{o}lder's inequality, we have that $\esps[\overline{\Gamma}_n^c(\theta_0)(\widehat{\Theta}_n- \theta_0)] = \esps[\overline{\Gamma}_n^c(\theta_0)(\widehat{\Theta}_n- \theta_0)\mathbf{1}_{\mathbb{B}_n(\theta_0)}] + o_\proba\l(h_n^{- \frac{3}{2}}\r)$. Thus we can assume without loss of generality the presence of the indicator of the event $\mathbb{B}_n(\theta_0)$ in the expectation of the left-hand side of (\ref{ordre2}). Take $\epsilon \in (0,1)$ and $\tilde{\epsilon} \in (\epsilon,1)$. As a consequence of (\ref{taylorConstant}), we have the representation,
\bea 
\widehat{\Theta}_n- \theta_0 = \inv{\sqrt{h_nT}}\Gamma_n^c(\theta_0)^{-1}\Delta_n^c(\theta_0) +\Gamma_n^c(\theta_0)^{-1}\frac{\partial_\theta^3 l_n^c(\xi_n)(\widehat{\Theta}_n- \theta_0)^{\otimes 2}}{2h_nT} + O_{\proba}\l(h_n^{-\tilde{\epsilon}\kappa}\r),
\label{expansionTheta}
\eea
 on the set $\mathbb{B}_n(\theta_0)$, where the residual term $O_{\proba}\l(h_n^{-\tilde{\epsilon}\kappa}\r)$ admits moments of any order with respect to the operator $\esps$. We inject (\ref{expansionTheta}) in the expectation and get
\beas 
\esps[\overline{\Gamma}_n^c(\theta_0)(\widehat{\Theta}_n- \theta_0)] &=& \inv{\sqrt{h_nT}}\esps[\overline{\Gamma}_n^c(\theta_0)\Gamma_n^c(\theta_0)^{-1}\Delta_n^c(\theta_0)\mathbf{1}_{\mathbb{B}_n(\theta_0)}]\\
&+& \esps \l[\overline{\Gamma}_n^c(\theta_0)\Gamma_n^c(\theta_0)^{-1}\frac{\partial_\theta^3 l_n^c(\xi_n)(\widehat{\Theta}_n- \theta_0)^{\otimes 2}}{2h_nT}\mathbf{1}_{\mathbb{B}_n(\theta_0)}\r]\\
&+& O_\proba(h_n^{-\epsilon \kappa}),
\eeas
where the residual term $O_\proba(h_n^{-\epsilon \kappa})$ is obtained by H\"{o}lder's inequality using the fact that $\epsilon < \tilde{\epsilon}$. By Lemma \ref{lemmaConvInv} {\bf (ii)}, the first term admits the expansion
\bea 
\inv{\sqrt{h_nT}}\esps[\overline{\Gamma}_n^c(\theta_0)\Gamma(\theta_0)^{-1}\Delta_n^c(\theta_0)] + O_\proba\l(h_n^{-\frac{3\epsilon}{2}}\r)
\label{expansionCov},
\eea
where we used H\"{o}lder's inequality to control $\inv{\sqrt{h_nT}}\esps\l[\overline{\Gamma}_n^c(\theta_0)(\Gamma_n^c(\theta_0)^{-1} -\Gamma(\theta_0)^{-1}) \Delta_n^c(\theta_0)\r]$ and we neglected the effect of the indicator function by Lemma \ref{lemmaConvInv} {\bf (i)}. For any $i \in \{0,1,2\}$, we develop the matrix product in (\ref{expansionCov}), use Lemma \ref{lemmaCovariance} along with (\ref{eqCConstant}), and this leads to the estimate
\bea
\inv{\sqrt{h_nT}}\esps[\overline{\Gamma}_n^c(\theta_0)\Gamma(\theta_0)^{-1}\Delta_n^c(\theta_0)]_i = \frac{\Gamma(\theta_0)^{jk}\l\{C(\theta_0)_{k,ij}+Q(\theta_0)_{k,ij}\r\}}{h_nT} + O_\proba\l(h_n^{-\frac{3\epsilon}{2}}\r).
\eea
It remains to control the term $\esps \l[\overline{\Gamma}_n^c(\theta_0)\Gamma_n^c(\theta_0)^{-1}\frac{\partial_\theta^3 l_n^c(\xi_n)(\widehat{\Theta}_n- \theta_0)^{\otimes 2}}{2h_nT}\mathbf{1}_{\mathbb{B}_n(\theta_0)}\r]$. Take $L \in (2,2\kappa)$. By boundedness of moments of $h_n^{\frac{\epsilon}{2}}\overline{\Gamma}_n^c(\theta_0)_{ij}\Gamma_n^c(\theta_0)^{jk}\frac{\partial_{\theta,klm}^3 l_n^c(\theta)}{2h_nT}\mathbf{1}_{\mathbb{B}_n(\theta_0)}$, for any $(i,j,k,l,m)$ and uniformly in $\theta_0 \in K$, we have
$$\esps \l[\overline{\Gamma}_n^c(\theta_0)_{ij}\Gamma_n^c(\theta_0)^{jk}\frac{\partial_{\theta,klm}^3 l_n^c(\xi_n)(\widehat{\Theta}_n- \theta_0)_l(\widehat{\Theta}_n- \theta_0)_m}{2h_nT}\mathbf{1}_{\mathbb{B}_n(\theta_0)}\r] $$
\beas
&\leq& Kh_n^{-\frac{\epsilon}{2}}\esps \l[\l|(\widehat{\Theta}_n- \theta_0)_l(\widehat{\Theta}_n- \theta_0)_m\r|^{\frac{L}{2}}\r]^{\frac{2}{L}}\\
&=& O_\proba\l(h_n^{-\frac{3\epsilon}{2} }\r),
\eeas
where  H\"{o}lder's inequality was applied for the first inequality, and Theorem \ref{convergence} was used with the function $f : x \to \l(x_l x_m\r)^{\frac{L}{2}}$, which is of polynomial growth of order $L$, to get the final estimate.  
\end{proof}
Finally, we derive the expansion of $\inv{2h_n T}\esps[\partial_\theta^3 l_n^c(\xi_n)(\widehat{\Theta}_n- \theta_0)^{\otimes 2}]$. First note that for any integer $p \geq 1$ and any $\epsilon \in (0,1)$, 
\bea
\sup_{\theta_0 \in K}  h_n^{\epsilon \frac{p}{2}} \esp_{\theta_0,n}\l|\inv{h_n T} \partial_\theta^3 l_n^c(\theta_0) - K(\theta_0)\r|^p \to^\proba 0,
\label{eqKConstant}
\eea
where $K(\theta_0)$ was introduced in (\ref{paramK}). The next lemma is proved the same way as for Lemma \ref{ordre2}.

\begin{lemma*}
Let $\epsilon \in (0,1)$ and $i \in \{0,1,2\}$. We have the expansion
\bea 
\inv{2h_n T}\esps[\partial_\theta^3 l_n^c(\xi_n)(\widehat{\Theta}_n- \theta_0)^{\otimes 2}]_i =  \frac{\Gamma(\theta_0)^{jk}K(\theta_0)_{ijk}}{2h_nT} + O_\proba\l(h_n^{-\epsilon \l(\kappa \wedge \frac{3}{2}\r)}\r).
\eea
\label{ordre3}
\end{lemma*}

\begin{proof}
 Consider three indices $i,j,k \in \{0,1,2\}$ and $\epsilon \in (0,1)$. We have the decomposition 
\beas 
\inv{2h_n T}\esps[\partial_{\theta,ijk}^3 l_n^c(\xi_n)(\widehat{\Theta}_n- \theta_0)_j(\widehat{\Theta}_n- \theta_0)_k] &=& \inv{2h_n T}\esps[\partial_{\theta,ijk}^3 l_n^c(\xi_n)]\esps[(\widehat{\Theta}_n- \theta_0)_j(\widehat{\Theta}_n- \theta_0)_k] \\
&+& \inv{2h_n T}\esps[\overline{\partial_{\theta,ijk}^3 l_n^c(\xi_n)}(\widehat{\Theta}_n- \theta_0)_j(\widehat{\Theta}_n- \theta_0)_k].
\eeas
We now remark that the first term admits the expansion
\bea 
\frac{\Gamma(\theta_0)^{jk}K(\theta_0)_{ijk}}{2h_nT} + O_\proba\l(h_n^{ -\epsilon \l(\kappa \wedge \frac{3}{2}\r)}\r),
\eea
 by replacing  $\esps[(\widehat{\Theta}_n- \theta_0)_j(\widehat{\Theta}_n- \theta_0)_k]$ and $\inv{2h_n T}\esps[\partial_{\theta,ijk}^3 l_n^c(\xi_n)]$ by their estimates
\bea 
\esps[(\widehat{\Theta}_n- \theta_0)_j(\widehat{\Theta}_n- \theta_0)_k] = \frac{\Gamma(\theta_0)^{jk}}{h_n T} + O_\proba\l(h_n^{-\epsilon \l(\kappa \wedge \frac{3}{2}\r)}\r),
\label{ExpansionQuad}
\eea
and
\bea 
\inv{2h_n T}\esps\l[\partial_{\theta,ijk}^3 l_n^c(\xi_n)\r] = K(\theta_0)_{ijk} + O_\proba \l( h_n^{- \frac{\epsilon}{2}} \r).
\label{ExpansionDeriv3}
\eea
(\ref{ExpansionQuad}) is obtained by injecting the expansion of $\widehat{\Theta}_n- \theta_0$ in (\ref{expansionTheta}) up to the first order only, and (\ref{ExpansionDeriv3}) is a consequence of (\ref{eqKConstant}) and the uniform boundedness of moments of $\frac{\partial_{\theta}^4 l_n^c(\theta_0)}{h_nT}$ in $\theta_0 \in K$ by Lemma \ref{lemmaBoundedDerivatives} {\bf (ii)}. Note that the expansion (\ref{ExpansionQuad}) is not a direct consequence of Theorem \ref{convergence} applied to $x \to x_jx_k$ since this would lead to the weaker estimate $\frac{\Gamma(\theta_0)^{jk}}{h_n T} + o_\proba(h_n^{-1})$ instead. Finally, the second term is of order $O_\proba\l(h_n^{- \frac{3\epsilon}{2}}\r)$ by H\"{o}lder's inequality along with Theorem \ref{convergence}, and thus we are done.
\end{proof}

Before we turn to the final theorem, we recall for any $j \in \{0,1,2\}$ the expression 
\bea 
b(\theta_0)_j = \half \Gamma(\theta_0)^{ij}\Gamma(\theta_0)^{kl}(K(\theta_0)_{ikl}+ 2\l\{C(\theta_0)_{k,il}+Q(\theta_0)_{k,il}\r\}),
\eea
which was defined in (\ref{paramb}). We are now ready to state the general theorem on bias correction of the local MLE, which we formulate with the block index $i$.

\begin{theorem*} \label{biasCorrection}
Let $\epsilon \in (0,1)$. The bias of the estimator $\widehat{\Theta}_{i,n}$ has the expansion
\bea 
\espc\l[\widehat{\Theta}_{i,n} - \theta_0\r] = \frac{b(\theta_0)}{h_n T} + O_\proba\l(h_n^{-\epsilon \l(\kappa \wedge \frac{3}{2}\r)}\r), 
\label{biasExpansion}
\eea 
uniformly in $i \in \{1, \cdots, B_n\}$ and in $\theta_0 \in K$. Moreover, the bias-corrected estimator $\widehat{\Theta}_{i,n}^{(BC)}$ defined in (\ref{localMLEBC}) has the (uniform) bias expansion
\bea 
\espc\l[\widehat{\Theta}_{i,n}^{BC} - \theta_0\r] =  O_\proba \l(h_n^{-\epsilon \l(\kappa \wedge \frac{3}{2}\r)}\r). 
\label{biasCorrection3}
\eea 
\end{theorem*}

\begin{proof} 
We drop the index $i$ in this proof. Take $\epsilon \in (0,1)$ and some $j \in \{0,1,2\}$. By Lemma \ref{ordre2} and Lemma \ref{ordre3}, we have
\beas 
   \esps\l[\Gamma_n^c(\theta_0)\r]_{jk}\esps\l[\widehat{\Theta}_n- \theta_0\r]_k =  \frac{\Gamma(\theta_0)^{kl}\l(K(\theta_0)_{jkl} + 2\l\{C(\theta_0)_{l,jk}+Q(\theta_0)_{l,jk}\r\}\r)}{2h_nT} + O_\proba\l(h_n^{-\epsilon \l(\kappa \wedge \frac{3}{2}\r)}\r),
\eeas
which is a set of simultaneous linear equations. After inversion of this system of equations and application of Lemma \ref{lemmaConvInv}, the expression of the bias becomes for $j \in \{0,1,2\}$,

\beas 
 \esps\l[\widehat{\Theta}_n- \theta_0\r]_j =  \frac{ \Gamma(\theta_0)^{ij}\Gamma(\theta_0)^{kl}\l(K(\theta_0)_{ikl}+ 2\l\{C(\theta_0)_{k,il}+Q(\theta_0)_{k,il}\r\}\r)}{2h_nT} + O_\proba\l(h_n^{-\epsilon \l(\kappa \wedge \frac{3}{2}\r)}\r),
\eeas
which is exactly (\ref{biasExpansion}). Finally, a calculation similar to the proofs of Lemmas \ref{ordre2} and \ref{ordre3} shows that

\bea 
\esps b\l(\widehat{\Theta}_n\r)  = \esps b\l(\theta_0\r) + O_\proba\l(h_n^{-\frac{\epsilon}{2}}\r) 
\eea
so that we have (\ref{biasCorrection3}) and this concludes the proof.
\end{proof}

We conclude by showing the version of the preceding theorem in terms of $\esp_{(i-1)\Delta_n}$. 

\begin{proof}[Proof of Theorem \ref{biasCorrection2}] 
This follows exactly the same argument as for the proof of Theorem \ref{conditionalCLT}. 
\end{proof}

\subsection{Proof of the GCLT}
\label{GCLTproof}
In this section we present the proof of Theorem \ref{GCLT} using a similar martingale approach as in \cite{potiron2016estimating}. Using a different decomposition than ($34$) on p. 22 of the cited work, we obtain following the same line of reasoning as in the proof of (37) on p. 47-48 that a sufficient condition to show that the GCLT holds is
\begin{conditionCstar}
  We have uniformly in $i \in \{ 1, \cdots, B_n \}$ that there exists $\epsilon > 0$ such that
\begin{eqnarray}
  \label{GCLTass1} \var_{(i-1) \Delta_n} \l[  \sqrt{h_n}\l(\widehat{\Theta}_{i,n}^{(BC)} - \theta_{(i-1) \Delta_n}^*\r) \r] & = &  T^{-1}\Gamma\l(\theta_{(i-1) \Delta_n}^*\r)^{-1} + o_\proba (1), \\ 
 \label{GCLTass2} \esp_{(i-1) \Delta_n} \l[  \l|\sqrt{h_n}\l(\widehat{\Theta}_{i,n}^{(BC)}- \theta_{(i-1) \Delta_n}^* \r)\r|^{2+\epsilon} \r] & = &  O_\proba (1), \\
\label{GCLTass3} \esp_{(i-1) \Delta_n} \Big[ \widehat{\Theta}_{i,n}^{(BC)} - \theta_{(i-1) \Delta_n}^* \Big] & = &  o_\proba \l(n^{- 1/2}\r),
\end{eqnarray}
where for any $t \in [0, T]$ and any random variable $X$, $\var_{t} [X ] = \esp_t \big[ (X - \esp_t [X])^2 \big]$.
  \end{conditionCstar}
  
  The above-mentioned approach is based on techniques introduced in \cite{potiron2016estimating}, but it is much different and deeper. Indeed, \cite{potiron2016estimating} provides conditions which in this specific case are hard to verify due to the past correlation of the model. We choose to go through a different path. More specifically, the cited author uses a different decomposition than (\ref{cltCondition}). We thus obtain different conditions which are hard to verify, and this is the main goal of the proofs.

\begin{proof}[Proof of Theorem \ref{GCLT} under  \textnormal{[}$\textnormal{C}^*$\textnormal{]}]
We split the proof into two parts.  

{\bf Step 1}. The first part of the proof consists in showing that
\begin{eqnarray}
\label{paramsmooth}
\Theta = \frac{1}{B_n} \sum_{i=1}^{B_n} \theta_{(i-1)\Delta_n}^*  + o_\proba\l(n^{-1/2}\r).
\end{eqnarray}
Note that (\ref{paramsmooth}) is to be compared to (\ref{outlineparamsmooth}) for the toy model. Moreover, (\ref{paramsmooth}) was also shown in (35) on pp. 46-47 in \cite{potiron2016estimating}, but the parameter process was restricted to follow a continuous It\^{o}-process. To show (\ref{paramsmooth}), it is sufficient to show that
\begin{eqnarray}
\label{paramsmooth1}
\frac{\sqrt{n}}{B_n} \sum_{i=1}^{B_n} \Big| \theta_{(i-1)\Delta_n}^* - \Delta_n^{-1} \int_{(i-1) \Delta_n}^{i \Delta_n} \theta_s^* ds \Big| = o_\proba(1).
\end{eqnarray}
We can bound (\ref{paramsmooth1}) by
\begin{eqnarray}
\label{paramsmooth2}
\frac{\sqrt{n}}{B_n} \sum_{i=1}^{B_n}  \Delta_n^{-1} \int_{(i-1) \Delta_n}^{i \Delta_n} \underbrace{\Big| \theta_{(i-1)\Delta_n}^* - \theta_s^* \Big|}_{O_\proba( \Delta_n^{\gamma})} ds  = o_\proba(1),
\end{eqnarray}
where we used [C]-{\bf (i)} to obtain the order in (\ref{paramsmooth2}). Thus, we deduce that the left-hand side in (\ref{paramsmooth2}) is of order $O_\proba(h_n^{\gamma} n^{\frac{1}{2} - \gamma})$. In view of the left inequality in [BC] and the fact that $\gamma > \frac{1}{2}$, this vanishes asymptotically. Thus, we have proved (\ref{paramsmooth}).\\

{\bf Step 2}. We keep here the techniques and notations introduced in Section \ref{outline}, and replace $\widehat{\Theta}_{i,n}$ by the local estimator $\widehat{\Theta}_{i,n}^{(BC)}$ in the definitions of $M_{i,n}$ and $B_{i,n}$. To show the GCLT, we will show that $S_{n}^{(B)} \to^\proba 0$ and we will prove the existence of some $V_T$ such that $\calf_T^{\theta^{*}}$-stably in law, $S_{n}^{(M)} \to V_T^{\half} \mathcal{N} (0,1)$. Note that the former is a straightforward consequence of (\ref{GCLTass3}). To show the latter $S_{n}^{(M)} \to V_T^{\half} \mathcal{N} (0,1)$, we will use Theorem 3.2 of p. 244 in \cite{Jacod1997}. First, we show the conditional Lindeberg condition (3.13), i.e. in our case that for any $\eta > 0$ we have 
\begin{eqnarray}
\label{lindeberg}
\frac{n}{B_n^2} \sum_{i=1}^{B_n} \esp_{(i-1)\Delta_n} \l[ M_{i,n}^2 \mathbf{1}_{\l\{ \frac{\sqrt{n}}{B_n} M_{i,n} > \eta \r\}}\r] \rightarrow^\proba 0.
\end{eqnarray}
Let $\eta >0$. First, note that $\frac{n}{B_n} = h_n$. Using H\"{o}lder's inequality, we obtain that
$$h_n \esp_{(i-1)\Delta_n} \l[ M_{i,n}^2 \mathbf{1}_{\l\{ \frac{\sqrt{n}}{B_n} M_{i,n} > \eta \r\}}\r] \leq \underbrace{\l( \esp_{(i-1)\Delta_n} \l[ \l( \sqrt{h_n} M_{i,n} \r)^{2+\epsilon} \r] \r)^{\frac{2}{2+\epsilon}}}_{a_{i,n}} \underbrace{ \l( \esp_{(i-1)\Delta_n} \l[ \mathbf{1}_{\l\{ \frac{\sqrt{n}}{B_n} M_{i,n} > \eta \r\}}\r] \r)^{\frac{\epsilon}{2+\epsilon}}}_{b_{i,n}}. $$
 On the one hand we have that $a_{i,n}$ is uniformly bounded in view of (\ref{GCLTass2}) from [$\textnormal{C}^*$]. On the other hand, using also (\ref{GCLTass2}) along with [C]-{\bf (ii)}, we have that $b_{i,n}$ goes uniformly to $0$. We have thus proved (\ref{lindeberg}). We now prove the conditional variance condition (3.11), i.e. that
\begin{eqnarray}
\label{condvar}
\frac{n}{ B_n^2} \sum_{i=1}^{B_n} \esp_{(i-1)\Delta_n} \big[ M_{i,n}^2 \big] \rightarrow^\proba V_T :=T^{-2} \int_0^T \Gamma(\theta_s^*)^{-1} ds.
\end{eqnarray}
 We have that 
$$\frac{n}{ B_n^2} \sum_{i=1}^{B_n} \esp_{(i-1)\Delta_n} \big[ M_{i,n}^2 \big] = \frac{1}{T} \sum_{i=1}^{B_n} h_n \esp_{(i-1)\Delta_n} \big[ M_{i,n}^2 \big] \Delta_n.$$
We use Proposition I.4.44 on p.51 in \cite{JacodLimit2003} along with (\ref{GCLTass1}) from [$\textnormal{C}^*$] to show (\ref{condvar}). Now, conditions (3.10) and (3.12) are automatically satisfied because $M_{i,n}$ is a martingale increment and since we consider the reference continuous martingale $\boldsymbol{M} = 0$. Finally we show condition (3.14) to get the stable convergence. We thus consider a bounded $\F^{\theta^{*}}$-martingale $Z$, and we show that 

\bea 
\frac{\sqrt{n}}{B_n}\sum_{i=1}^{B_n}\esp_{(i-1)\Delta_n} \big[ M_{i,n} \Delta Z_{i,n} \big] \to^\proba 0,
\label{stableConvCondition}
\eea 
where $\Delta Z_{i,n} := Z_{i\Delta_n}-Z_{(i-1)\Delta_n}$. Using the Taylor expansion (\ref{taylorConstant}) and the boundedness of $Z$, by a similar calculation as in Lemma \ref{ordre2}, we have 

\beas 
\frac{\sqrt{n}}{B_n}\sum_{i=1}^{B_n}\esp_{(i-1)\Delta_n} \big[ M_{i,n} \Delta Z_{i,n} \big] = \frac{h_n}{\sqrt{n}} \sum_{i=1}^{B_n} \Gamma\l(\theta_{(i-1)\Delta_n}^{*}\r)^{-1}  \esp_{(i-1)\Delta_n} \l[\partial_\theta l_{i,n}^c\l(\theta_{(i-1)\Delta_n}^{*}\r)\Delta Z_{i,n}  \r] + o_\proba(1).
\eeas 
Note now that $l_{i,n}^c\l(\theta_{(i-1)\Delta_n}^{*}\r)$ can be written as an integral over the canonical Poisson martingale :

\beas 
l_{i,n}^c\l(\theta_{(i-1)\Delta_n}^{*}\r) = \int_0^{h_n T} \int_{\reels_+} \frac{\partial_\theta\lambda^{i,n,c}(s,\theta_{(i-1)\Delta_n}^{*})}{\lambda^{i,n,c}(s,\theta_{(i-1)\Delta_n}^{*})} \mathbb{1}_{\{0 \leq z \leq \lambda^{i,n,c}(s,\theta_{(i-1)\Delta_n}^{*}) \}} \l\{ \overline{N}^{i,n}(ds,dz) - \overline{\Lambda}^{i,n}(ds,dz)\r\},
\eeas 
with $\overline{\Lambda}^{i,n}(ds,dz) = ds \otimes dz$. We deduce from the above representation that $\esp_{(i-1)\Delta_n} \l[\partial_\theta l_{i,n}^c\l(\theta_{(i-1)\Delta_n}^{*}\r)\Delta Z_{i,n}  \r] = 0$, since both $\sigma$-fields $\calf_T^{\theta^{*}}$ and $\calf_T^{\overline{N}}$ are independent, so that $Z$ and $\overline{N}^{i,n} - \overline{\Lambda}^{i,n}$ are orthogonal. Thus (\ref{stableConvCondition}) holds. Thus, by Theorem 3.2 of \cite{Jacod1997}, we have the $\calf_T^{\theta^{*}}$-stable convergence in law of $S_n^{(M)}$ toward an $\calf_T^{\theta^{*}}$-conditional Gaussian limit with random variance $V_T$. In particular, we have that $V_T$ and $\mathcal{N} (0,1)$ in Theorem \ref{GCLT} are independent from each other.
\end{proof}

 We prove now that we can obtain (\ref{GCLTass1}), (\ref{GCLTass2}) and (\ref{GCLTass3}) in Condition [$C^*$]. First note that for any $L \in (0,2\kappa)$, a calculation gives 
 \beas  
 \esp_{(i-1)\Delta_n}  \l| \sqrt{h_n}\l( \widehat{\Theta}_{i,n}^{(BC)}-\widehat{\Theta}_{i,n}\r) \r|^L &=& h_n^{-\frac{L}{2}} T^{-L} \esp_{(i-1)\Delta_n} \l|b\l(\widehat{\Theta}_{i,n}\r)\r|^L  = O_\proba\l(h_n^{-\frac{L}{2}}\r)
 \eeas  
 uniformly in $i \in \{1,...,B_n\}$. Thus, combining the previous estimate with Theorem \ref{conditionalCLT}, we have shown that Theorem \ref{conditionalCLT} remains true if $\widehat{\Theta}_{i,n}$ is replaced by $\widehat{\Theta}_{i,n}^{(BC)}$. We will use this fact in the following. If we decompose the conditional variance in (\ref{GCLTass1}) as
 $$\esp_{(i-1) \Delta_n} \l[ \l( \sqrt{h_n}\l(\widehat{\Theta}_{i,n}^{(BC)} - \theta_{(i-1) \Delta_n}^*\r) \r)^2 \r] - \esp_{(i-1) \Delta_n} \l[ \sqrt{h_n} \l(\widehat{\Theta}_{i,n}^{(BC)} - \theta_{(i-1) \Delta_n}^*\r)  \r]^2,$$
then (\ref{GCLTass1}) follows from Theorem \ref{conditionalCLT}. Moreover, (\ref{GCLTass2}) is a direct consequence of Theorem \ref{conditionalCLT}. Finally, in view of (\ref{EqBiasCorrection2}) in Theorem \ref{biasCorrection2},  (\ref{GCLTass3}) holds if there exists $\epsilon \in (0,1)$ such that $\sqrt{n} = o_\proba\l(h_n^{\epsilon(\kappa \wedge \frac{3}{2})}\r) $. From the relation $\sqrt n = h_n^{\frac{\delta}{2}}$, this can be reexpressed as $\frac{\delta}{2} < \kappa \wedge \frac{3}{2}$. If we replace $\kappa$ by its expression, we get the two conditions $\frac{\delta}{2} < \gamma(\delta-1)$ and $\frac{\delta}{2} < \frac{3}{2}$, that is $\frac{\gamma}{\gamma-\half} < \delta < 3$. This is exactly condition [BC].

\subsection{Proof of Proposition \ref{propConsistency}}

\begin{proof}
Let $\gamma \in (0,1]$ and $\alpha \in (0, \frac{\gamma}{1+\gamma})$ and finally $\delta \in (1+\inv{\gamma},\inv{\alpha})$. We follow the proof of Theorem \ref{GCLT}. (\ref{GCLTass1}) and (\ref{GCLTass2}) are true since $\delta > 1+\inv{\gamma}$. Moreover, by assumption on $\delta$ and $\alpha$, (\ref{GCLTass3}) is replaced by $\esp_{(i-1) \Delta_n} \Big[ \widehat{\Theta}_{i,n}^{(BC)} - \theta_{(i-1) \Delta_n}^* \Big]  = O_\proba \l( n^{- \gamma(1-\delta^{-1}) \wedge \delta^{-1}}\r) = o_\proba \l( n^{-\alpha}\r)$. writing the decomposition 

\bea 
\frac{n^{\alpha}}{B_n} \sum_{i =1}^{B_n} \l( \widehat{\Theta}_{i,n}-\theta_{(i-1)\Delta_n}^* \r) = n^{\alpha-\half}  \l\{S_n^{(B)}+S_n^{(M)}\r\},
\eea 
we have $n^{\alpha-\half} S_n^{(M)} \to^\proba 0$ since the central limit theorem for $S_n^{(M)}$ is still valid and $\alpha < \half$. Finally $n^{\alpha-\half} S_n^{(B)} = o_\proba \l(1\r)$. This concludes the proof for $\widehat{\Theta}_n$. The proof for the bias corrected case follows the same path using $\esp_{(i-1) \Delta_n} \Big[ \widehat{\Theta}_{i,n}^{(BC)} - \theta_{(i-1) \Delta_n}^* \Big]  = O_\proba \l( n^{- \gamma(1-\delta^{-1}) \wedge \frac{3}{2}\delta^{-1}}\r)$ in lieu of the previous estimate.
\end{proof}

\subsection{Proof of Proposition \ref{propStudent}}%\ref{AV}}
\label{proofAV}
Note that for any $\theta \in K$, we have
\bea 
\partial_{\xi}^2 l_{i,n}\l(n^{-1}\xi\r)_{|\xi = n \theta} = n^{-2}\partial_{\theta}^2l_{i,n}(\theta),
\eea
and thus 
\beas 
n^{-1}\widehat{C}_n &=& \inv{B_n}\sum_{i=1}^{B_n}{ \partial_{\theta}^2l_{i,n}\l(\widehat{\Theta}_{i,n}\r)^{-1}h_n} \\
&=& \inv{T B_n}\sum_{i=1}^{B_n}{ \Gamma_{i,n}\l(\widehat{\Theta}_{i,n}\r)^{-1}},  \\
\eeas 
so that it is sufficient to prove uniformly in $i \in \{1,...,B_n\}$ the estimates
\bea 
 \Gamma_{i,n}\l(\widehat{\Theta}_{i,n}\r)^{-1}= \Gamma\l(\theta_{\blockd}^*\r)^{-1} + o_\proba(1)
\label{approxGamma1}
\eea 
and
\bea 
\Gamma\l(\theta_{\blockd}^*\r)^{-1} = \Delta_n^{-1}\int_{\blockd}^{\blockf}{\Gamma\l(\theta_t^*\r)^{-1}dt} + o_\proba(1).
\label{approxGamma2}
\eea
To show (\ref{approxGamma1}), we consider the decomposition
\beas 
\Gamma_{i,n}\l(\widehat{\Theta}_{i,n}\r)^{-1} - \Gamma\l(\theta_{\blockd}^*\r)^{-1} = \underbrace{\Gamma_{i,n}\l(\widehat{\Theta}_{i,n}\r)^{-1} - \Gamma_{i,n}\l(\theta_{\blockd}^*\r)^{-1}}_{a_{i,n}}  +\underbrace{\Gamma_{i,n}\l(\theta_{\blockd}^*\r)^{-1} - \Gamma\l(\theta_{\blockd}^*\r)^{-1}}_{b_{i,n}}.
\eeas
We have that
\bea
|a_{i,n}| &\leq& \sup_{\theta \in K} \inv{h_n} \l|\partial_\theta \l(\partial_\theta^2 l_{i,n}(\theta) \r)^{-1}\r|\l|\widehat{\Theta}_{i,n}-\theta_{\blockd}^*\r|.
\eea
By some algebraic calculus it is straightforward to show that the term $\sup_{\theta \in K} \inv{h_n} \l|\partial_\theta \l(\partial_\theta^2 l_{i,n}(\theta) \r)^{-1} \r|$ is $\mathbb{L}_p$ bounded by virtue of Lemma \ref{lemmaBoundedDerivatives} {\bf (i)} and Lemma \ref{lemmaConvInv} {\bf (i)}. By uniform consistency of $\widehat{\Theta}_{i,n}$, this yields $a_{i,n} = o_\proba(1)$. Moreover, we have that $b_{i,n} = o_\proba(1)$ as a direct consequence of Lemma \ref{lemmaConvInv} {\bf (ii)}. Thus (\ref{approxGamma1}) holds. Finally the approximation (\ref{approxGamma2}) is a straightforward consequence of Lemma  \ref{lemmaBoundedDerivatives} {\bf (i)} and Lemma \ref{lemmaConvInv} {\bf (i)} along with assumption [C]-{\bf (i)}.

%\subsection*{Proof of Lemma \ref{lemmaMartingale}}

%\subsection*{Proof of Lemma \ref{lemmaTimeChange}}

%\subsection*{Proof of Lemma \ref{lemmaConvolution}}

%\subsection*{Proof of Lemma \ref{boundedDeviation}}

%\subsection*{Proof of Lemma \ref{lemma4Inequalities}}

%\subsection*{Proof of Lemma \ref{lemmaCovariance}}

%\subsection*{Proof of Lemma \ref{lemmaConvInv}}

%\subsection*{Proof of Theorem \ref{conditionalCLT}}

%\subsection*{Proof of Lemma \ref{ordre2}}

%\subsection*{Proof of Lemma \ref{ordre3}}

%\subsection*{Proof of Theorem \ref{biasCorrection}}

\section*{Acknowledgments}
The research of Simon Clinet is in part supported by CREST Japan Science and Technology Agency. The research of Yoann Potiron is supported by the National Science Foundation 
[DMS 14-07812], Japanese Society for the Promotion of Science Grant-in-Aid for Young Scientists (60781119) and a private grant from Keio University. All financial data is provided by the Chair of Quantitative Finance of the Ecole Centrale Paris. We would like to thank Per Mykland, Nakahiro Yoshida, Frederic Abergel, Feng Chen, Holger Dette (the Editor), an anonymous Associate Editor and two anonymous referees for helpful discussions and advice.
\bibliography{biblio2}
% BibTeX users please use one of
%\bibliographystyle{spbasic}      % basic style, author-year citations
\bibliographystyle{abbrv} %{apalike}      % mathematics and physical sciences

\end{document}

%% file: nakamacro27.tex
\def\bd{\begin{description}}
\def\ed{\end{description}}

\def\D2{\bbD_{2,\infty-}}

\def\D{{\bf D}}
\def\E{{\bf E}}
\def\F{{\bf F}}
\def\G{{\bf G}}

\def\I{{\bf I}}

\def\calb{{\cal B}}

\def\calf{{\cal F}}
\def\calg{{\cal G}}
\def\calh{{\cal H}}

\def\half{\frac{1}{2}}

\def\be{\begin{equation}}
\def\ee{\end{equation}}
\def\bea{\begin{eqnarray}}
\def\eea{\end{eqnarray}}
\def\beas{\begin{eqnarray*}}
\def\eeas{\end{eqnarray*}}
\def\bi{\begin{itemize}}
\def\ei{\end{itemize}}
\def\im{\item}
\def\bd{\begin{description}}
\def\ed{\end{description}}
\def\l{\left}
\def\r{\right}
\newcommand{\bbD}{{\mathbb D}}

\newcommand{\reels}{\mathbb{R}}
\newcommand{\naturels}{\mathbb{N}}

\newcommand{\esp}{\mathbb{E}}
\newcommand{\proba}{\mathbb{P}}

\newcommand{\inv}[1]{\frac{1}{#1}}